\documentclass[11pt]{article}

\usepackage{fullpage}
\usepackage{url}
\usepackage{xspace}
\usepackage{color}
\usepackage{graphics}
\usepackage[dvips]{epsfig}

\usepackage{amsmath}
\usepackage{amssymb}
\usepackage{amsfonts}
\usepackage{graphicx}
\usepackage{algorithm}
\usepackage{comment}
\usepackage{multirow}
\usepackage{algpseudocode}

\newtheorem{theorem}{Theorem}[section]
\newtheorem{lemma}{Lemma}[section]

\newtheorem{claim}{Claim}[section]
\newtheorem{proposition}{Proposition}[section]

\newcommand{\qed}{\hfill $\Box$ \bigbreak}
\newenvironment{proof}{\noindent {\bf Proof.}}{\qed}

\newcommand{\cC}{{\cal C}}
\newcommand{\cN}{{\cal N}}
\newcommand{\cP}{{\cal P}}

\newcommand{\cH}{{\cal H}}
\newcommand{\cS}{{\cal S}}

\newcommand{\cA}{{\cal A}}

\newcommand{\cF}{{\cal F}}
\newcommand{\cG}{{\cal G}}
\newcommand{\cT}{{\cal T}}
\newcommand{\cW}{{\cal W}}

\newcommand{\cV}{{\cal{V}}}
\newcommand{\cB}{{\cal{B}}}

\newcommand{\remove}[1]{}

\begin{document}

\baselineskip  0.18in 
\parskip     0.0in  
\parindent   0.3in 

\title{{\bf Impact of Knowledge on Election Time\\ in Anonymous Networks}}
\date{}
\newcommand{\inst}[1]{$^{#1}$}

\author{
Yoann Dieudonn\'{e}\thanks{
MIS, Universit\'{e} de Picardie  Jules  Verne Amiens,  France. E-mail:  {\tt yoann.dieudonne@u-picardie.fr}. Partially supported by the research project TOREDY funded by the Picardy Regional Council and 
the European Regional Development.}
 \and Andrzej Pelc\thanks{
 D\'epartement d'informatique, Universit\'e du Qu\'ebec en Outaouais, Gatineau,
Qu\'ebec J8X 3X7, Canada. {\tt pelc@uqo.ca}. Partially supported by NSERC discovery grant 8136 -- 2013, 
and by the Research Chair in Distributed Computing at the
Universit\'e du Qu\'ebec en Outaouais.}
}

\date{ }
\maketitle

\thispagestyle{empty}

\begin{abstract}
Leader election is one of the basic problems in distributed computing. This is a symmetry breaking problem:  all nodes of a network must agree on a single node, called the leader.
If the nodes of the network have distinct labels, then such an agreement means that all nodes have to output the label of the elected leader.
For anonymous networks,
the task of leader election is formulated as follows: every node $v$ of the network must output a simple path, which is coded as a sequence of port numbers, such that
all these paths 
end at a common node, the leader. In this paper, we study deterministic leader election in arbitrary anonymous networks. 

It is well known that leader election is impossible in some networks, regardless of the allocated amount of time, even if nodes know the map of the network. This is due
to possible symmetries in it. However, even in networks in which it is possible to elect a leader knowing the map, the task may be still impossible without any knowledge,
regardless of the allocated time. On the other hand, for any network in which leader election is possible knowing the map, there is a minimum time, called
the {\em election index},
in which this can be done. Informally,  the election index of a network is the minimum depth at which views of all nodes are distinct.
Our aim is to establish tradeoffs between the allocated time $\tau$ and the amount of information that has to be given {\em a priori} to the nodes to enable leader election in time $\tau$ in all networks for which leader election in this time is at all possible.
Following the framework of {\em algorithms
with advice}, this information (a single binary string) is provided to all nodes at the start by an oracle knowing the entire network. The length of this string is called the {\em size of advice}.  
 For a given time $\tau$ allocated to leader election, we give upper and lower bounds on the minimum size
of advice sufficient to perform leader election in time $\tau$. 

We focus on the two sides of the time spectrum. For the smallest possible time, which is the election index of the network, we show that the minimum size of advice is linear in the size $n$ of the
network, up to polylogarithmic factors. On the other hand, we consider large values of time: larger than the diameter $D$ by a summand, respectively, linear, polynomial, and exponential in the election index; for these values,
we prove tight bounds on the minimum size of advice, up to multiplicative constants. We also show that
constant advice is not sufficient for leader election in all graphs, regardless of the allocated time.

\vspace{2ex}


\vspace{1ex}

\setcounter{page}{0}

\end{abstract}

\pagebreak

\section{Introduction}

{\bf Background.} 
Leader election is one of the basic problems in distributed computing \cite{Ly}. 
This is a symmetry breaking problem:  all nodes of a network must agree on a single node, called the leader.
It was first formulated in \cite{LL} in the study of local area token ring networks, where, at all times, exactly one node (the owner of a circulating token) is allowed to initiate
communication. When the token is accidentally lost, a leader must be elected as the initial owner of the token.

If the nodes of the network have distinct labels, then agreeing on a single node means that all nodes output the label of the elected leader. However, in many
applications, even if nodes have distinct identities, they may decide to refrain from revealing them, e.g., for privacy or security reasons. Hence it is important to design leader election algorithms that do not rely on knowing distinct labels of nodes, and that can work in anonymous networks as well. 
Under this scenario, agreeing on a single leader means
that every node has to output a simple path (coded as a sequence of port numbers) to a common node.

\noindent
{\bf Model and Problem Description.} The network is modeled as a simple undirected  connected $n$-node graph with diameter $D$, for $n \geq 3$.
Nodes do not have any identifiers.
On the other hand, we assume that, at each node $v$,
each edge incident to $v$ has a distinct {\em port number} from 
$\{0,\dots,d-1\}$, where $d$ is the degree of $v$. Hence, each edge has two corresponding port numbers, one at each of its endpoints. 
Port numbering is {\em local} to each node, i.e., there is no relation between
port numbers at  the two endpoints of an edge. Initially, each node knows only its own degree.
The task of leader election is formulated as follows. Every node $v$ must output a sequence $P(v)=(p_1,q_1,\dots,p_k,q_k)$ of nonnegative integers.
For each node $v$, let $P^*(v)$ be the path starting at $v$, such that port numbers $p_i$ and $q_i$ correspond to the $i$-th edge of $P^*(v)$, in the order from $v$ to the other end of this path.
All paths $P^*(v)$ must be simple paths in the graph (i.e., paths without repeated nodes) that end at a common node, called the leader. In this paper, we consider deterministic leader election algorithms.
 
In the absence of port numbers, there would be no way to identify the elected leader by non-leaders, as all
ports, and hence all neighbors, would be indistinguishable to a node.
Security and privacy reasons for not revealing node identifiers do not apply in the case of port numbers. 

The central notion in the study of anonymous networks is that of the view of a node \cite{YK3}. Let $G$ be a graph and let  $v$ be a node of $G$.  We first define, for any $l \geq 0$,  the {\em truncated view}
$\cV^l(v)$ at depth $l$, by induction on $l$. $\cV^0(v)$ is a tree consisting of a single node $x_0$. 
If $\cV^l(u)$ is defined for any node $u$ in the graph, then $\cV^{l+1}(v)$ is the port-labeled tree
rooted at $x_0$ and defined as follows.
For every node $v_i$, $i=1,\dots ,k$, adjacent to $v$, 
there is a child $x_i$ of $x_0$ in $\cV^{l+1}(v)$ such that the port number at $v$ corresponding to edge $\{v,v_i\} $ is the same as the port number 
at $x_0$ corresponding to edge $\{x_0,x_i\}$,
and the port number at $v_i$ corresponding to edge $\{v,v_i\} $ is the same as the port number at $x_i$ corresponding to edge $\{x_0,x_i\}$. 
Now node $x_i$, for $i=1,\dots ,k$, becomes 
the root of the truncated view $\cV^l(v_i)$.   

The {\em view} from $v$ is the infinite rooted tree $\cV(v)$ with labeled ports, such that $\cV^l(v)$ is its truncation to level $l$, for each $l$.

We use the extensively studied $\cal{LOCAL}$ communication model \cite{Pe}. In this model, communication proceeds in synchronous
rounds and all nodes start simultaneously. In each round, each node
can exchange arbitrary messages with all of its neighbors and perform arbitrary local computations. The information that $v$ gets about the graph in $r$ rounds
is precisely the truncated view $\cV^r(v)$, together with degrees of leaves of this tree. Denote by $\cB^r(v)$ the truncated view $\cV^r(v)$ whose leaves are labeled by their degrees in the graph, and call it the {\em augmented truncated view} at depth $r$.
If no additional knowledge is provided {\em a priori} to the nodes, the decisions of a node $v$ in round $r$ in any deterministic algorithm are a function of $\cB^r(v)$.
Note that all augmented truncated views can be canonically coded as binary strings, and hence the set of all augmented truncated views can be ordered lexicographically.
The {\em time} of leader election is the minimum number of rounds sufficient to complete it by all nodes. 
It is well known that the synchronous process of the $\cal{LOCAL}$  model can be simulated in an asynchronous network using time-stamps.

Unlike in labeled networks, if the network is anonymous then leader election is sometimes impossible, regardless of the allocated time, even if the network is a tree
and its topology is known. This is due to symmetries, and the simplest example is the two-node graph. It follows from \cite{YK3} that if nodes know the map of the graph
(i.e., its isomorphic copy with all port numbers indicated)
then leader election is possible if and only if views of all nodes are distinct. We will call such networks {\em feasible} and restrict attention to them. However, even in the class of 
feasible networks, leader election is impossible without any {\em a priori} knowledge about the network. This simple observation follows from a slightly stronger result
proved in Proposition \ref{constant}. On the other hand, for any fixed feasible network $G$, whose map is given to the nodes, there is a minimum time, called
the {\em election index} and denoted by $\phi(G)$, in which leader election can be performed. 
The election index of a network is equal to the smallest integer $\ell$, such that the augmented truncated views
at depth $\ell$ of all nodes are distinct. This will be proved formally in Section 2. The election index is always a strictly positive integer because there is no graph all of whose
nodes have different degrees.

Our aim is to establish tradeoffs between the allocated time and the amount of information that has to be given {\em a priori} to the nodes to enable them to perform
leader election.
Following the framework of {\em algorithms
with advice}, see, e.g.,   \cite{DP,EFKR,FGIP,FKL,FP,IKP,SN}, this information (a single binary string) is provided to all nodes at the start by an oracle knowing the entire network. The length of this string is called the {\em size of advice}. It should be noted that, since the advice given to all nodes is the same, this information does not increase the asymmetries 
of the network (unlike in the case when different pieces of information could be given to different nodes) but only helps to harvest the existing asymmetries and use them to elect the leader. Hence the high-level formulation of our problem is the following. What is the minimum amount of identical information that can be given to nodes to enable them to use asymmetries present in the graph to elect a leader in a given time?

Of course, since the faithful map of the network is the total information about it, asking about the minimum size of advice
to solve leader election in time $\tau$ is meaningful only in the class of networks $G$ for which $\phi(G) \leq \tau$, because otherwise, no advice can help. The central problem of this paper can be now precisely formulated as follows.
\begin{quotation}
\noindent
For a given time $\tau$, what is the minimum size of advice that permits leader election in time $\tau$, for all networks $G$ for which $\phi(G) \leq \tau$?
\end{quotation}  

The paradigm of algorithms with advice has been proven very important in the domain of network algorithms. Establishing a strong lower bound on the minimum size of advice sufficient to accomplish a given task in a given time permits to rule out
entire classes of algorithms and thus focus only on possible candidates. For example, if we prove that $\Omega(n/\log n)$ bits of advice are needed to perform a certain task in $n$-node networks (as we do in this paper for leader election in minimum possible time), this rules out all 
potential algorithms that can work using only the size $n$ of the network, 
as $n$ can be given
to the nodes using $O(\log n)$ bits.  Lower bounds on the size of advice
give us impossibility results based strictly on the \emph{amount} of initial knowledge allowed in a model.
This is much more general than the traditional approach based on
specific {categories} of information given to nodes, such as the size, diameter, or maximum node degree.

\noindent
{\bf Our results.} 
 For a given time $\tau$ allocated to leader election, we give upper and lower bounds on the minimum size
of advice sufficient to perform leader election in time $\tau$ for networks with election index at most $\alpha$. 
An upper bound $U$ for a class of networks $\cC$ means that, for all networks in $\cC$, leader election in time $\tau$ is possible given advice of size $O(U)$. We prove such a bound by constructing advice of size $O(U)$ together with a leader election algorithm for all networks in the class $\cC$, that uses this advice and works in time $\tau$.
 A lower bound $L$ for a class of networks $\cC$ means that there exist networks in $\cC$ for which leader election in time $\tau$  requires advice of size $\Omega(L)$.
 Proving such a bound means constructing a subclass $\cC'$ of $\cC$ such that no leader election algorithm running in time $\tau$ with advice of size $o(L)$
 can succeed for all networks in $\cC'$.

We focus on the two sides of the time spectrum. For the smallest possible time, which is the election index of the network, we show that the minimum size of advice is linear in the size $n$ of the
network, up to polylogarithmic factors. More precisely, we establish a general upper bound $O(n\log n)$ and lower bounds $\Omega(n\log\log n)$ and $\Omega(n(\log\log n)^2/\log n)$,
for election index equal to 1 and larger than 1, respectively.

On the other hand, we consider large values of time: those exceeding the diameter $D$ by a summand, respectively, linear, polynomial, and exponential in the election index; for these values,
we prove tight bounds on the minimum size of advice, up to multiplicative constants. More precisely, for any positive integer $\alpha$, consider the class of networks with election index at most $\alpha$.
Let $c>1$ be an integer constant. For any graph of election index $\phi \leq \alpha$ in this class, consider leader election algorithms working in time, respectively,
at most $D+ \phi +c$, at most $D+c\phi$, at most $D+\phi^c$, and at most $D+c^{\phi}$. Hence the additive offset above $D$ in the time of leader election is
asymptotically equal to $\phi$ in the first case, it is linear in $\phi$ but with a multiplicative  constant larger than 1  in the second case,
it is polynomial in $\phi$ but super-linear in the third case, and it is exponential in $\phi$ in the fourth case. In the considered class we show that  
the minimum size of advice is $\Theta(\log \alpha)$ in the first case, it is $\Theta(\log\log \alpha)$ in the second case, 
it is $\Theta(\log\log\log \alpha)$ in the third case, and it is $\Theta(\log(\log^*\alpha))$ in the fourth case. Hence, perhaps surprisingly, the jumps in the minimum
size of advice, when the time of leader election varies between the above milestones, are all exponential.  We also show that constant advice is not sufficient for leader election in all graphs, regardless of the allocated time.

\noindent
{\bf Related work.}
The first papers on leader election focused on the scenario 
where all nodes have distinct labels. Initially, it was investigated for rings in the message passing model.
A synchronous algorithm based on label comparisons was given in \cite{HS}. It used 
$O(n \log n)$ messages.  In \cite{FL} it was proved that
this complexity cannot be improved for comparison-based algorithms. On the other hand, the authors showed
a leader election algorithm using only a linear number of messages but requiring very large running time.
An asynchronous algorithm using $O(n \log n)$ messages was given, e.g., in \cite{P}, and
the optimality of this message complexity was shown in \cite{B}. Leader election was also investigated in the radio communication model,
both in the deterministic \cite{JKZ,KP,NO} and in the randomized \cite{Wil} scenarios.
In \cite{HKMMJ}, leader election for labeled networks was
studied using mobile agents.

Many authors \cite{An,AtSn,ASW,BSVCGS,BV,YK2,YK3} studied leader election
in anonymous networks. In particular, \cite{BSVCGS,YK3} characterize message-passing networks in which
leader election is feasible. In \cite{YK2}, the authors study
the problem of leader election in general networks, under the assumption that node labels exist but are
not unique. They characterize networks in which leader election can be performed and give an algorithm
which achieves election when it is feasible. 
In  \cite{DoPe,FKKLS},  the authors
study message complexity of leader election in rings with possibly
nonunique labels. 
Memory needed for leader election in unlabeled networks was studied in \cite{FP}. 
In \cite{DP1}, the authors investigated the feasibility of leader election among anonymous agents that
navigate in a network in an asynchronous way.

Providing nodes or agents with arbitrary types of knowledge that can be used to increase efficiency of solutions to network problems 
 has previously been
proposed in \cite{AKM01,DP,EFKR,FGIP,FIP1,FIP2,FKL,FP,FPR,GPPR02,IKP,KKKP02,KKP05,MP,SN,TZ05}. This approach was referred to as
{\em algorithms with advice}.  
The advice is given either to the nodes of the network or to mobile agents performing some task in a network.
In the first case, instead of advice, the term {\em informative labeling schemes} is sometimes used if (unlike in our scenario) different nodes can get different information.

Several authors studied the minimum size of advice required to solve
network problems in an efficient way. 
 In \cite{KKP05}, given a distributed representation of a solution for a problem,
the authors investigated the number of bits of communication needed to verify the legality of the represented solution.
In \cite{FIP1}, the authors compared the minimum size of advice required to
solve two information dissemination problems using a linear number of messages. 
In \cite{FKL}, it was shown that advice of constant size given to the nodes enables the distributed construction of a minimum
spanning tree in logarithmic time. 
In \cite{EFKR}, the advice paradigm was used for online problems.
In \cite{FGIP}, the authors established lower bounds on the size of advice 
needed to beat time $\Theta(\log^*n)$
for 3-coloring cycles and to achieve time $\Theta(\log^*n)$ for 3-coloring unoriented trees.  
In the case of \cite{SN}, the issue was not efficiency but feasibility: it
was shown that $\Theta(n\log n)$ is the minimum size of advice
required to perform monotone connected graph clearing.
In \cite{IKP}, the authors studied radio networks for
which it is possible to perform centralized broadcasting in constant time. They proved that constant time is achievable with
$O(n)$ bits of advice in such networks, while
$o(n)$ bits are not enough. In \cite{FPR}, the authors studied the problem of topology recognition with advice given to the nodes. 
In \cite{DP}, the task of drawing an isomorphic map by an agent in a graph was considered, and the problem was to determine the minimum advice that has to be given to the agent for the task to be feasible. 

Among papers studying the impact of information on the time of leader election, the papers \cite{FP1,GMP,MP} are closest to the present work.
In \cite{MP}, the authors investigated the minimum size of advice sufficient to find the largest-labelled node in a graph, all of whose nodes have distinct labels.
The main difference between  \cite{MP} and the present paper is that we consider networks without node labels. This is a fundamental difference:
breaking symmetry in anonymous networks relies heavily on the structure of the graph, rather than on labels, and, as far as results
are concerned, much more advice is needed for a given allocated time.
In \cite{FP1}, the authors investigated the time of leader election in anonymous networks
by characterizing this time in terms of the network size, the diameter of the network, and an additional
parameter called the level of symmetry, similar to our election index.
This paper used the traditional approach  of providing nodes with some parameters of the network, rather than any type of advice, as in our setting.
Finally, the paper \cite{GMP} studied leader election under the advice paradigm for anonymous networks, but restricted attention to trees. It should be stressed that leader election in anonymous trees and in arbitrary anonymous networks present completely different difficulties. The most striking difference is that, in the case of trees, for the relatively modest time
equal to the diameter $D$, leader election can be done in feasible trees without any advice, as all nodes can reconstruct the map of the tree. This should be contrasted with the class of arbitrary networks, in which leader election with no advice is impossible.
Our results for large election time values (exceeding the diameter $D$) give a hierarchy of sharply differing tight bounds on the size of advice in situations in which leader election in trees can be performed with no advice at all.  

\section{Preliminaries}

We use the word ``graph'' to mean a simple undirected connected graph with unlabeled nodes and all port numbers fixed. 
In the sequel we use the word ``graph'' instead of ``network''.
Unless otherwise specified, all logarithms are to the base 2.

We will use the following characterization of the election index.

\begin{proposition}\label{prop-index}
The election index of a feasible graph is equal to the smallest integer $\ell$, such that the augmented truncated views
at depth $\ell$ of all nodes are distinct. 
\end{proposition}

\begin{proof}
Fix a feasible graph $G$, let $\alpha$ be its election index, and let $\beta$ be the smallest integer $\ell$, such that the augmented truncated views
at depth $\ell$ of all nodes are distinct. 
 
Let $r$ be an integer such that augmented truncated views $\cB^r(v)$ are distinct for all nodes $v$ of the graph.
If nodes are provided with a map of the graph, then after time $r$ every node gets $\cB^r(v)$, can locate itself on the map because all  
augmented truncated views at depth $r$ are distinct, and can find the node $v_0$ for which $\cB^r(v_0)$ is lexicographically smallest.
Then every node outputs a simple path leading from it to $v_0$. Hence $\alpha \leq \beta$.

Conversely, suppose that $r$ is an integer for which there exist two nodes $v$ and $w$ with $\cB^r(v)=\cB^r(w)$. Then, after time $r$,
nodes $v$ and $w$ have identical information, and hence, when running any hypothetical leader election algorithm they must output an identical sequence
of ports. This sequence must correspond to two simple paths, one starting at $v$, the other starting at $w$, and ending at the same node. Let $x$ be the first node common in these paths, and consider the parts of these paths from $v$ to $x$ and from $w$ to $x$, respectively.  These parts must have the same length,
and hence the corresponding sequences of port numbers must be identical. In particular, the last terms in these sequences must be the same, which is impossible because they correspond to different ports at node $x$. Hence $\alpha \geq \beta$.
\end{proof}

The value of the election index is estimated in the following proposition, which is an immediate consequence of the main result of \cite{H}.

\begin{proposition}\label{H}
For any $n$-node feasible graph of diameter $D$, its election index is in $O(D\log(n/D))$.
\end{proposition}

{ Our algorithms use the subroutine $COM$ to exchange augmented truncated views at different depths with their neighbors. This subroutine is detailed in Algorithm~\ref{alg:COM}.}

\begin{algorithm}
{ \caption{$COM(i)$\label{alg:COM}}

{\bf send} $\cB^i(u)$ to all neighbors;\\
{\bf foreach} neighbor $v$ of $u$\\
\hspace*{1cm} {\bf receive} $\cB^i(v)$ from $v$
}
\end{algorithm}

When all nodes repeat this subroutine for $i=0,\dots,t-1$, every node acquires its augmented truncated view at depth $t$.

\section{Election in minimum time}

We start this section by designing a leader election algorithm working in time $\phi$, for any graph of size $n$ and  election index $\phi$, and using advice of size $O(n\log n)$.
The high-level idea of the algorithm is the following. The oracle knowing the graph $G$ produces the advice consisting of three items: the integer $\phi$, $A_1$ and $A_2$. The integer $\phi$ serves the nodes to determine for how long they have to exchange information with their neighbors. 
The item $A_1$ is the most difficult to construct. Its aim is to allow every node that knows its augmented truncated view at depth $\phi$ (which is acquired in the allocated time $\phi$) to construct a unique integer label from the set {$\{1,2,\dots,n\}$}. Recall that the advice is the same for all nodes, and hence each node has to
produce a distinct label using this common advice, relying only on its (unique) augmented truncated view at depth $\phi$. The third item in the advice, that we call $A_2$,
is a labeled BFS tree of the graph $G$. (To avoid ambiguity, we take the {\em canonical} BFS tree, in which the parent of each node $u$ at level $i+1$ is the node at level $i$ corresponding to the smallest port number at $u$.) The labels of nodes are equal to those that nodes will construct using item $A_1$, the root is the node {labeled 1}, and all port numbers in the BFS tree (that come from the graph $G$) are faithfully given. More precisely, $A_2$ is the {\em code} of this tree, i.e., a binary string of length
$O(n\log n)$ which permits the nodes to reconstruct unambiguously this labeled tree (the details are given below).  
After receiving the entire advice, Algorithm {\tt Elect} works as follows. Each node acquires its augmented truncated view at depth $\phi$, then positions itself in the obtained BFS tree, thanks to the unique constructed label, and outputs  the sequence of port numbers corresponding to the unique path from itself to the root of this BFS tree. 

The main difficulty is to produce item $A_1$ of the advice succinctly, i.e., using only $O(n\log n)$ bits, and in such a way that allows nodes to construct unique {\em short} labels. Note that a naive way in which nodes could attribute themselves distinct labels would require no advice at all and could be done as follows. Nodes could list all possible augmented truncated views at depth $\phi$, order them lexicographically in a canonical way, and then each node could adopt as its label the rank in this list. However, already for $\phi=1$, there are $\Omega(n)^{\Omega(n)}$
different possible augmented truncated views at depth 1, and hence these labels would be of size $\Omega(n\log n)$. Now item $A_2$ of the advice would have to give the tree with all these 
labels, thus potentially requiring at least $\Omega(n^2\log n)$ bits, which significantly exceeds the size of advice that we want to achieve. This is why item $A_1$ of the advice is needed, and must be constructed in a subtle way.  On the one hand, it must be sufficiently short (use only $O(n\log n)$ bits) and on the other hand it must allow nodes to construct distinct labels
of size $O(\log n)$. Then item $A_2$ of the advice can be given using also only  $O(n\log n)$ bits.

We now give some intuitions concerning the construction of item $A_1$ of the advice. This item can be viewed as a carefully constructed {\em trie}, cf.\cite{AHU},
which is a rooted binary tree whose  leaves correspond to objects, and whose internal nodes correspond to yes/no queries concerning these objects.
The left child of each internal node corresponds to port 0 and to the answer ``no'' to the query, and the right child corresponds to port 1 and to the answer ``yes'' to the query.
The object in a given leaf corresponds to all answers on the branch from the root to the leaf, and must be unique. In our case, objects in leaves of the trie are nodes
of the graph, and queries serve to discriminate all views $\cB^{\phi}(v)$, for all nodes $v$ of the graph $G$. Since each node $v$ knows its augmented truncated view
$\cB^{\phi}(v)$, after learning the trie it can position itself as a leaf of it and adopt a unique label from the set {$\{1,2\dots,  n\}$}.  

As an example, consider the case $\phi=1$. All augmented
truncated views at depth 1 can be coded by binary sequences of length $O(n \log n)$. In this case the queries at internal nodes of the trie are of two types:
``Is the binary representation of your augmented truncated view at depth one of length smaller than $t$?''  (this query is coded as $(0,t)$), and 
``Is the $j$th bit of the binary representation of your augmented truncated view at depth one equal to 1?'' (this query is coded as $(1,j)$). Since both the possible lengths
$t$ and the possible indices $j$ are of size $O(\log n)$, the  entire trie can be coded as a binary sequence of length $O(n \log n)$, because there are $n$ leaves of the trie.

For $\phi >1$ the construction is more complicated. Applying the same method as for $\phi=1$ (by building a large trie discriminating between all 
augmented truncated views at depth $\phi$, using similar questions as above, only concerning depth $\phi$ instead of depth 1) is impossible, because the sizes of the queries would exceed $\Theta(\log n)$. Actually, queries would be of size $\Omega(\phi \log n)$, resulting in advice of size $\Omega(\phi n\log n)$, and not $O(n\log n)$. Hence we apply a more subtle strategy.
The upper part of the trie is as for the case $\phi=1$. However, this is not sufficient, as in this case there exist nodes $u$ and $v$ in the graph such that
$\cB^1(u)=\cB^1(v)$ but  $\cB^{\phi}(u)\neq \cB^{\phi}(v)$, and hence such a small trie would not discriminate between all augmented truncated views at depth $\phi$. Hence leaves of this partial trie, corresponding to sets of nodes in the graph that have the same augmented truncated view at depth 1, have to be further developed, by adding sub-tries rooted at these leaves, to further discriminate between all augmented truncated views at depth $\phi$. This is
done recursively in such a way that these further queries are still of size $O(\log n)$, and constitutes the main difficulty of the advice construction.

We now proceed with the detailed description of the advice and of Algorithm {\tt Elect}  using it. We first address technical issues concerning coding various objects by binary strings.
This will be needed to define the advice formally. First we show how to encode a sequence of several binary substrings, 
corresponding to various parts of the advice, into a single string,
in a way that permits the algorithm to unambiguously decode the original sequence of substrings, and hence recover all parts of the advice.
This can be done as follows. We encode the sequence of substrings $(A_1,\ldots,A_{k})$ by doubling each digit in each substring and putting 01 between substrings. 
Denote by $Concat(A_1,\ldots,A_{k})$ this encoding and let $Decode$ be the inverse (decoding) function, i.e. $Decode(Concat(A_1,\ldots,A_{k})) = (A_1,\ldots,A_{k})$. As an example, $Concat((01),(00)) = (0011010000)$. Note that the encoding increases the total number of advice bits by a constant factor.

When constructing the advice, we will need to encode rooted trees with port numbers and labeled nodes.
The code will be a binary sequence of length $O(n\log n)$, if the tree is of size $n$ and all labels are of length $O(\log n)$. One way to produce such a code is the following. Consider the DFS walk in the tree,
starting and ending at its root, where children of any node $v$ are explored in the increasing order of  the corresponding port numbers at $v$. Let $S_1$ be the sequence of length $4(n-1)$ of all port numbers
encountered in this walk, in the order of traversing the edges of the walk, and listing the port 0 at each leaf twice in a row: when entering the leaf and when leaving it. Let $S_2$ be the sequence of 
length $n$  of node labels, in the order of visits during this walk, without repetitions. Consider the couple $(S_1,S_2)$. Using the sequence $S_1$ it is possible to reconstruct the topology of the rooted tree with all port numbers, and using the sequence $S_2$ it is possible to correctly assign labels to all nodes, starting from the root. It remains to encode the couple $(S_1,S_2)$ as a binary string.
Let $S_1=(a_1,\dots a_{4(n-1)})$, and let $S_2=(b_1,\dots, b_n)$, where $a_i$ and $b_i$ are non-negative integers. For any non-negative integer $x$, let $bin(x)$ denote its binary representation.
The couple $(S_1,S_2)$ can be unambiguously coded by the string $Concat(Concat(bin(a_1),\dots, bin(a_{4(n-1)})), Concat(bin(b_1),\dots,bin(b_n)))$. The above described code of any labeled tree
$T$ will be denoted by $bin(T)$. By the definition of $bin(T)$ we have the following proposition.

{
\begin{proposition}\label{pro:bfs}
Let $T$ be a rooted labeled $n$-node tree. If all node labels are integers in $O( n)$,  then the length of $bin(T)$ is in $O(n\log n)$. 
\end{proposition}}

Next, we define a binary code {$bin(Tr)$ of a trie $Tr$}. Since a trie can be considered as a rooted binary tree whose internal nodes are labeled by queries, the above described definition of codes for labeled trees can be adapted,
with the following modification: the strings $bin(b_i)$ at internal nodes are now binary codes of queries in the trie, instead of binary representations of integers. In our case, the queries are pairs of integers, hence their
binary codes are straightforward. In order to have labels of all nodes of the respective binary tree, we put the string $(0)$ at all leaves. This implies the following proposition.

{
\begin{proposition}\label{pro:trie}
Let $Tr$ be a trie of size $O(n)$. If the query $(a,b)$ at each internal node of $Tr$ is such that $a$ and $b$ are integers in $O(n \log n)$, then the length of $bin(Tr)$ is in $O(n \log n)$.
\end{proposition}}

We also need to encode augmented truncated views at depth 1. Consider a node $v$ of degree $k$, and call $v_j$ the neighbor of $v$ corresponding to the port $j$ at $v$. Let $a_j$ be the port at node $v_j$ corresponding to edge $\{v,v_j\}$, and let $b_j$ be the degree of $v_j$.  The augmented truncated view  $\cB^1(v)$ 
can be represented as a list $((0,a_0,b_0),\dots , (k-1,a_{k-1},b_{k-1}))$. Hence its encoding {$bin(\cB^1(v))$} is 
$Concat( Concat( bin(0),bin(a_0),bin(b_0)),\dots , Concat(bin(k-1),bin(a_{k-1}),bin(b_{k-1})))$, and we have the following proposition.

{
\begin{proposition}
\label{pro:longueur}
Let $v$ be a node of a graph of size $n$. The length of $bin(\cB^1(v))$ is in $O(n\log n)$.
\end{proposition}}

In the construction of our advice we will manipulate {\em nested lists}. These are lists of the form $L=((a_1, L_1), \dots ,(a_k, L_k))$, where each $a_i$ is a non-negative integer, and each $L_i$ is a list of the form  $((b_1, T_1), \dots ,(b_m, T_m))$, where $b_j$ are non-negative integers, and $T_j$ are tries. 
We have already defined binary codes $bin(T_j)$ of tries. The binary code of each list $L_i$ is defined as $bin(L_i)=Concat( bin(b_1),bin(T_1),\dots, bin(b_m), bin(T_m))$, and the binary
code {$bin(L)$} of the nested list $L$ is defined as $bin(L)= Concat( bin(a_1),bin(L_1),\dots, bin(a_k), bin(L_k))$. This implies the following proposition.

\begin{proposition}\label{pro:nested}
{Let $L$ be a list of couples $(a_i,L_i)$, where $a_i$ is a non-negative integer, and $L_i$ is a list of couples $(b_j,T_j)$, such that $b_j$ is a non-negative integer and $T_j$ is a trie. The length of $bin(L)$ is in $O(n\log n)$, if the following three conditions are satisfied.}

\begin{itemize}
\item {The length of $L$ is in $O(n)$, and the sum of lengths of all lists $L_i$, such that there exists a couple $(*,L_i)$ in $L$, is in $O(n)$.}

\item {The sum of sizes of all tries $T_j$, such that there exists a couple $(*,L_i)$ in $L$, where $L_i$ contains a couple $(*,T_j)$, is in $O(n)$. For each of these tries $T_j$ the query $(a,b)$ at each internal node is such that $a$ and $b$ are {integers} in $O(n)$.}

\item {For each $(a_i,L_i)$ in $L$, the integer $a_i$ is in $O(n)$, and, for each $(b_j,T_j)$ in $L_i$, the integer $b_j$ is in $O(n)$.}
\end{itemize}

\end{proposition}


We first describe the construction of the advice produced by the oracle knowing graph $G$. This construction is formulated using Algorithm ${\tt ComputeAdvice}(G)$
that will be executed by the oracle. This algorithm uses two procedures: {\tt BuildTree} and {\tt RetrieveLabel}, the latter using a subroutine {\tt LocalLabel}.
We start with the description of this subroutine.

The subroutine {\tt LocalLabel} is a recursive procedure that takes three arguments. The first is an augmented truncated view $\cB$ at some depth $d$, rooted at some node $v$, the second is a list $X$ of integers,
and the third is a trie $T$. The list $X$ is a list of temporary labels that have been previously assigned to children of $v$. The trie $T$ permits to discriminate between all views from the set $Y$ of augmented truncated
views at depth $d$ which correspond to the same augmented truncated view at depth $d-1$ as that of the root of $\cB$.  The list $X$ permits to determine
to which leaf of $T$ corresponds $\cB$. 
 {\tt LocalLabel} returns an integer label from the set {$\{1,\dots, |Y|\}$} with the following property. Consider the set $P$ 
 of nodes whose augmented truncated view at depth $d$  belongs to $Y$. The labels returned
by {\tt LocalLabel} are different for all nodes $u,v \in P$, for which $\cB^d(u)\neq \cB^d(v)$.

\pagebreak


\begin{algorithm}
{ \caption{{\tt LocalLabel}$(\cB,X,T)$\label{alg:LocalLabel}}

{\bf if} $T$ is a single node {\bf then}\\
\hspace*{1cm}{{\bf return} 1}\\
{\bf else}\\
\hspace*{1cm}$(x,y)\leftarrow$ the label of the root of $T$\\
\hspace*{1cm}$T_l \leftarrow$ the subtree of $T$ rooted at the left child of the root of $T$\\
\hspace*{1cm}$T_r \leftarrow$ the subtree of $T$ rooted at the right child of the root of $T$\\
\hspace*{1cm}$left \leftarrow false$\\
\hspace*{1cm}{\bf if} $X$ is the empty list {\bf then}\\
\hspace*{2cm}{{\bf if} ($x=0)$ and the length of the binary representation $bin(\cB)$ is smaller than $y$ {\bf then}}\\
\hspace*{3cm}$left \leftarrow true$\\
\hspace*{2cm}{{\bf if} ($x=1)$ and the $y$th bit of the binary representation $bin(\cB)$ is 0 {\bf then}}\\
\hspace*{3cm}$left \leftarrow true$\\
\hspace*{1cm}{\bf else}\\
\hspace*{2cm}{\bf if} the $(x+1)$th  term of the list $X$ is different from $y$ {\bf then}\\
\hspace*{3cm}$left \leftarrow true$\\
\hspace*{1cm}{\bf if} $left=true$ {\bf then}\\
\hspace*{2cm}{\bf return} {\tt LocalLabel}$(\cB,X,T_l)$\\
\hspace*{1cm}{\bf else}\\
\hspace*{2cm}$numleaves \leftarrow$ the number of leaves in $T_l$\\
\hspace*{2cm}{\bf return} $numleaves +${\tt LocalLabel}$(\cB,X,T_r)$

}
\end{algorithm}

In what follows, the integer returned by {\tt LocalLabel}$(\cB,X,T_r)$ will be denoted, for simplicity,  by {\tt LocalLabel}$(\cB,X,T_r)$. A similar convention
will be used for the objects returned by procedures {\tt BuildTree} and {\tt RetrieveLabel}.

The procedure {\tt RetrieveLabel} has three arguments: an augmented truncated view $\cB$ at some depth $d$, {a trie $E_1$, and  a (possibly empty) list $E_2$ of couples},
such that the first term of each couple is an integer, and the second is a list.  The procedure uses the subroutine {\tt LocalLabel} and returns a temporary integer label with the following two properties. 

1. The label is from the set {$\{1,2,\dots, |Z|\}$}, where $Z$ is the set of augmented truncated views at depth $d$ in $G$.

2. The labels returned by {\tt RetrieveLabel}$(\cB, E_1,E_2)$ and by {\tt RetrieveLabel}$(\cB', E_1,E_2)$, for any $E_1, E_2$, and any augmented truncated views  $\cB \neq \cB'$ at the same depth,
are different.

This temporary label will serve to construct queries in the trie built  by the procedure {\tt BuildTrie}, and will serve the nodes to position themselves in the BFS tree,
given by the item $A_2$ of the advice.  

Below is the pseudocode of {\tt RetrieveLabel}.

 \pagebreak
 
 \begin{algorithm}
{ \caption{{\tt RetrieveLabel}$(\cB, E_1,E_2)$\label{alg:RetrieveLabel}}

$d\leftarrow$ depth of $\cB$\\
{\bf if} $d=1$ {\bf then} {\bf return} {{\tt LocalLabel}$(\cB, (), E_1)$}\\
{\bf else}\\
\hspace*{1cm}$X \leftarrow$ empty list\\
\hspace*{1cm}{$deg \leftarrow$ degree of the root of $\cB$}\\
\hspace*{1cm}{{\bf for} $j=0$ {\bf to} $deg-1$ {\bf do}}\\
\hspace*{2cm}{$X \leftarrow X$ with {\tt RetrieveLabel}$(\cB^{d-1}(v_j), E_1,E_2)$ appended as the last term,}\\
\hspace*{3cm}{where $v_j$ is the node at depth 1 in $\cB$ corresponding to the port number $j$ at}\\ 
\hspace*{3cm}{the root of $\cB$}\\
\hspace*{1cm}{$\cB'\leftarrow$ the augmented truncated view at depth $d-1$ of the root of $\cB$\\}
\hspace*{1cm}$label \leftarrow$ {\tt RetrieveLabel}$(\cB', E_1,E_2)$\\
\hspace*{1cm}$sum \leftarrow 0$\\
\hspace*{1cm}$L \leftarrow$ the second term of the couple from $E_2$ whose first term is $d$ (* $L$ is a list*)\\
\hspace*{1cm}{{\bf for} $i=1$ {\bf to} $label$ {\bf do}}\\
\hspace*{2cm}{\bf if} the list $L$ has a term which is a couple $(i,T)$ (*$T$ is a trie*) \\
\hspace*{2cm}{\bf then}\\
\hspace*{3cm}{\bf if} $i<label$ {\bf then}\\
\hspace*{4cm}$numleaves \leftarrow$ the number of leaves in $T$\\
\hspace*{4cm}$sum \leftarrow sum +numleaves$\\
\hspace*{3cm}{\bf else}\\
\hspace*{4cm}$sum \leftarrow sum +${\tt LocalLabel}$(\cB,X,T)$\\
\hspace*{2cm}{\bf else}\\
\hspace*{3cm}{$sum \leftarrow sum + 1$}\\
\hspace*{1cm}{\bf return} $sum$

}
\end{algorithm}

The last of the three procedures is {\tt BuildTrie}. It takes three arguments.
The first argument $S$ is a non-empty {set} of distinct augmented truncated views at the same positive depth $\ell$,  $E_1$ is a (possibly empty) trie, and 
{$E_2$ is a (possibly empty)} list of couples, 
such that the first term of each couple is an integer, and the second is a list. The procedure returns a trie which permits to discriminate between all 
augmented truncated views of $S$. This is done as follows. {If $E_1$ is empty then the returned trie is constructed using the differences between the binary representations of the augmented truncated views of $S$. Actually, as we will see in the proof of Theorem~\ref{theo:smalltime}, this case occurs only when the depth of the augmented truncated views of $S$ is $1$. Otherwise, the returned trie is constructed using intermediate labels that were previously assigned to augmented truncated views at depth $l-1$: these labels can be recursively computed using the arguments $E_1$ and $E_2$. In particular,} the trie $E_1$, that is the second argument, permits to discriminate between all augmented 
truncated views at depth 1. The list $E_2$, that is the third argument, permits to further discriminate between {the 
augmented truncated views from $S$ corresponding to any given leaf} of $E_1$.


We use the following notions. {For two nodes $u$ and $v$ in the graph $G$, 
such that $\cB^{\ell}(u) \neq  \cB^{\ell}(v)$ and $\cB^{\ell-1}(u) =  \cB^{\ell-1}(v)$ with $l\geq2$}, the {\em discriminatory index} for the couple $u$ and $v$ is the smallest integer $i$, such that $\cB^{\ell-1}(u') \neq  \cB^{\ell-1}(v')$,
where $u'$ is the neighbor of $u$ corresponding to the port number $i$ at $u$, and $v'$ is the neighbor of $v$ corresponding to the port number $i$ at $v$.
{The discriminatory index of the list $S$ of augmented truncated views at depth $l\geq2$ that are all identical at depth $l-1$} (in the case when the length of $S$ is larger than 1) is the discriminatory index for nodes $u$ and $v$, such that $\cB^{\ell}(u)$ and $ \cB^{\ell}(v)$
are the augmented truncated views from $S$ with the two smallest binary representations. The {\em discriminatory subview} of $S$ is the 
augmented truncated view $\cB^{\ell-1}(u')$, where $u'$ and $v'$ are defined above and $\cB^{\ell-1}(u')$ has a lexicographically smaller binary representation than $ \cB^{\ell-1}(v')$.

\begin{algorithm}
{ \caption{{\tt BuildTrie}$(S,E_1,E_2)$\label{alg:BuildTrie}}

{\bf if}  $|S|=1$ {\bf then} {{\bf return} a single node labeled (0)}\\
{\bf else}\\
\hspace*{1cm}{\bf if} $E_1=\emptyset$ {\bf then}\\
\hspace*{2cm}{{\bf if} there exist augmented truncated views $\cB$ from $S$ with binary representations $bin(\cB)$}\\ 
\hspace*{2cm}of different lengths {\bf then}\\
\hspace*{3cm}$max \leftarrow$ the length of the longest binary representation of an element of $S$\\
\hspace*{3cm}{$S'\leftarrow$ the set of views $\cB$ from $S$ with binary representations $bin(\cB)$ of }\\
\hspace*{4cm}{lengths $<max$}\\
\hspace*{3cm}$val \leftarrow (0,max)$\\
\hspace*{2cm}{\bf else}\\
\hspace*{3cm}$j\leftarrow$ the smallest index such that the binary representations of some elements\\ 
\hspace*{4cm}of $S$ differ at the $j$th position\\
\hspace*{3cm}{$S'\leftarrow$ the set of elements $\cB$ of $S$ whose $j$th bit of the binary representation $bin(\cB)$}\\
\hspace*{4cm}{is 0}\\
\hspace*{3cm}$val \leftarrow (1,j)$\\
\hspace*{1cm}{\bf else}\\
\hspace*{2cm}$\ell \leftarrow$ the depth of augmented truncated views from $S$\\
\hspace*{2cm}$i \leftarrow$ the discriminatory index of $S$\\ 
\hspace*{2cm}$\cB_{disc}\leftarrow$ the discriminatory subview of $S$\\
\hspace*{2cm}$S'\leftarrow$ the set of elements $\cB^{\ell}(v)$ of $S$, such that  $\cB^{\ell-1}(v')\neq \cB_{disc}$,\\ 
\hspace*{3cm}where $v'$ is the neighbor of $v$ corresponding to the port number $i$ at $v$\\ 
\hspace*{2cm}$val \leftarrow (i,${\tt RetrieveLabel}$(\cB_{disc},E_1,E_2))$\\
\hspace*{1cm}{\bf return} the node labeled $val$ with the left child equal to  {\tt BuildTrie}$(S',E_1,E_2)$\\ 
\hspace*{1cm}and the right child equal to {\tt BuildTrie}$(S \setminus S',E_1,E_2)$
}
\end{algorithm}

Finally, we present Algorithm {\tt ComputeAdvice} used by the oracle to compute the advice given to nodes of the graph. Recall that $bin(x)$, for any non-negative integer $x$, is the binary representation
of $x$, $bin(T)$, for any labeled tree $T$, is the binary code of this tree, and $bin(L)$, for any nested list $L$, is the binary code of this list, as described previously. 

Conceptually, the advice consists of three items. The first item is the binary representation of the election index $\phi$. It serves the nodes to determine when to stop exchanging information
with their neighbors. The second item is $A_1=Concat (bin (E_1),bin(E_2))$, where $E_1$ and $E_2$ are as follows.
$E_1$ is a trie that permits to discriminate between all augmented truncated views at depth 1. Hence this is a labeled tree. $E_2$ is a list of couples
constructed using procedures {\tt RetrieveLabel} and {\tt BuildTrie}.  The $i$th couple of this list permits to discriminate between all augmented truncated views at depth $i+1$,
for $i<\phi$, using the previous couples and $E_1$. The couples in $E_2$ are of the form $(x,\lambda)$, where $x$ is an integer and $\lambda$ is a list of couples $(a,T_a)$,
where $a$ is a non-negative integer, and $T_a$ is a trie. Hence $E_2$ is a nested list.
Parts $E_1$ and $E_2$ together permit to discriminate between all augmented truncated views at depth $\phi$. 
The third item of the advice is $A_2$. This is the code of the canonical BFS tree of $G$ rooted at $r$, with each node $u$ labeled by 
{\tt RetrieveLabel}$(\cB^{\phi}(u),E_1,E_2)$, where $\cB^{\phi}(u)$ is the augmented truncated view in $G$. Each node will position itself in this tree, thanks to the label computed using $A_1$,
and then find the path to the root. 
The advice computed by the oracle and given to the nodes of the graph is
{\bf Adv}  $= Concat(bin(\phi),A_1,A_2)$.

\begin{algorithm}
{ \caption{${\tt ComputeAdvice}(G)$\label{alg:ComputeAdvice}}

$\phi \leftarrow$ election index of $G$\\
{$S_1 \leftarrow$} the set of all augmented truncated views at depth 1 in $G$\\
{\bf for} $i=1$ {\bf to} $\phi$ {\bf do}\\
\hspace*{1cm}{\bf if} $i=1$ {\bf then}\\  
\hspace*{2cm}$E_1 \leftarrow$ {\tt {BuildTrie$(S_1,\emptyset,())$}}\\
\hspace*{2cm}{$E_2(1) \leftarrow ()$}\\
\hspace*{1cm}{\bf else}\\
\hspace*{2cm}{$L(i)\leftarrow ()$ }\\
\hspace*{2cm} {\bf for all} augmented truncated views $\cB'$ at depth $i-1$ in $G$ {\bf do}\\ 
\hspace*{3cm}$N \leftarrow$ the set of nodes $u$ in $G$ for which $\cB^{i-1}(u)=\cB'$\\
\hspace*{3cm}$X \leftarrow$ the set of augmented truncated views $\cB^i(v)$ in $G$, for all $v\in N$\\
\hspace*{3cm}{\bf if} $|X| >1$ {\bf then}\\
\hspace*{4cm}{$j \leftarrow$ {\tt RetrieveLabel}$(\cB',E_1,E_2(i-1))$}\\
\hspace*{4cm}{$T_j \leftarrow$ {\tt BuildTrie}$(X,E_1,E_2(i-1))$}\\
\hspace*{4cm}{$L(i) \leftarrow$ $L(i)$ with the couple $(j,T_j)$ appended as the last term}\\
\hspace*{2cm}{$E_2(i) \leftarrow$ $E_2(i-1)$ with the couple $(i,L(i))$ appended as the last term}\\
{$E_2 \leftarrow E_2(\phi)$}\\
$r \leftarrow$ the node of $G$ such that {{\tt RetrieveLabel}$(\cB^{\phi}(r),E_1,E_2)=1$}\\ 
$A_1 \leftarrow Concat (bin (E_1),bin(E_2))$\\
{$T \leftarrow$ the canonical BFS tree of $G$ rooted at $r$, with each node $u$ labeled by}\\ 
{\hspace*{1cm}{\tt RetrieveLabel}$(\cB^{\phi}(u),E_1,E_2)$, where $\cB^{\phi}(u)$ is the augmented truncated view in $G$}\\
{$A_2 \leftarrow bin(T)$}\\
{\bf return} $Concat(bin(\phi),A_1,A_2)$
}
\end{algorithm}

The main algorithm {\tt Elect} uses advice {\bf Adv} $ =Concat(bin(\phi),A_1,A_2)$ given by the oracle, and is executed by a node $u$ of the graph. The node decodes the parts $\phi$, $E_1$, $E_2$, and $A_2$
of the advice from the obtained binary string {\bf Adv}. Then it acquires the augmented truncated view $\cB^{\phi}(u)$ in $\phi$ rounds. It assigns itself the unique label $x$ returned by {\tt RetrieveLabel}$(\cB^{\phi}(u), E_1,E_2)$.
It decodes the labeled tree coded by $A_2$, and positions itself in this tree, using the assigned label. Finally, it outputs the sequence of port numbers corresponding to the unique simple path in this tree
from the node with node $x$ to the root {(labeled 1)}.
Below is the pseudocode of the main algorithm.

\begin{algorithm}
{ \caption{{\tt Elect} \label{alg:Elect}}

{\bf for} $i=0$ {\bf to} $\phi -1$ {\bf do}\\
\hspace*{1cm}$COM(i)$\\
$x \leftarrow$ {\tt RetrieveLabel}$(\cB^{\phi}(u), E_1,E_2)$\\
{\bf output} the sequence of port numbers corresponding to the unique simple path in the tree coded by $A_2$,
from the node labeled $x$ to the node {labeled 1}

}
\end{algorithm}

The following theorem shows that Algorithm {\tt Elect}, executed on any  $n$-node graph, performs leader election in time equal to the election index of this graph, using advice of size $O(n\log n)$.

{\begin{theorem}\label{theo:smalltime}
For any $n$-node graph $G$ with election index $\phi$,  the following properties hold:
\begin{enumerate}
\item Algorithm {\tt ComputeAdvice(G)} terminates and returns a binary string of length $O(n\log n)$.
\item Using the advice returned by Algorithm {\tt ComputeAdvice(G)}, Algorithm {\tt Elect} performs leader election in time $\phi$.
\end{enumerate}
\end{theorem}}

\begin{proof}
In order to prove Part 1, it is enough to show that the values of variables $E_1$, $E_2$ and $T$ in the Algorithm ${\tt ComputeAdvice}(G)$ are computed in finite time, and that the length of the binary string $Concat(bin(\phi),A_1,A_2)$ (where $A_1= Concat (bin (E_1),bin(E_2))$ and $A_2=bin(T)$) is in  $O(n \log n)$. 
We first show that the length of  $bin (E_1)$ is in $O(n \log n)$.  In what follows, for any integer  $x\geq 0$,  we denote by $\mathcal{S}_x$ the set of all augmented truncated views at depth $x$ in $G$. 

We will use the following claims. 

\begin{claim}\label{claim1}
For any $S\subseteq \mathcal{S}_1$ such that $|S|\geq 1$, the procedure {\tt BuildTrie}$(S,\emptyset,())$ terminates and returns a trie of size $2|S|-1$ with exactly $S$ leaves. The leaves of this trie are labeled by $(0)$ and the internal nodes are labeled by queries of the form $(a,b)$, where $a$ and $b$ are integers in $O(n \log n)$. 
\end{claim}

\begin{claim}\label{claim2}
For any $S\subseteq \mathcal{S}_1$ such that $|S|\geq 1$, and for any  $ \cB\in S$, the procedure\\  {\tt LocalLabel}$(\cB$,$()$,{\tt BuildTrie}$(S,\emptyset,()))$ terminates and returns an integer from $\{1,2,\dots,|S|\}$. Moreover, for any $ \cB'\in S$ such that $\cB\ne\cB'$, the integers {\tt LocalLabel}$(\cB$,$()$,{\tt BuildTrie}$(S,\emptyset,()))$ and {\tt LocalLabel}$(\cB'$,$()$,{\tt BuildTrie}$(S,\emptyset,()))$ are different.
\end{claim}

The two claims are proved by simultaneous induction on the size of $S$. We first prove them for $|S|=1$. In this case {\tt BuildTrie}$(S,\emptyset,())$ consists of a single node with label $(0)$, and the integer {\tt LocalLabel}($\cB$,$()$,{\tt BuildTrie}$(S,\emptyset,()))$ is 1, where $\cB$ is the unique element of $S$. Hence, both claims hold,
if $|\mathcal{S}_1|=1$. Consider the case $|\mathcal{S}_1|\geq 2$, and suppose, by the inductive hypothesis, that, for some integer $k\in \{1,2,\dots,|\mathcal{S}_1|-1\}$, both claims hold when $1\leq|S|\leq k$. We prove that they hold for $|S|=k+1$. We have $|S|\geq 2$ and there are two cases. The first case is when there exist
$\cB_1,\cB_2 \in S$, such that the lengths of $bin(\cB_1)$ and $bin(\cB_2)$ are different. The second case is when all codes $bin(\cB)$, for $\cB \in S$, have equal length, but there exist $\cB_1,\cB_2 \in S$, and an index $j$, such that the $j$th bits of $bin(\cB_1)$ and $bin(\cB_2)$ are different.We consider only the first case, as the second one can be proved similarly. Denote by $max$ the largest among lengths of codes $bin(\cB)$, for $\cB \in S$, and denote by $S'$ the subset of $S$ consisting of those augmented truncated views $\cB \in S$, for which $bin(\cB)$ has length smaller than $max$.

By the inductive hypothesis, procedure {\tt BuildTrie}$(S,\emptyset,())$ terminates and  returns a binary tree with the root labeled $(0,max)$,  whose left child is 
the root of {\tt BuildTrie}$(S',\emptyset,())$ and whose right child is the root of {\tt BuildTrie}$(S\setminus S',\emptyset,())$. Since codes of elements of $S$ are not all of the same length, we have $1\leq |S'|< |S|=k+1$. Hence, the number of nodes and the number of leaves of {\tt BuildTrie}$(S',\emptyset,())$ are, respectively, $2|S'|-1$ and $|S'|$. Also,  the number of nodes and the number of leaves of {\tt BuildTrie}$(S\setminus S',\emptyset,())$ are, respectively, $2|S\setminus S'|-1$ and $|S\setminus S'|$.
Hence,  the number of nodes and the number of leaves of {\tt BuildTrie}$(S,\emptyset,())$ are, respectively, $2|S|-1$ and $|S|$. Moreover, by the inductive hypothesis, 
all leaves of {\tt BuildTrie}$(S',\emptyset,())$ and of {\tt BuildTrie}$(S\setminus S',\emptyset,())$ have label $(0)$, while internal nodes of these tries are labeled by couples $(a,b)$, where $a$ and $b$ are integers  in $O(n \log n)$. Since the integer $max$ used in the label $(0,max)$ of the root of {\tt BuildTrie}$(S,\emptyset,())$,
is in $O(n \log n)$, by Proposition~\ref{pro:longueur}, this implies that  Claim \ref{claim1} holds for $|S|=k+1$.

Concerning Claim \ref{claim2}, consider distinct views $\cB$ and $\cB'$ from $S$, and suppose, without loss of generality, that the length of $bin(\cB)$ is not smaller than the length of $bin(\cB')$.
If the length of $bin(\cB')$ is equal to $max$, then, by the inductive hypothesis {and by Algorithms~\ref{alg:LocalLabel} and~\ref{alg:BuildTrie}, we have}\\ {\tt LocalLabel}($\cB'$,$()$,{\tt BuildTrie}$(S,\emptyset,()))=$ $|S'|+$ {\tt LocalLabel}($\cB'$,$()$,{\tt BuildTrie}$(S\setminus S',\emptyset,()))$ and \\{\tt LocalLabel}($\cB$,$()$,{\tt BuildTrie}$(S,\emptyset,()))=$ $|S'|+$ {\tt LocalLabel}($\cB$,$()$,{\tt BuildTrie}$(S\setminus S',\emptyset,()))$. Again in view of the inductive hypothesis, the integers {\tt LocalLabel}($\cB$,$()$,{\tt BuildTrie}$(S\setminus S',\emptyset,()))$ and {\tt LocalLabel}($\cB'$,$()$,{\tt BuildTrie}$(S\setminus S',\emptyset,()))$
both belong to the set $\{1,\dots,|S\setminus S'|\}$ and are different. Hence, if the length of $bin(\cB')$ is equal to $max$, then Claim~\ref{claim2} holds for $|S|=k+1$. It remains to consider the case when
the length of $bin(\cB')$  is smaller than $max$. In this case we have {\tt LocalLabel}($\cB'$,$()$,{\tt BuildTrie}$(S,\emptyset,()))=$ {\tt LocalLabel}($\cB'$,$()$,{\tt BuildTrie}$(S',\emptyset,()))$, and this integer belongs to the set $\{1,\dots,|S'|\}$. If the length of $bin(\cB)$ is $max$, a similar argument as above implies that {\tt LocalLabel}($\cB$,$()$,{\tt BuildTrie}$(S,\emptyset,()))$ belongs to the set $\{|S'|+1,\dots,|S'|+|S\setminus S'|\}$, and Claim \ref{claim2} holds for $|S|=k+1$. Otherwise, we have {\tt LocalLabel}($\cB$,$()$,{\tt BuildTrie}$(S,\emptyset,()))=$ {\tt LocalLabel}($\cB$,$()$,{\tt BuildTrie}$(S',\emptyset,()))$.  However, by the inductive hypothesis {and by Algorithms~\ref{alg:LocalLabel} and~\ref{alg:BuildTrie}}, {\tt LocalLabel}($\cB$,$()$,{\tt BuildTrie}$(S',\emptyset,()))$ belongs to the set $\{1,\dots,|S'|\}$ but is different from {\tt LocalLabel}($\cB'$,$()$,{\tt BuildTrie}$(S',\emptyset,()))$. Hence Claim \ref{claim2} holds for $|S|=k+1$ in this case as well. 

We proved that Claims \ref{claim1} and~\ref{claim2} hold for $|S|=k+1$, which completes the proof of these claims by induction.

In Algorithm {\tt ComputeAdvice} we have $E_1=$ {\tt BuildTrie}$(S_1,\emptyset,())$. Hence, Claim \ref{claim1} and Proposition~\ref{pro:trie} imply the following claim.

\begin{claim}\label{claim3}
The computations of variables
$E_1$ and $E_2(1)$ in Algorithm {\tt ComputeAdvice}$(G)$ terminate, $E_2(1)$ is the empty list, and the length of $bin(E_1)$ is in $O(n \log n)$.
\end{claim}

In view of Algorithm \ref{alg:RetrieveLabel},  we have
 {\tt RetrieveLabel}$(\cB, E_1,())=$  {\tt LocalLabel}($\cB$,$()$,$E_1$), for every $\cB$ in $\mathcal{S}_1$.  Claim~\ref{claim2} and the equality $E_1=$ {\tt BuildTrie}$(\mathcal{S}_1,\emptyset,())$ imply the following claim.
 
\begin{claim}\label{claim4}
For every $\cB\in \mathcal{S}_1$, the procedure {\tt RetrieveLabel}$(\cB, E_1,())$ terminates and returns a value belonging to $\{1,\dots,|\mathcal{S}_1|\}$. For all $\cB' \neq \cB$ from $\mathcal{S}_1$, we have {\tt RetrieveLabel}$(\cB, E_1,())$ $\ne$ {\tt RetrieveLabel}$(\cB', E_1,())$.
\end{claim}

Notice that, if the election index of the graph $G$ is 1, then in Algorithm {\tt ComputeAdvice}$(G)$ we have $E_2=E_2(1)=()$, in view of Claim~\ref{claim3}. Hence, in view
of Claim~\ref{claim4}, the computation of the labeled BFS tree $T$ terminates in this algorithm, and labels of nodes of $T$ belong to $O(n)$, since $|\mathcal{S}_1|\leq n$. Hence,  Proposition~\ref{pro:bfs} and  Claim~\ref{claim3} imply that, if $\phi=1$, then computations of variables $E_1$, $E_2$ and $T$ in Algorithm ${\tt ComputeAdvice}(G)$ terminate, and the length of the returned binary string $Concat(bin(\phi),A_1,A_2)$ (where $A_1= Concat (bin (E_1),bin(E_2))$ and $A_2=bin(T)$)
belongs to $O(n \log n)$. This proves Part 1 of our theorem for $\phi=1$. We continue the proof of this part assuming that $\phi\geq 2$.
In the rest of this proof, for all integers $i\in \{2,\dots,\phi\}$, and for all integers $j$, we denote by $\mathcal{S}_i(j)$ the set of all augmented truncated views at depth $i$ of all nodes $u$ of $G$, such that {\tt RetrieveLabel}$(\cB^{i-1}(u), E_1,E_2(i-1))=j$. 

We will use the following three claims.

\begin{claim}\label{claim5}
For every integer $i\in \{2,\dots,\phi\}$, Algorithm {\tt ComputeAdvice}$(G)$ terminates the computation of variable $E_2(i)$, and assigns to it
the value $((2,L(2)),(3,L(3)),\dots,(i,L(i)))$, such that for all $k\in \{2,\dots,i\}$ the following four properties hold:
\begin{itemize}
\item {\bf Property~1.} $L(k)$ is a list of distinct couples $(j,T_j)$, such that $j$ is a non-negative integer and $T_j$ is a trie.
\item {\bf Property~2.} For any couples $(j,T_j)$ and $(j',T_{j'})$ of $L(k)$, such that $(j,T_j)$ and $(j',T_{j'})$ are not at the same position in $L(k)$, we have $j\ne j'$.
\item {\bf Property~3.} For every couple $(j,T_j)$ of $L(k)$, we have $T_j=$ {\tt BuildTrie}$(\mathcal{S}_k(j),E_1,E_2(k-1))$.
\item {\bf Property~4.} There is a couple $(j,T_j)$ in $L(k)$ if and only if {$|\mathcal{S}_k(j)|\geq 2$}.
\end{itemize}
\end{claim}

\begin{claim}\label{claim6}
For every integer $i\in \{2,\dots,\phi\}$ and for every integer $j$,  the following  properties hold:
\begin{itemize}
\item {\bf Property~1.} For every $S\subseteq \mathcal{S}_i(j)$, such that $|S|\geq 1$, the procedure {\tt BuildTrie}$(S,E_1,E_2(i-1))$ terminates and returns a trie of size $2|S|-1$ with exactly $S$ leaves. The leaves of this trie are labeled by $(0)$ and its internal nodes are labeled by queries of the form $(a,b)$ where $a$ and $b$ are {integers} in $O(n)$.
\item {\bf Property~2.} For every $S\subseteq \mathcal{S}_i(j)$ such that $|S|\geq 1$, and for every node $u$ of $G$ such that $\cB^i(u)\in S$,  let $X_u$ be the list\\ $(${\tt RetrieveLabel}$(\cB^{i-1}(u_0, E_1,E_2(i-1))$,$\dots$,{\tt RetrieveLabel}$(\cB^{i-1}(u_{deg(u)-1}, E_1,E_2(i-1))$$)$, where $deg(u)$ is the degree of $u$ and, for all $0\leq l \leq deg(u)-1$, $u_l$ is the neighbor of $u$ corresponding to port $l$ at $u$.
\begin{itemize}
\item {\bf Property~2.1} The  list $X_u$ is not empty.
\item {\bf Property~2.2} The procedure {\tt LocalLabel}($\cB^i(u)$,$X_u$,{\tt BuildTrie}$(S,E_1,E_2(i-1)))$ terminates and returns an integer belonging to $\{1,\dots,|S|\}$.
\item {\bf Property~2.3} For every node $u'$, such that $\cB^i(u')\in S$ and $\cB^i(u)\ne\cB^i(u')$, we have {\tt LocalLabel}($\cB^i(u)$,$X_u$,{\tt BuildTrie}$(S,E_1,E_2(i-1)))$\\ $\ne$ 
{\tt LocalLabel}($\cB^i(u')$,$X_{u'}$,{\tt BuildTrie}$(S,E_1,E_2(i-1)))$.
\end{itemize}
\end{itemize}
\end{claim}

\begin{claim}\label{claim7}
For every integer $i\in \{2,\dots,\phi\}$, for every integer $k\in \{1,\dots,i\}$, and for every  $\cB \in \mathcal{S}_k$, the following properties hold:
\begin{itemize}
\item {\bf Property~1.} The procedure {\tt RetrieveLabel}$(\cB, E_1,E_2(i))$ terminates and returns an integer belonging to $\{1,\dots,|\mathcal{S}_k|\}$.
\item {\bf Property~2.} For all augmented truncated views $\cB' \neq \cB$ from $\mathcal{S}_k$, we have\\ {\tt RetrieveLabel}$(\cB, E_1,E_2(i))$ $\ne$ {\tt RetrieveLabel}$(\cB', E_1,E_2(i))$.
\end{itemize} 
\end{claim}

Claims~\ref{claim5},~\ref{claim6} and~\ref{claim7} are proved by simultaneous induction on $i$. We first prove them for $i=2$. In this case we have $E_2(i-1)=E_2(1)=()$, in view of Claim~\ref{claim3}. First consider Claim \ref{claim6}. Fix an integer $j$. If $\mathcal{S}_2(j)=\emptyset$, then both properties of the claim are immediately verified. Hence suppose that $\mathcal{S}_2(j)\ne\emptyset$. We show by induction on the size of $S$ that both these properties are satisfied. 
If $|S|=1$, then procedure {\tt BuildTrie}$(S,E_1,E_2(1))$ returns a single node labeled $(0)$, and hence property 1 is satisfied in this case. As for property 2, 
for every node $v$ in $G$,  such that $\cB^2(v)\in S$, the list $X_v$ is non-empty by Claim~\ref{claim4}. Since {\tt BuildTrie}$(S,E_1,E_2(1))$ returns a single node, 
procedure {\tt LocalLabel}($\cB^2(v)$,$X_v$,{\tt BuildTrie}$(S,E_1,E_2(1)))$ terminates and returns integer 1. Hence properties 2.1 and 2.2 hold in this case.
Moreover, property 2.3 holds as well, since $S$ is a singleton in this case. It follows that  property 2 of Claim~\ref{claim6} holds for $|S|=1$. 

It follows that
Claim~\ref{claim6} holds, if $|\mathcal{S}_2(j)|=1$. Suppose that $|\mathcal{S}_2(j)|\geq 2$, and assume, by inductive hypothesis on the size of $S$, that properties 1 and 2 of Claim~\ref{claim6} hold, when $1\leq |S|\leq k$, for some integer $1\leq k\leq |\mathcal{S}_2(j)|-1$. We prove that these properties hold if $|S|=k+1$.
 If $|S|=k+1$ then $|S|\geq 2$. Let $p$ and $\cB_{disc}$ be, respectively, the discriminatory index of $S$ and the discriminatory subview of $S$. These objects are well defined because, by definition, all augmented truncated views {at depth 1 of the roots of the augmented truncated views} from $\mathcal{S}_2(j)$ are identical. Notice that the depth of $\cB_{disc}$ is $1$ because the depth of all augmented truncated views in $\mathcal{S}_2(j)$ is 2. Denote by $S'$ the set of augmented truncated views $\cB^2(v)$ from $S$, such that $\cB^1(w)\ne \cB_{disc}$, where $w$ is the neighbor of $v$ corresponding to the port number $p$ at $v$. {In view of} the inductive hypothesis and of Claim~\ref{claim4}, the procedure {\tt BuildTrie}$(S,E_1,E_2(1)))$ terminates and returns a binary tree with root labeled by  $(p$,{\tt RetrieveLabel}$(\cB_{disc}, E_1,())$$)$, whose left child is the root of
 {\tt BuildTrie}$(S',E_1,())$,  and whose right child is the root of {\tt BuildTrie}$(S\setminus S',E_1,())$. Hence property 1 of Claim~\ref{claim6} holds for $|S|=k+1$:
 in particular, $p$ is {an integer} smaller than the degree of $v$, 
 {\tt RetrieveLabel}$(\cB_{disc}, E_1,())$ is in $\{1,\dots,\mathcal{S}_1\}$, and  hence $p$ and {\tt RetrieveLabel}$(\cB_{disc}, E_1,())$ are {integers} in $O(n)$.
 
 Concerning property 2, notice that, for every node $v$ from $G$, such that $\cB^2(v)\in S$, the list $X_v$ is non-empty by Claim~\ref{claim4}, and hence property 2.1
 holds for $|S|=k+1$. We now prove properties 2.2 et 2.3. Consider any nodes $v$ and $v'$ of $G$, such that $\cB^2(v)$ and $\cB^2(v')$ are in $S$ and $\cB^2(v)\ne \cB^2(v')$. First suppose that $\cB^1(v_p)\ne \cB_{disc}$ (i.e., $\cB^2(v)\in S'$), where $v_p$ is the neighbor of $v$ corresponding to the port number $p$ at $v$.
 In this case, in view of Claim \ref{claim4}, the $(p+1)$th term of the list $X_v$ is different from {\tt RetrieveLabel}$(\cB_{disc}, E_1,())$, and, in view of Algorithms \ref{alg:LocalLabel} and~\ref{alg:BuildTrie}, we have
 {\tt LocalLabel}($\cB^2(v)$,$X_v$,{\tt BuildTrie}$(S,E_1,E_2(1)))$ $=${\tt LocalLabel}($\cB^2(v)$,$X_v$,{\tt BuildTrie}$(S',E_1,E_2(1)))$. We have $1\leq |S'| <k+1$.
 By the inductive hypothesis and by property 2.1, proved above for $|S|=k+1$, we know that  procedure {\tt LocalLabel}($\cB^2(v)$,$X_v$,{\tt BuildTrie}$(S,E_1,E_2(1)))$ terminates and returns an integer from $\{1,\dots,|S'|\}$. Since $|S'|<|S|$,  property 2.2 is proved if $\cB^2(v)\in S'$. Moreover, if $\cB^2(v')\in S'$, then by Claim~\ref{claim4} and {Algorithms~\ref{alg:LocalLabel} and~\ref{alg:BuildTrie}}, the $(p+1)$th term of the list $X_{v'}$ is different from {\tt RetrieveLabel}$(\cB_{disc}, E_1,())$, and by property 2.1, proved above for $|S|=k+1$, procedure
 {\tt LocalLabel}($\cB^2(v')$,$X_{v'}$,{\tt BuildTrie}$(S,E_1,E_2(1)))$ terminates and returns the same integer as {\tt LocalLabel}($\cB^2(v')$,$X_{v'}$,{\tt BuildTrie}$(S',E_1,E_2(1)))$. By the inductive hypothesis, we have\\  {\tt LocalLabel}($\cB^2(v)$,$X_{v}$,{\tt BuildTrie}$(S',E_1,E_2(1)))$ $\ne$ {\tt LocalLabel}($\cB^2(v')$,$X_{v'}$,{\tt BuildTrie}$(S',E_1,E_2(1)))$. This implies\\ {\tt LocalLabel}($\cB^2(v)$,$X_{v}$,{\tt BuildTrie}$(S,E_1,E_2(1)))$ $\ne$ {\tt LocalLabel}($\cB^2(v')$,$X_{v'}$,{\tt BuildTrie}$(S,E_1,E_2(1)))$.
 
Moreover, if $\cB^2(v')\notin S'$, then the $(p+1)$th term of the list $X_{v'}$ is equal to the integer {\tt RetrieveLabel}$(\cB_{disc}, E_1,())$, and hence {\tt LocalLabel}($\cB^2(v')$,$X_{v'}$,{\tt BuildTrie}$(S,E_1,E_2(1)))$ $=\\$ numleaves $+$ {\tt LocalLabel}($\cB^2(v')$,$X_{v'}$,{\tt BuildTrie}$(S\setminus S',E_1,E_2(1)))$ {in view of Algorithms~\ref{alg:LocalLabel} and~\ref{alg:BuildTrie}}, where numleaves is the number of leaves in {\tt BuildTrie}$(S',E_1,E_2(1))$. Since $1\leq |S'| \leq k$ and $1\leq |S\setminus S'| \leq k$, by the inductive hypothesis and
property 2.1, proved for $|S|=k+1$,\\ {\tt LocalLabel}($\cB^2(v')$,$X_{v'}$,{\tt BuildTrie}$(S,E_1,E_2(1)))$ terminates and returns an integer from $\{|S'|+1,\dots,|S'|+|S\setminus S'|\}$. It follows that\\ {\tt LocalLabel}($\cB^2(v')$,$X_{v'}$,{\tt BuildTrie}$(S,E_1,E_2(1)))$ $>$ {\tt LocalLabel}($\cB^2(v)$,$X_{v}$,{\tt BuildTrie}$(S,E_1,E_2(1)))$. Hence, if $\cB^1(v_p)\ne \cB_{disc}$ (i.e., $\cB^2(v)\in S'$) the properties  2.2 et 2.3 hold when $|S|=k+1$. If $\cB^1(v_p)= \cB_{disc}$, a similar reasoning shows that they hold as well. By induction, we deduce that properties 1 and 2 of Claim~\ref{claim6} hold for every non-empty set $S\subseteq \mathcal{S}_2(j)$,
which finishes the proof of Claim~\ref{claim6} for $i=2$.

We now prove that Claim~\ref{claim5} holds for $i=2$. {According to Algorithm~\ref{alg:ComputeAdvice},} $L(2)$ is a (possibly empty) list of couples  $(L_a,L_b)$, such that $L_a$ is the integer returned by {\tt RetrieveLabel}$(\cB', E_1,())$, for some augmented truncated view $\cB'$ at depth 1,  and $L_b$ is the trie {\tt BuildTrie}$(X,E_1,()))$, where $X$ is a non-empty set of augmented truncated views at depth $1$. By Claim~\ref{claim4} and Claim~\ref{claim6} for $i=2$ (proved above) the computation of $L(2)$ terminates, the variable
$E_2(2)$ is assigned the value $(2,L(2))$, $L_a$ is a non-negative integer and $L_b$ is a trie. Hence property 1 of Claim~\ref{claim5} holds for $i=2$. Concerning property 2 of this claim, notice that, if $L(2)$ is empty or has a unique term, then this property holds for $i=2$. If $L(2)$ has at least two terms $(L_a,L_b)$ et $(L'_a,L'_b)$,
we have $L_a=$ {\tt RetrieveLabel}$(\cB', E_1,())$ and $L'_a=$ {\tt RetrieveLabel}$(\cB'', E_1,())$,  where $\cB'$ and $\cB''$ are two different augmented  truncated views at depth $1$. By Claim~\ref{claim4}, we have {\tt RetrieveLabel}$(\cB', E_1,())$ $\ne$ {\tt RetrieveLabel}$(\cB'', E_1,())$. This implies property 2 for $i=2$ because $L_a\ne L'_a$. As for properties 3 and 4, consider a couple $(L_a,L_b)$ from $L(2)$. We have $L_a=$ {\tt RetrieveLabel}$(\cB', E_1,())$ for some augmented truncated view $\cB'$ at depth $1$, and $L_b=$ {\tt BuildTrie}$(X,E_1,()))$, where $X$ is the set of augmented truncated views at depth 2 of all nodes $u$ of $G$ such that $\cB^1(u)=\cB'$. Hence $X=\mathcal{S}_2(L_a)$, and property  3 holds for $i=2$. Moreover, {according to Algorithm~\ref{alg:ComputeAdvice}}, $|X|\geq 2$. Hence, if there exists a couple $(L_a,L_b)$ in $L(2)$, then
 $|\mathcal{S}_2(L_a)|\geq2$. In order to prove property 4 for $i=2$, we have to show that the converse implication holds as well. For each augmented truncated view
 $\cB'$ at depth $1$, consider the set $Y$ of augmented truncated views $\cB^2(u)$ of all nodes $u$ such that $\cB^1(u)=\cB'$. If $Y$ is of cardinality greater than 1, then there is a couple $(${\tt RetrieveLabel}$(\cB', E_1,())$,{\tt BuildTrie}$(Y,E_1,())))$ in $L(2)$. Hence, in view of Claim~\ref{claim4}, this implies the converse implication of property 4, which proves this property  for $i=2$. This completes the proof of Claim~\ref{claim5} for $i=2$.
 
 We next prove Claim~\ref{claim7} for $i=2$. We have to show its validity for $k=1$ and for $k=2$. If $k=1$ then, for every augmented truncated view $\cB\in \mathcal{S}_1$, we have {\tt RetrieveLabel}$(\cB, E_1,E_2(2))=$
{\tt RetrieveLabel}$(\cB, E_1,E_2(1))$, as the depth of $\cB$ is $1$. Hence, in view of Claim~\ref{claim4}, Claim~\ref{claim7} holds for $k=1$. 

Suppose that $k=2$. Let $u$ be any node of $G$. In view of Claim ~\ref{claim4} and of Claims ~\ref{claim5} and~\ref{claim6} shown above for $i=2$, the procedure
{\tt RetrieveLabel}$(\cB^2(u), E_1,E_2(2))$ terminates. Hence, to complete the proof of Claim~\ref{claim7} for $i=2$, it is enough to prove the following facts:

{\bf Fact~1.} {\tt RetrieveLabel}$(\cB^2(u), E_1,E_2(2))$ $\in\{1,\dots,|\mathcal{S}_2|\}$.

{\bf Fact~2.} For every node $v$ of $G$, such that $\cB^2(u)\ne\cB^2(v)$, we have

 {\tt RetrieveLabel}$(\cB^2(u), E_1,E_2(2))$ $\ne$ {\tt RetrieveLabel}$(\cB^2(v), E_1,E_2(2))$.
 
 We first prove Fact 1. The value of the variable $label$ in {\tt RetrieveLabel}$(\cB^2(u), E_1,E_2(2))$ is equal to the integer returned by
 {\tt RetrieveLabel}$(\cB^1(u), E_1,E_2(2))$ which is equal to the integer returned by {\tt RetrieveLabel}$(\cB^1(u), E_1,E_2(1))$ because the depth of $\cB^1(u)$ is $1$.
In view of Claim~\ref{claim4}, we have to consider two cases: when $label=1$ and when $label>1$.

\begin{itemize}
\item{$label=1$}.
If there is no couple $(1,*)$ in the list $L(2)$, then {\tt RetrieveLabel}$(\cB^2(u), E_1,E_2(2))$ returns the integer 1. If there is a couple $(1,T_1)$ in the list $L(2)$, then
Claim~\ref{claim5} for $i=2$ implies that this is the only couple in this list whose first term is 1, and $T_1$ is the trie equal to {\tt BuildTrie}$(\mathcal{S}_2(1),E_1,E_2(1)))$. We have {\tt RetrieveLabel}$(\cB^2(u), E_1,E_2(2))$ $=$ {\tt LocalLabel}($\cB^2(u)$,$X_{u},T_1)$, and this integer belongs to $\{1,\dots,|\mathcal{S}_2(1)|\}$, in view of property 2.2 of Claim~\ref{claim6}. Since $|\mathcal{S}_2(1)|\leq |\mathcal{S}_2|$, the proof of Fact 1 is completed in this case.

\item{$label> 1$}.
Denote by $\mathcal{K}$ the set of indices $l\leq label$, such that there exists a  couple $(l,T_l)$ in the list $L(2)$. Claims~\ref{claim5} and~\ref{claim6} for $i=2$
imply that, for every $l\in\mathcal{K}$, there exists a unique couple in $L(2)$ whose first term is $l$. Also, $T_l=$ {\tt BuildTrie}$(\mathcal{S}_2(l),E_1,E_2(1))))$, 
and the number of leaves of $T_l$ is $|\mathcal{S}_2(l)|$. Hence, if $label\notin\mathcal{K}$, we have {\tt RetrieveLabel}$(\cB^2(u), E_1,E_2(2))$ $=$ $\sum_{z\in\mathcal{K}}|\mathcal{S}_2(z)|+$ $|\{1,\dots,label\}\setminus\mathcal{K}|$, which is at most 
$\sum_{z\in\{1,\dots,label\}}|\mathcal{S}_2(z)|\leq |\mathcal{S}_2|$, because, for every pair of distinct integers  $z_1$ and $z_2$ from $\{1,\dots,label\}$,
we have $\mathcal{S}_2(z_1)\cap\mathcal{S}_2(z_2)=\emptyset$, in view of Claim~\ref{claim4}. On the other hand, if $label\in\mathcal{K}$, then {\tt RetrieveLabel}$(\cB^2(u), E_1,E_2(2))$ $=$ $\sum_{z\in\mathcal{K}\setminus\{label\}}|\mathcal{S}_2(z)|+$ $|\{1,\dots,label\}\setminus\mathcal{K}|$ $+$ {\tt LocalLabel}($\cB^2(u)$,$X_{u}$,$T_{label})$\\ (where $T_{label}=${\tt BuildTrie}$(\mathcal{S}_2(label),E_1,E_2(1))$). Hence
{\tt RetrieveLabel}$(\cB^2(u), E_1,E_2(2))$ $\leq$ $\sum_{z\in\{1,2,\dots,label-1\}}|\mathcal{S}_2(z)|+$ {\tt LocalLabel}($\cB^2(u)$,$X_{u}$,$T_{label})$. 
The latter integer is not larger than $\sum_{z\in\{1,2,\dots,label\}}|\mathcal{S}_2(z)|$ because {\tt LocalLabel}($\cB^2(u)$,$X_{u}$,$T_{label})$ belongs to $\{1,\dots,|\mathcal{S}_2(label)|\}$, by property 2.2 of Claim~\ref{claim6} for $i=2$. As mentioned above, $\sum_{z\in\{1,\dots,label\}}|\mathcal{S}_2(z)|\leq |\mathcal{S}_2|$. This completes
the proof of Fact 1 in the case $label> 1$.
\end{itemize}

We now prove Fact 2. Denote by $local(u)$ (resp. $local(v)$) the variable $local$ in procedure {\tt RetrieveLabel}$(\cB^2(u), E_1,E_2(2))$ (resp. in {\tt RetrieveLabel}$(\cB^2(v), E_1,E_2(2))$), and first consider the case $local(u)=local(v)$. Then $\cB^2(u)$ and $\cB^2(v)$ belong to the set $\mathcal{S}_2(local(v))$, and $|\mathcal{S}_2(local(v))|\geq 2$.
Moreover, in view of Claim~\ref{claim5} for $i=2$, there exists a unique couple $(label(v),T_{label(v)})$ in $L(2)$, where $T_{label(v)}=$ {\tt BuildTrie}$(\mathcal{S}_2(label(v)),E_1,E_2(1)))$. For some non-negative integer $W$,  we have {\tt RetrieveLabel}$(\cB^2(u), E_1,E_2(2))=$ $W+$ {\tt LocalLabel}($\cB^2(u)$,$X_{u}$,$T_{label(v)})$, and\\
 {\tt RetrieveLabel}$(\cB^2(v), E_1,E_2(2))=$ $W+$ {\tt LocalLabel}($\cB^2(v)$,$X_{v}$,$T_{label(v)})$. In view of Claim~\ref{claim6} for $i=2$, we have
 {\tt LocalLabel}($\cB^2(u)$,$X_{u}$,$T_{label(v)})\ne$ {\tt LocalLabel}($\cB^2(v)$,$X_{v}$,$T_{label(v)})$. It follows that {\tt RetrieveLabel}$(\cB^2(u), E_1,E_2(2))\ne$ {\tt RetrieveLabel}$(\cB^2(v), E_1,E_2(2))$, which proves Fact 2 when $local(u)=local(v)$.
 
 Now suppose that $local(u)\ne local(v)$. Without loss of generality,  $local(u)<local(v)$. There are 3 cases.

\begin{itemize}
\item{\bf Case~1.} There is no couple in $L(2)$ whose first term is $local(u)$ or $local(v)$.

Then {\tt RetrieveLabel}$(\cB^2(u), E_1,E_2(2))\leq$ {\tt RetrieveLabel}$(\cB^2(v), E_1,E_2(2))+1$.

\item{\bf Case~2.} There is a couple $(local(u),T_{local(u)})$ and a couple $(local(v),T_{local(v)})$ in $L(2)$. 

In view of Claim~\ref{claim5} for $i=2$, we have $T_{local(u)}=$
{\tt BuildTrie}$(\mathcal{S}_2(label(u)),E_1,E_2(1)))$ and $T_{local(v)}=$
{\tt BuildTrie}$(\mathcal{S}_2(label(v)),E_1,E_2(1)))$. Hence, Claim~\ref{claim6} for $i=2$ implies that {\tt LocalLabel}($\cB^2(u)$,$X_{u}$,$T_{label(u)})$ returns an integer at most equal to the number of leaves in $T_{label(u)})$, and implies the inequality {\tt LocalLabel}($\cB^2(v)$,$X_{v}$,$T_{label(v)})\geq1$. It follows that {\tt RetrieveLabel}$(\cB^2(u), E_1,E_2(2))<$ {\tt RetrieveLabel}$(\cB^2(v), E_1,E_2(2))$.

\item{\bf Case~3.} There is exactly one couple in $L(2)$ whose first term is $local(u)$ or $local(v)$. 

The inequality {\tt RetrieveLabel}$(\cB^2(u), E_1,E_2(2))<$ {\tt RetrieveLabel}$(\cB^2(v), E_1,E_2(2))$ is shown using similar arguments as in the two preceding cases.
\end{itemize}

This concludes the proof of Fact 2, and thus also the proof of  Claim~\ref{claim7} for $i=2$.
Hence, Claims~\ref{claim5},~\ref{claim6} and~\ref{claim7} are valid for $i=2$. The inductive step for these three claims is proved using similar arguments.

We now prove that the computation of $E_2$  in Algorithm {\tt ComputeAdvice} terminates, and that the length of $bin(E_2)$ is in $O(n \log n)$. In this algorithm, the value of $E_2$ is set to $E_2(\phi)$. Claim~\ref{claim5} implies that the computation of $E_2$ terminates, and that $E_2$ is a list
$((2,L(2)),(3,L(3)),\dots,(\phi,L(\phi)))$, where in every couple $(i,L(i))$, $i$ is an integer and $L(i)$ is a list of couples $(j,T_j)$, such that $j$ is an integer, and $T_j$ is a trie. In order to show that the length of $bin(E_2)$ is in $O(n \log n)$, it is enough to prove that the three conditions from Proposition~\ref{pro:nested} are satisfied.
In view of Claim~\ref{claim5}, the number of couples in $E_2$ is in $O(n)$ because $\phi\in O(n)$. 
Moreover, for every couple $(i,L(i))$ in $E_2$, there exists a couple $(j,*)\in L(i)$ if and only if $|\mathcal{S}_i(j)|\geq 2$. 
 Since $\mathcal{S}_i(j)$ is the set of augmented truncated views at depth $i$ of all nodes $u$ of $G$, such that {\tt RetrieveLabel}$(\cB^{i-1}(u), E_1,E_2(i-1))=j$,  Claim~\ref{claim6} implies that $j\in O(n)$. For every trie appearing in a term of some list $L(i)$, where $(i,L(i))$ is in $E_2$, all its leaves are labeled by $(0)$, 
 and all its internal nodes are labeled by  queries of the form $(a,b)$, where $a$ and $b$ are integers in $O(n)$. Hence, in order to show that the length of $bin(E_2)$ is in $O(n \log n)$, it is enough to show that the following two conditions are satisfied.\\
 C1. The sum of lengths of all lists $L_i$ appearing in terms of $E_2$ is in $O(n)$.\\
 C2. The sum of sizes of all tries appearing in terms of lists $L(i)$ appearing in terms of $E_2$ is in $O(n)$.
 
For every couple $(i,L(i))$ in $E_2$ and for every couple $(j,T_j)$ in $L(i)$, the term $T_j$ is a non-empty trie. Indeed, by property 3 of Claim~\ref{claim5}, we have $T_j=$ {\tt BuildTrie}$(\mathcal{S}_i(j),E_1,E_2(i-1)))$. In view of property 4 of Claim~\ref{claim5}, we have $|\mathcal{S}_i(j)|\geq 2$,  and in view of property 1 of Claim~\ref{claim6}, the size of $T_j$ is $2|\mathcal{S}_i(j)|-1$. Hence condition C2 implies condition C1, and thus it is enough to prove condition C2.

In view of Claim~\ref{claim5} and of property~1 of Claim~\ref{claim6}, for every list $L(i)$, such that the couple $(i,L(i))$ is in $E_2$, the sum of sizes of all tries appearing in terms of $L(i)$ is $(2\sum_{j\in\mathcal{H}}|\mathcal{S}_i(j)|)-|\mathcal{H}|$, where $\mathcal{H}$ is the set of integers $j$ such that $|\mathcal{S}_i(j)|\geq 2$ and $\mathcal{S}_i(j)$ is the set of augmented truncated views at depth $i$ of all nodes $u$ of $G$, for which {\tt RetrieveLabel}$(\cB^{i-1}(u), E_1,E_2(i-1))=j$. 

Let $\mathcal{H}'$ be the set of integers $j$, such that $|\mathcal{S}_i(j)|=1$. By Claim~\ref{claim7} we have $\mathcal{H}\cup\mathcal{H}'=\{1,2,\dots,|\mathcal{S}_{i-1}|\}$, $\bigcup_{j\in\mathcal{H}\cup\mathcal{H}'}\mathcal{S}_i(j)=\mathcal{S}_i$, and $ \mathcal{S}_i(j_1)\cap\mathcal{S}_i(j_2)=\emptyset$,  for distinct
$ j_1,j_2\in\mathcal{H}\cup\mathcal{H}'$. Hence $\sum_{j\in\mathcal{H}\cup\mathcal{H}'} |\mathcal{S}_i(j)|=|\mathcal{S}_i|$. It follows that:

\begin{align}
\sum_{j\in\mathcal{H}}|\mathcal{S}_i(j)|+ \sum_{j\in\mathcal{H}'}|\mathcal{S}_i(j)|  &= |\mathcal{S}_i| \\
2(\sum_{j\in\mathcal{H}}|\mathcal{S}_i(j)|+ \sum_{j\in\mathcal{H}'}|\mathcal{S}_i(j)|) + |\mathcal{H}|-|\mathcal{H}|  &= 2|\mathcal{S}_i| \\
2(\sum_{j\in\mathcal{H}}|\mathcal{S}_i(j)|)-|\mathcal{H}|  &= 2|\mathcal{S}_i| -2(\sum_{j\in\mathcal{H}'}|\mathcal{S}_i(j)|) - |\mathcal{H}| \\
  &= 2|\mathcal{S}_i| -2(\sum_{j\in\mathcal{H}'}|\mathcal{S}_i(j)|) - 2|\mathcal{H}|+|\mathcal{H}|\\
  &\leq 2|\mathcal{S}_i| -2|\mathcal{H}\cup\mathcal{H}'|+|\mathcal{H}|\\
  &\leq 2|\mathcal{S}_i| -2|\mathcal{S}_{i-1}|+|\mathcal{H}|
\end{align}

However, we know that $\sum_{j\in\mathcal{H}\cup\mathcal{H}'} |\mathcal{S}_i(j)|=|\mathcal{S}_i|$ and $\mathcal{H}\cup\mathcal{H}'=\{1,2,\dots,|\mathcal{S}_{i-1}|\}$. 
Thus we have:
\begin{align}
|\mathcal{S}_i|-|\mathcal{S}_{i-1}| &= \sum_{j\in\mathcal{H}\cup\mathcal{H}'} |\mathcal{S}_i(j)| - |\mathcal{H}\cup\mathcal{H}'|\\
&= \sum_{j\in\mathcal{H}\cup\mathcal{H}'}(|\mathcal{S}_i(j)|-1)\\
&= \sum_{j\in\mathcal{H}}(|\mathcal{S}_i(j)|-1)+\sum_{j\in\mathcal{H}'}(|\mathcal{S}_i(j)|-1)
\end{align}

Note that $\sum_{j\in\mathcal{H}'}(|\mathcal{S}_i(j)|-1)=0$ because, by definition of $\mathcal{H}'$, we have $|\mathcal{S}_i(j)|=1$, for all $j\in\mathcal{H}'$. It follows that
\begin{align}
|\mathcal{S}_i|-|\mathcal{S}_{i-1}| &= \sum_{j\in\mathcal{H}} |\mathcal{S}_i(j)|-|\mathcal{H}|
\end{align}

Since for all $j\in\mathcal{H}$ we have $|\mathcal{S}_i(j)|\geq2$, it follows that

\begin{align}
\sum_{j\in\mathcal{H}} |\mathcal{S}_i(j)|-|\mathcal{H}|&\geq|\mathcal{H}|.
\end{align}

Equations $(10)$ and $(11)$ imply that

\begin{align}
|\mathcal{S}_i|-|\mathcal{S}_{i-1}|&\geq|\mathcal{H}|.
\end{align}
 
Hence equations $(6)$ and $(12)$ imply:

\begin{align}
2(\sum_{j\in\mathcal{H}}|\mathcal{S}_i(j)|)-|\mathcal{H}|&\leq3(|\mathcal{S}_i| -|\mathcal{S}_{i-1}|).
\end{align}

As mentioned above, the sum of sizes of all tries appearing in terms of $L(i)$ is $(2\sum_{j\in\mathcal{H}}|\mathcal{S}_i(j)|)-|\mathcal{H}|$.
Thus, by Claim~\ref{claim5}, the sum of sizes of all tries appearing in terms of lists $L(i)$ appearing in terms of $E_2$  is at most $\sum_{i=2}^{\phi}3(|\mathcal{S}_i| -|\mathcal{S}_{i-1}|)=3(|\mathcal{S}_\phi| -|\mathcal{S}_{2}|)\leq 3n$. This proves condition C2 and concludes the proof that the length of $bin(E_2)$ is in $O(n \log n)$.

We are now able to conclude the proof of Part 1 of our theorem. We have seen that the computation of $E_2$ terminates, and that the length of $bin(E_2)$ is in $O(n \log n)$.  By Claim~\ref{claim4}, the computation of $E_1$ terminates and the length of $bin(E_1)$ is in $O(n \log n)$.
 By property 1 of Claim~\ref{claim7}, the computation of the labeled BFS tree $T$ in {\tt ComputeAdvice}$(G)$ terminates, and the labels of  nodes of 
  $T$ are in $O(n)$ because $|\mathcal{S}_\phi|=n$. Hence, Proposition~\ref{pro:bfs} implies that the length of $bin(T)$ is in $O(n \log n)$. Finally, the length of 
  $bin(\phi)$ is in $O(\log n)$ because $\phi\leq n$. It follows that Algorithm  {\tt ComputeAdvice}$(G)$ terminates, and the length of the returned string 
 $Concat(bin(\phi),A_1,A_2)$ (where $A_1= Concat (bin (E_1),bin(E_2))$ and $A_2=bin(T)$) is in $O(n \log n)$.
 
 It remains to prove Part 2 of the theorem. Using advice $Concat(bin(\phi),A_1,A_2)$ returned by ${\tt ComputeAdvice}(G)$, every node of $G$ executing Algorithm {\tt Elect} can decode the objects
 $\phi,E_1,E_2,$ and $T$ computed by Algorithm  {\tt ComputeAdvice}$(G)$. After $\phi$ rounds,
 each node $u$ of $G$ learns its augmented truncated view at depth $\phi$, and can execute {\tt RetrieveLabel}$(\cB^\phi(u), E_1,E_2)$ which terminates in view of Claims~\ref{claim4} (if $\phi=1$) and~\ref{claim7} (if $\phi\geq2$), as $E_2=E_2(\phi)$. According to these claims and by Proposition \ref{prop-index}, for all distinct nodes $u$ and $u'$ of $G$, we have
{\tt RetrieveLabel}$(\cB^\phi(u), E_1,E_2)$ $\ne$ {\tt RetrieveLabel}$(\cB^\phi(u'), E_1,E_2)$.  Moreover,  there exists exactly one node $u''$ of $G$ such that {\tt RetrieveLabel}$(\cB^\phi(u''), E_1,E_2)=1$. Since each node $u$ of the BFS tree $T$ is labeled by the integer {\tt RetrieveLabel}$(\cB^\phi(u), E_1,E_2)$, 
each node executing Algorithm {\tt Elect} outputs the sequence of port numbers corresponding to the unique simple path in the tree $T$, from this node to the node $u''$.
Consequently, all nodes perform correct leader election, which proves Part 2 of our theorem.
\end{proof}

We now prove two lower bounds on the size of advice for election in the minimum time, i.e., in time equal to the election index $\phi$. 
For $\phi=1$ we establish the lower bound $\Omega(n\log\log n)$,
and for $\phi >1$ we establish the lower bound $\Omega(n(\log\log n)^2/\log n)$. Both these bounds differ from our upper bound $O(n\log n)$ only by a polylogarithmic factor. 

The high-level idea of the proofs of these bounds is the following. Given a positive integer $\phi$ we construct, for arbitrarily large integers $n$,  families of $n$-node graphs with election index $\phi$
and with the property that each graph of such a family must receive different advice for any election algorithm working in time $\phi$ for all graphs of this family. This property is established by showing that,
if two graphs $G_1$ and $G_2$ from the family received the same advice, some nodes $v_1$ in $G_1$ and $v_2$ in $G_2$ would have to output identical sequences of port numbers because they have identical augmented truncated views at depth $\phi$, which would result in failure of leader election in one of these graphs. Since the constructed families are large enough, the above property implies the desired lower bound on the size of advice for at least one graph in the family.

We start with the construction of a family $\cF(x)=\{C_1,\dots ,C_y\}$ of labeled $(x+1)$-node cliques, for $x \geq 2$. All these cliques will have node labels  
 $r,v_0,v_1,\dots ,v_{x-1}$. We first define a clique $C$ by assigning its port numbers.
Assign port number $i$, for $0\leq i \leq x-1$, to the port at $r$ corresponding to the edge $\{r,v_i\}$. The rest of the port numbers are assigned arbitrarily. {We now show how} to produce the cliques of the family $\cF(x)$ from the clique $C$.
Consider all sequences of $x$ integers from the set $\{1,\dots, x-1\}$. There are  $y=(x-1)^x$ such sequences. Let $(s_1,\dots, s_y)$ be
any enumeration of them. Let $s_t=(h_0,h_1,\dots , h_{x-1})$, for a fixed $t=1,\dots, y$. 
The clique $C_t$ is defined from clique $C$ by assigning port $(p+ h_j) \mod x$ instead of port $p$ at node $v_j$, 
for all pairs $0\leq j, p \leq x-1$.

We first consider the election index $\phi=1$.

\begin{theorem}\label{lb1}
For arbitrarily large integers $n$, there exist $n$-node graphs with election index 1, such that leader election in time 1 in these graphs requires advice of size 
$\Omega(n \log\log n)$.
\end{theorem}

\begin{proof}
Fix an integer $k \geq 2^{16}$, and let $x=\lceil 2\log k/\log\log k\rceil$. We have $k \leq y= (x-1)^x$.
We first define a graph $H_k$ using the family $\cF(x)${, cf. Fig.~\ref{fig:f1}}.  

\begin{figure}[httb!]
	\begin{center}
	\includegraphics[width=0.4\textwidth]{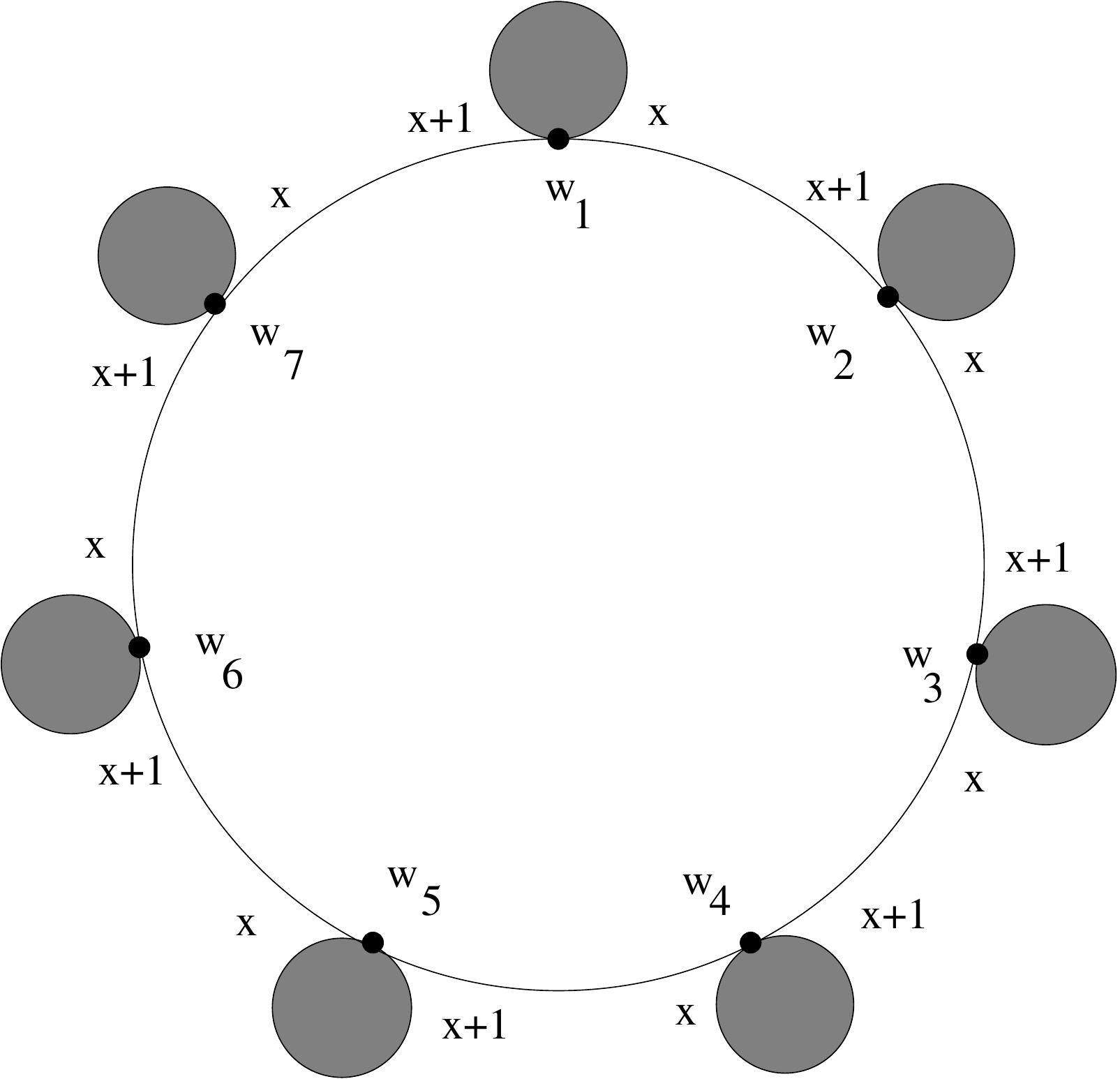}
	\caption{A representation of the graph $H_7$. The grey discs represent cliques from $\cF(x)$.}
	\label{fig:f1}
	\end{center}
\end{figure}

Consider a ring of size $k$ with nodes $w_1, \dots , w_k$. Attach an isomorphic copy of the clique $C_t$ to node $w_t$,
by identifying $w_t$ with node $r$ of this copy and taking all other nodes different from the nodes of the ring. (The term ``isomorphic'' means that all port numbers are preserved.)  
All attached cliques are pairwise node-disjoint.
Assign ports $x$ and $x+1$ corresponding to edges of the ring at each of its nodes, in the clockwise order. This concludes the construction of graph $H_k$.

Finally we produce a family $\cG_k$ consisting of $(k-1)!$ graphs as follows. Keep the clique at node $w_1$ of the ring in {$H_k$} fixed, and permute arbitrarily cliques attached to all other nodes of the ring. 
Then delete all node labels.
 
The proof  relies on the following two claims.
The first claim establishes the election index of graphs in  $\cG_k$.

\begin{claim}\label{index1}
All graphs in the family $\cG_k$ have election index 1.
\end{claim}

To prove the claim it is enough to show that all these graphs have election index at most 1.  Hence it suffices to show that all augmented truncated views at depth 1 are distinct.
Fix a graph $G$ in $\cG_k$. Consider two nodes $u$ and $v$ of $G$. First suppose that they belong to the same clique. If they have different degrees (in the graph $G$) then 
 $\cB^1(u) \neq \cB^1(v)$. If they have the same degree, then the port numbers at the unique node $r$ of degree $x+2$ in this clique corresponding
 to edges $\{r,u\}$ and $\{r,v\}$ must be different, and hence $\cB^1(u) \neq \cB^1(v)$ as well. Next suppose that $u$ and $v$ are in different cliques $C_u$ and $C_v$,
 respectively.
 Again, if they have different degrees then 
 $\cB^1(u) \neq \cB^1(v)$. Hence assume that they have the same degree. Consider two cases.
 
 \noindent
 Case 1. The degree of $u$ and of $v$ is $x+2$.

 By the construction of the family $\cF(x)$, there exists an integer $0 \leq i \leq x-1$ with the following property.
 Let $u'$ be the node of $C_u$, such that the port at $u$ corresponding to edge $\{u,u'\}$ is $i$, and let
 $v'$ be the node of $C_v$, such that the port at $v$ corresponding to edge $\{v,v'\}$ is $i$.
 The port at $u'$ corresponding to edge $\{u,u'\}$ is different from  the  port at $v'$ corresponding to edge $\{v,v'\}$, and hence  $\cB^1(u) \neq \cB^1(v)$.
 
  \noindent
 Case 2. The degree of $u$ and of $v$ is $x$.

 Let $r_u$ be the unique node of degree $x+2$ in the clique $C_u$, and let $r_v$ be the unique node of degree $x+2$ in the clique $C_v$.
 Consider the edges $e_u=\{r_u,u\}$ and $e_v=\{r_v,v\}$. If the port number at $r_u$ corresponding to edge $e_u$ is different from the
 port number at $r_v$ corresponding to edge $e_v$, or  the port number at $u$ corresponding to edge $e_u$ is different from the
 port number at $v$ corresponding to edge $e_v$, then $\cB^1(u) \neq \cB^1(v)$. Hence assume that the respective port numbers are equal. 
 For any integer $0 \leq i \leq x-1$, let $a_i$ be the node of degree $x$ in $C_u$, such that the port number at $r_u$ corresponding to edge $\{r_u, a_i\}$ is $i$,
 and let $b_i$ be the node of degree $x$ in $C_v$, such that the port number at $r_v$ corresponding to edge $\{r_v, b_i\}$ is $i$. 
 By the construction of  the family $\cF(x)$, we have the following two properties:\\
 1. The port number at $u$ corresponding to edge $\{u, a_i\}$, for any $a_i\neq u$, is equal to the port number at  $v$ corresponding to edge $\{v, b_i\}$
 (because the port number at $u$ corresponding to edge $e_u$ is equal to the
 port number at $v$ corresponding to edge $e_v$);\\
 2. There exists an integer $0 \leq i \leq x-1$, {such that} $a_i\neq u$ and the port number at $a_i$ corresponding to edge $\{u, a_i\}$, is different from the port number at  $b_i$ corresponding to edge $\{v, b_i\}$ (because the cliques $C_u$ and $C_v$ correspond to different cliques from the family $\cF(x)$).\\
 Hence $\cB^1(u) \neq \cB^1(v)$ in all cases.  This concludes the proof of the claim. 
 
 The next claim will imply a lower bound on the number of different pieces of advice needed to perform election in the family  $\cG_k$ in time 1.
 
\begin{claim}\label{distinct-advice} 
Consider any election algorithm working for the family  $\cG_k$ in time 1. The advice given to distinct graphs in this family must be different. 
\end{claim}

In the proof of the claim we will use the following observation that follows from the fact that port numbers in the ring in all graphs of the family  $\cG_k$ are the same.

\noindent
{\bf Observation.}
Let $G_1$ and $G_2$ be any graphs from the family $\cG_k$. Let $C_t$ be an arbitrary clique from the family $\cF(x)$ used to form the graph $H_k$.
For $j=1,2$, let $r_j$ be the node of $G_j$ by which the clique isomorphic to $C_t$ is attached to the ring in this graph. Then $\cB^1(r_1) = \cB^1(r_2)$.

The proof of the claim is by contradiction. Fix an election algorithm and suppose that two graphs $G_1$ and $G_2$ from the family $\cG_k$ get the same advice.
Let $z_1$ be the node elected in $G_1$ and $z_2$ the node elected in $G_2$. Denote by $C'$ the clique containing $z_1$ and by $C''$ the clique containing
$z_2$. Let $r'$ be the unique node of degree $x+2$ in $C'$ and let $r''$ be the unique node of degree $x+2$ in $C''$.

Consider two cases.

\noindent
Case 1. The cliques $C'$ and $C''$ are non-isomorphic.

Let $s$ be the node in $G_2$ at which the clique isomorphic to $C'$ is attached. By the observation, the augmented truncated view $\cB^1(r')$ in $G_1$
is equal to the augmented truncated view $\cB^1(s)$ in $G_2$.  Since $G_1$ and $G_2$ get the same advice, nodes $r'$ in $G_1$ and $s$ in $G_2$ must output
the same sequence of port numbers corresponding to a simple path to the leader. The sequence outputted by node $r'$ in $G_1$ cannot contain the port number $x$, hence the sequence outputted by node $s$ in $G_2$ cannot contain the port number $x$ either.
This is a contradiction, because every path from $s$ to $z_2$ in $G_2$ must use port $x$ at least once.  

\noindent
Case 2. The cliques $C'$ and $C''$ are isomorphic.

Since graphs $G_1$ and $G_2$ differ only by the permutation of cliques, there must exist isomorphic cliques $D'$ in $G_1$ and $D''$ in $G_2$, such that the clockwise
distance (in the ring) from the unique node $s'$ of degree $x+2$ in $D'$ to node $r'$ is different than the clockwise
distance (in the ring) from the unique node $s''$ of degree $x+2$ in $D''$ to node $r''$. Since $D'$ and $D''$ are isomorphic, 
by the observation, the augmented truncated view $\cB^1(s')$ in $G_1$
is equal to the augmented truncated view $\cB^1(s'')$ in $G_2$.  Since $G_1$ and $G_2$ get the same advice, nodes $s'$ in $G_1$ and $s''$ in $G_2$ output
the same sequence of port numbers corresponding to a simple path to the leader. Without loss of generality, suppose that the first number in the outputted sequences is $x$;
the case when it is $x+1$ is similar. Due to the differences in clockwise distances from $s'$ to $r'$ and from $s''$ to $r''$, the number of integers $x$ in both sequences 
must be different, which gives a contradiction. This concludes the proof of the claim.

Our lower bound will be shown on the family $\cG= \bigcup _{k=2^{16}}^\infty \cG_k$. 
Consider an election algorithm working in all graphs of this family in time 1. 
For any $k\geq 2^{16}$, let $n_k=k( \lceil 2\log k/\log\log k\rceil +1)$.
Graphs in the family $\cG _k$ have size $n_k$. 
By Claim \ref{index1}, all these graphs have election index 1. By Claim \ref{distinct-advice}, all of them must get different advice.
Since, for any $k\geq 2^{16}$, there are $(k-1)!$ graphs in $\cG _k$, at least one of them must get advice of size $\Omega( \log ((k-1)!))=\Omega(k \log k)$.
We have $k \log k \in \Theta(n_k \log\log n_k)$. Hence there exists an infinite sequence of integers $n_k$ such that there are $n_k$-node graphs  with election index 1
that require advice of size  $\Omega(n_k \log\log n_k)$ for election in time~1.
\end{proof}

We next consider the election index $\phi>1$. The lower bound in this case uses a construction slightly more complicated than for $\phi=1$. Note that a straightforward generalization of the previous construction would lead to a lower bound $\Omega(n\log\log n/\phi)$ which would be too weak for our purpose, as $\phi$ can be
much larger than polylogarithmic in $n$. 

\begin{theorem}\label{lb>1}
Let $\phi$ be an integer larger than 1.
For arbitrarily large integers $n$, there exist $n$-node graphs with election index $\phi$, such that leader election in time $\phi$ in these graphs requires advice of size 
$\Omega(n (\log\log n)^2)/\log n)$.
\end{theorem}

\begin{proof}
Fix an even integer $k \geq 2^{16}$, and let $x=\lceil 2\log k/\log\log k\rceil$. We have $k \leq y= (x-1)^x$.
We construct a family $\cN_k$ of graphs, called $k$-{\em necklaces}. This family  is derived from a graph $M_k$ defined as follows{, cf. Fig. \ref{fig:f2}}. 

\begin{figure}[httb!]
	\begin{center}
	\includegraphics[width=0.8\textwidth]{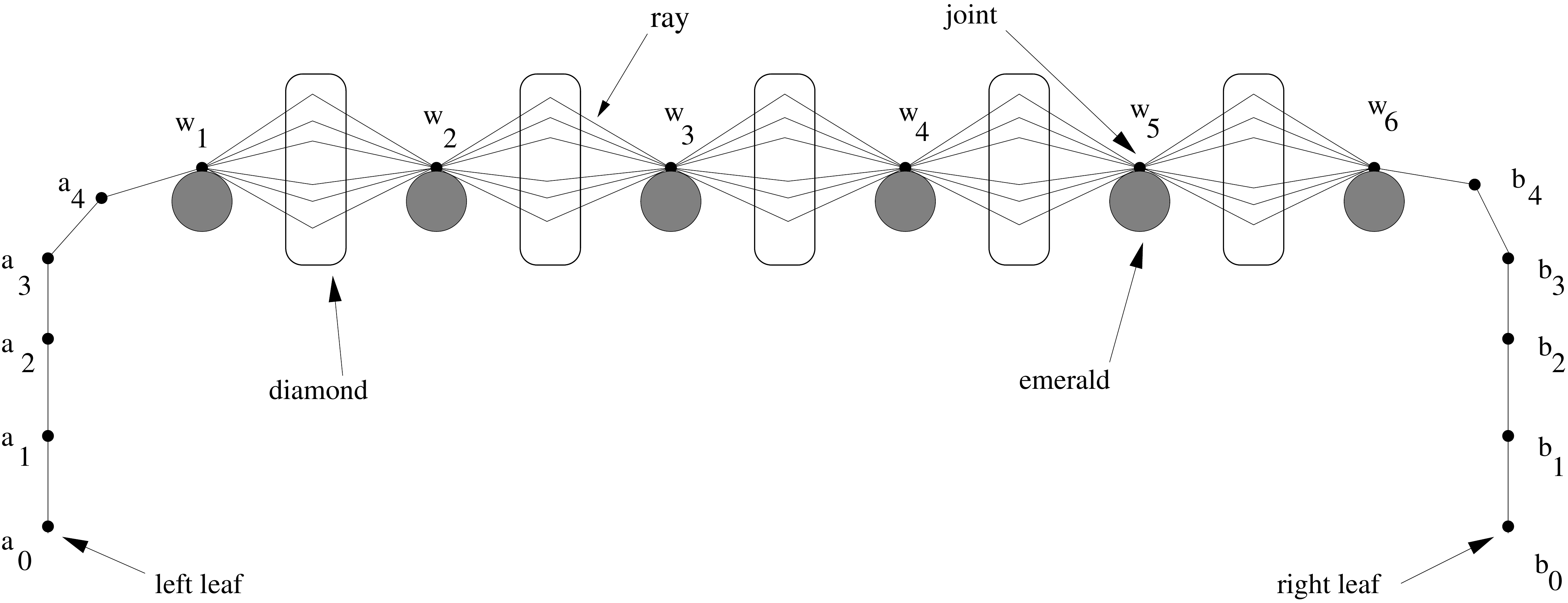}
	\caption{A representation of the graph $M_k$ for $k=6$.}
	\label{fig:f2}
	\end{center}
\end{figure}

Let $w_1,\dots, w_k$ be nodes, that will be called {\em joints}. Let $D_1, \dots, D_{k-1}$
be pairwise disjoint cliques of size $x$. These cliques will be called {\em diamonds}. Attach every node of $D_i$, for $i=1,\dots,k-1$, to $w_i$ and to $w_{i+1}$ by edges called {\em rays}.
Next, let $E_1,\dots ,E_k$ be distinct cliques from the family $\cF(x)$. These cliques will be called {\em emeralds}. Attach $E_i$ to $w_i$, for $i=1,\dots,k$, by identifying
$w_i$ with the node $r$ of $E_i$ (see the definition of the family $\cF(x)$). Since $k \leq (x-1)^x$, there are enough cliques in $\cF(x)$. Finally, consider chains of nodes
$(a_0,\dots,a_{\phi -2})$ and $(b_0,\dots,b_{\phi -2})$, such that all these nodes are distinct and different from all previously constructed nodes. Attach $a_{\phi -2}$ to $w_1$ by an edge and attach $b_{\phi -2}$ to $w_k$ by an edge. 

We next assign port numbers to the above constructed graph $M_k$. First fix the same arbitrary port numbering (from the range $\{0,\dots ,x-2\}$) inside each diamond $D_i$. Then,  
for any diamond $D_i$, assign number $x-1$ to the port at any node of $D_i$ corresponding to the ray joining it to $w_i$, and assign number $x$ to the port at any node of $D_i$ corresponding to the ray joining it to $w_{i+1}$. Keep the port numbering inside each emerald as it is defined in the description of the family $\cF(x)$. 
Next we define port numbers corresponding to all rays at nodes $w_1,\dots, w_k$. Port numbers at node $w_1$ corresponding to rays joining it  to diamond $D_1$,
and port numbers at node $w_k$ corresponding to rays joining it  to diamond $D_{k-1}$ are assigned arbitrarily from the range $\{x,\dots,2x-1\}$. The port number at node $w_1$ corresponding to the edge joining it to $a_{\phi -2}$, and the port number at node $w_k$ corresponding to the edge joining it to $b_{\phi -2}$ is $2x$.   
Consider a node $w_i$,
for $1<i<k$. If $i$ is even, then port numbers at node $w_i$ corresponding to rays joining it  to diamond $D_{i-1}$ are assigned arbitrarily from the range $\{x,\dots,2x-1\}$,
and port numbers at node $w_i$ corresponding to rays joining it  to diamond $D_{i}$ are assigned arbitrarily from the range $\{2x,\dots,3x-1\}$.
If $i$ is odd, then port numbers at node $w_i$ corresponding to rays joining it  to diamond $D_{i-1}$ are assigned arbitrarily from the range $\{2x,\dots,3x-1\}$,
and port numbers at node $w_i$ corresponding to rays joining it  to diamond $D_{i}$ are assigned arbitrarily from the range $\{x,\dots,2x-1\}$. It remains to assign port numbers at nodes of the two chains. {The unique port number at nodes $a_0$ and $b_0$
is 0.  Call node $a_0$ the {\em left leaf} and call node $b_0$ the {\em right leaf}. If $\phi>2$ then $\phi-2>0$, and the port number at node $a_{\phi-2}$ (resp. $b_{\phi-2}$) corresponding to the edge joining it to $w_1$ (resp. to $w_k$) is $0$, while the other port number at node $a_{\phi-2}$ (resp. $b_{\phi-2}$) corresponding to the unique other edge is $1$. If $\phi>3$ then $\phi-2>1$ and, for every $0<i<\phi-2$, the port number at node $a_i$ (resp. $b_i$) corresponding to the edge joining it to $a_{i-1}$ (resp. $b_{i-1}$) is 1, and the port number at node $a_i$ (resp. $b_i$) corresponding to the edge joining it to $a_{i+1}$ (resp. $b_{i+1}$) is 0.} Finally we delete all node labels.  This concludes the construction of graph  $M_k$.


The family  $\cN_k$ of $k$-necklaces is defined from the graph $M_k$ as follows. Let $C=(c_1,\dots ,c_k)$ be any sequence of integers from the range $\{0,\dots ,x\}$,
such that $c_1=c_k=0$. Such a sequence is called the {\em code} of a graph in $\cN_k$. The graph corresponding to code $C$ is obtained from $M_k$ by
replacing any port number $p$ at any node of $D_i$, for $i \leq k$, by $(p+c_i) \mod (x+1)$. All the rest of the graph $M_k$ remains intact. This concludes the construction of the family of $k$-necklaces.

The proof  relies on the following two claims similar to those from the proof of Theorem \ref{lb1}.
The first claim establishes the election index of graphs in  $\cN_k$.

\begin{claim}\label{index-phi}
All graphs in the family $\cN_k$ have election index $\phi$.
\end{claim}

To prove the claim first observe that the election index of graphs in the family $\cN_k$ is at least $\phi$ because, by the construction of the graph $M_k$,
we have $\cB^{\phi -1}(v)= \cB^{\phi -1}(w)$, where $v$ and $w$ are the only nodes of degree 1 in any $k$-necklace {(i.e., the left and right leaves)}. Hence it is enough to prove that the augmented truncated views at
depth $\phi$ of any two nodes in any $k$-necklace are different.

Consider any two nodes $v$ and $w$ in a $k$-necklace. 
If {$v$} is a joint and {$w$} is not, then $v$ and $w$ have different degrees, hence {$\cB^1(v) \neq \cB^1(w)$}, and hence $\cB^{\phi }(v)\neq \cB^{\phi }(w)$.
If both $v$ and $w$ are joints then {$\cB^1(v) \neq \cB^1(w)$}, using an argument similar to that in Case 1 in the proof of Claim \ref{index1}
(because {$v$ and $w$} correspond to nodes $r$ in two different graphs from the family $\cF(x)$), and hence $\cB^{\phi }(v)\neq\cB^{\phi }(w)$. Hence we may assume that none of nodes $v$ and $w$ are joints. Every node that is not a joint is at distance at most $\phi-1$ from a joint. Consider the joint $v'$ corresponding
to the lexicographically smallest among the shortest paths from $v$ to a joint, and a joint $w'$ corresponding
to the lexicographically smallest among the shortest paths from {$w$} to a joint. Let $s_v$ be the sequence of port numbers in the first path, and
let $s_w$ be the sequence of port numbers in the second path.  If $s_v \neq s_w$, then $\cB^{\phi-1 }(v)\neq\cB^{\phi -1 }(w)$, and hence $\cB^{\phi }(v)\neq\cB^{\phi }(w)$.
If $s_v= s_w$ then $v'\neq w'$ because the same sequence of port numbers cannot correspond to paths from distinct nodes to the same node.
As stated above, the augmented truncated views at depth 1
of all joints are unique. 
It follows that  $\cB^{\phi }(v)\neq \cB^{\phi }(w)$. This proves the claim.

 The next claim will imply a lower bound on the number of different pieces of advice needed to perform election in the family  $\cN_k$ in time $\phi$.
 
\begin{claim}\label{distinct-advice-phi} 
Consider any election algorithm working for the family  $\cN_k$ in time $\phi$. The advice given to distinct graphs in this family must be different. 
\end{claim}

In the proof of the claim we will use the following observation following from the fact that all codes of $k$-necklaces start and finish with a 0.

\noindent
{\bf Observation.}
Let $N_1$ and $N_2$ be any graphs from the family $\cN_k$. Augmented truncated views at depth $\phi$ of left leaves in $N_1$ and $N_2$ are equal, and
augmented truncated views at depth $\phi$ of right leaves in $N_1$ and $N_2$ are equal.  

The proof of the claim is by contradiction. Fix an election algorithm and suppose that two graphs from the family $\cN_k$, graph $N_1$ with code {$(c_1,\dots, c_{k})$} and graph $N_2$ with code {$(c'_1,\dots, c'_{k})$}, get the same advice. Let $i$ be the smallest index such that $c_i \neq c'_i$. 
In view of the unicity of the leader in every  $k$-necklace, the sequence outputted by the left leaf of $N_1$, or the sequence outputted by the right leaf of $N_1$
must correspond to a simple path containing a ray of the diamond $D_i$. Suppose that this is the case for the left leaf and that $i$ is even. The remaining three cases
can be analyzed similarly. Denote by $\sigma$ the prefix of the sequence outputted by the left leaf of $N_1$ that corresponds to a path from this leaf to the ``$i$th joint
from the left'', i.e., more precisely, the joint at distance $(\phi -1)+2(i-1)$ from the left leaf. The first term following the prefix $\sigma$ in the sequence outputted by the left leaf of $N_1$ is some integer $y \in \{2x,2x+1,\dots , 3x-1\}$ (because otherwise the corresponding path would visit at least twice the $i$th joint
from the left, and hence this path would not be simple) and the second term following the prefix $\sigma$ in the sequence outputted by the left leaf of $N_1$ is $(x+c_i) \mod (x+1)$ by the construction of the family $\cN_k$.
By the observation, the left leaf of $N_2$ outputs the same sequence as the left leaf of $N_1$, with the same prefix $\sigma$, followed by$(y,(x+c_i) \mod (x+1))$.
By the minimality of $i$, the path corresponding to $\sigma$ reaches the $i$th joint
from the left in $N_2$. By the construction of the family $\cN_k$, the edge incident to the $i$th joint
from the left in $N_2$ with port number $y$ at this joint has the other port number $(x+c'_i) \mod (x+1)$. 
This is a contradiction, because $(x+c_i) \mod (x+1)\neq (x+c'_i) \mod (x+1)$, as $c_i\neq c'_i$ and both $c_i$ and $c'_i$ are at most $x$. This proves the claim.

Consider the family $\cN= \bigcup _{k=\max(2^{16},\phi)}^\infty \cN_k$. Our lower bound will be proven on the family $\cN$.
Consider an election algorithm working in all graphs of this family in time $\phi$. For any $k \geq \max(2^{16},\phi)$, let 
$n_k=2(\phi-1)+ k(x-1)+(k-1)x$. Graphs in the family $\cN _k$ have size $n_k$. Since $x=\lceil 2\log k/\log\log k\rceil$, we have $n_k \in \Theta(k\log k/\log\log k)$.
By Claim \ref{index-phi}, all these graphs have election index $\phi$. By Claim \ref{distinct-advice-phi}, all of them must get different advice.
Since, for any $k \geq \max(2^{16},\phi)$, there are $(x+1)^{k-3}$ graphs in $\cN _k$, at least one of them must get advice of size $\Omega( \log( (x+1)^{k-3}))=\Omega(k \log\log k)$. We have $k \log\log k \in \Theta(n_k( \log\log n_k)^2/\log n_k)$. 
Hence there exists an infinite sequence of integers $n_k$ such that there are $n_k$-node graphs  with election index $\phi$
that require advice of size  $ \Omega(n_k( \log\log n_k)^2/\log n_k)$ for election in time $\phi$.
\end{proof}

\section{Election in large time}

In this section we study the minimum size of advice sufficient to perform leader election when the allocated time is large, i.e., when it exceeds the diameter of the graph
by an additive offset which is some function of the election index $\phi$ of the graph. We consider four values of this offset, for an integer constant $c>1$:  $\phi+c$,
$c\phi$, $\phi^c$, and $c^{\phi}$. In the first case the offset is asymptotically equal to $\phi$, in the second case it is linear in $\phi$ but the multiplicative  constant is larger than 1,
in the third case it is polynomial in $\phi$ but super-linear, and in the fourth case it is exponential in $\phi$. Note that, even in the first case, that calls for the fastest election among these four cases (in time $D+\phi +c$), the allocated
time is large enough for all nodes to see all the differences in truncated views of other nodes, which makes a huge difference between leader election in such a time and in the minimum possible time $\phi$. For all these four election times, we  establish tight bounds on the minimum size of advice that enables election in this time,
up to multiplicative constants. 

We start by designing an algorithm that performs leader election in time at most $D+x+1$, for any graph of diameter $D$ and election index $\phi$,
provided that nodes receive as input an integer $x \geq \phi)$. Note that nodes of the graph know  neither $D$ nor $\phi$.
We will then show how to derive from this generic algorithm four leader election algorithms using larger and larger time and smaller and smaller advice.

The high-level idea of  Algorithm {\tt Generic} is the following. Nodes of the graph communicate between them and acquire augmented truncated views at increasing depths. Starting from round $x$, they discover augmented truncated views at depth $x$ of an increasing set of nodes. They stop in the round when no new 
augmented truncated views at depth $x$ are discovered. At this time, we have the guarantee that all nodes learned all augmented truncated views at depth $x$.
Hence, to solve leader election, it suffices that every node outputs a sequence of port numbers leading to a node with the  lexicographically smallest
augmented truncated view at depth $x$.  
Since $x \geq \phi$, the augmented truncated view at depth $x$ of every node is unique in the graph. Hence all nodes output sequences of port numbers leading to the same
node.

Below we give a detailed description of Algorithm {\tt Generic}. It is executed by a node $u$ that gets the integer $x$ as input. 

\begin{algorithm}
{ \caption{\tt{Generic}$(x)$\label{alg:generic}}

{\bf for} $r:=0$ {\bf to} $x-1$ {\bf do} $COM(r)$\\
$r \leftarrow x$\\
{\bf repeat}\\
\hspace*{1cm}$COM(r)$\\
\hspace*{1cm}$\cB \leftarrow \cB^{r+1}(u)$\\
\hspace*{1cm}$X \leftarrow$ the set of augmented truncated views $\cB^{x}(v)$,\\ 
\hspace*{2cm}for all nodes $v$ at depth at most $r-x$ in $\cB$\\ 
\hspace*{1cm}$Y \leftarrow$ the set of augmented truncated views $\cB^{x}(v)$,\\ 
\hspace*{2cm}for all nodes $v$ at depth exactly $r-x+1$ in $\cB$
\hspace*{1cm}\\
{\bf until} $Y \subseteq X$\\
$\cB_{min} \leftarrow$ the lexicographically smallest among augmented truncated views from $X$\\
$W \leftarrow$ the set of nodes $v$ of smallest depth in $\cB$, such that $\cB^{x}(v)=\cB_{min}$\\
$w \leftarrow$ the node from $W$ corresponding to the lexicographically smallest sequence of port numbers\\
{\bf return} the sequence of port numbers corresponding to the shortest path from $u$ to {$w$} in $\cB$

}
\end{algorithm}

The following lemma establishes the correctness of Algorithm {\tt Generic} and estimates its execution time.

\begin{lemma}\label{generic}
For any graph $G$ of diameter $D$ and election index $\phi$,  Algorithm {\tt Generic}$(x)$, with any parameter $x \geq \phi$, is a correct  leader election algorithm
and works in time at most $D+x+1$.
\end{lemma}

\begin{proof}
All rounds are numbered starting at 0.
Let $u$ be a node of graph $G$ executing Algorithm {\tt Generic}$(x)$, with any parameter $x \geq \phi$.
Let $S$ be the set of augmented truncated views at depth $x$ in $G$. Let $s$ be the node of $G$ such that $\cB^{x}(s)$ is the lexicographically smallest among the
 augmented truncated views from $S$. The node $s$ is unique in view of Proposition \ref{prop-index} and because $x \geq \phi$. First we show that the execution of Algorithm {\tt Generic}$(x)$ by node $u$ stops
at the latest in round $D+x$. If this execution did not stop earlier, in round $D+x$ the value of $X$ is $S$
because, in round $D+x$, the set $X$ is the set of augmented truncated views $\cB^{x}(v)$,
for all nodes $v$ at depth at most $D$ in $\cB^{D+x+1}(u)$. Since $Y\subseteq S$ in every round, it follows that $Y\subseteq X$ in round $D+x$. Hence, in round
$D+x$, node $u$ outputs a sequence of port numbers and stops its execution. 

It remains to show that the sequences of port numbers outputted by all nodes in $G$ correspond to simple paths in this graph, whose other extremity is the same node.
In view of the unicity of $s$, it is enough to show that the node $u$ outputs a sequence corresponding to a simple path in $G$, whose other extremity is $s$.
By the description of the algorithm, node $u$ cannot terminate before round $x$. Hence, as shown above,
there exists an integer $0 \leq j \leq D$, such that $u$ outputs a sequence of ports in round $x+j$. Hence,  in this round, $Y\subseteq X$.
Since every node in $G$ has a unique augmented truncated view at depth $x$, all nodes of $G$ are at distance at most $j$ from $u$.
Hence, in round $x+j$, the set $X$ contains $\cB^{x}(s)$.
By the definition of the variable $w$ in the algorithm, it follows that the node $u$ outputs a sequence of port numbers corresponding to a shortest path leading to $s$. Such a shortest path must be simple, which concludes the proof.  
\end{proof}

We now describe four algorithms, called {\tt Election}$_i$, for $i=1,2,3,4$, working for graphs of diameter $D$ and election index $\phi$.
Recall the notation $^ic$ defined by induction as follows: $^0c=1$ and $^{i+1}c=c^{^ic}$. Intuitively it denotes a tower of powers.
For an integer constant $c>1$, let $T_1=D+\phi+c$, $T_2=D+c\phi$, $T_3=D+\phi^c$, and $T_4=D+c^{\phi}$. 
Let $A_1$ be the binary representation of $\phi$, let $A_2$ be the binary representation of $\lfloor \log \phi \rfloor$, let $A_3$ be the binary representation of $\lfloor \log \log \phi \rfloor$, and
let $A_4$ be the binary representation of $\log^*\phi$. Hence the size of $A_1$ is $O(\log \phi)$, the size of $A_2$ is $O(\log\log \phi)$, the size of $A_3$ is $O(\log\log\log \phi)$,
and the size of $A_4$ is $O(\log (\log^* \phi))$.
Define the following integers. $P_1=\phi$, $P_2=2^{\lfloor \log \phi \rfloor+1}-1$, $P_3=2^{2^{\lfloor \log \log \phi \rfloor+1}}-1$, and $P_4={^{(\log^*\phi)+1}}2-1$.

Algorithm {\tt Election}$_i$ uses advice $A_i$ and will be shown to work in time $T_i$. It consists of a single instruction:

\begin{algorithm}
{ \caption{{\tt Election}$_i$\label{alg:elect-i}}

{\tt Generic}$(P_i)$

}
\end{algorithm}

We will prove the following theorem.

\begin{theorem}\label{theorem-quadruple}
Let $G$ be a graph of diameter $D$ and election index $\phi$. Let $c>1$ be any integer constant. 
\begin{enumerate}
\item
Algorithm {\tt Election}$_1$ solves  leader election in $G$ in time at most $D+\phi +c$ and using $O(\log \phi)$ bits of advice.
\item
Algorithm {\tt Election}$_2$ solves leader election in $G$ in time at most $D+c\phi$ and using $O(\log\log \phi)$ bits of advice.
\item
Algorithm {\tt Election}$_3$ solves leader election in $G$ in time at most $D+\phi^c$ and using $O(\log\log\log \phi)$ bits of advice.
\item
Algorithm {\tt Election}$_4$ solves leader election in $G$ in time at most $D+c^{\phi}$ and using $O(\log(\log^* \phi))$ bits of advice.
\end{enumerate}
\end{theorem}

\begin{proof} 

1. Algorithm {\tt Election}$_1$ uses advice $A_1$ which is the binary representation of $\phi$ of size $O(\log \phi)$. It first computes $\phi$ using $A_1$, and then calls  
Algorithm {\tt Generic}$(\phi)$. In view of Lemma \ref{generic}, Algorithm {\tt Generic}$(\phi)$ solves leader election in $G$ in time at most $D+\phi +1<D+\phi+c$.
This proves part 1 of the theorem.  

2. Algorithm {\tt Election}$_2$ uses advice $A_2$ which is the binary representation of $\lfloor \log \phi\rfloor$ of size $O(\log \log \phi)$. 
It first computes $\lfloor \log \phi\rfloor$ using $A_2$, then computes $P_2=2^{\lfloor \log \phi \rfloor+1}-1$ and calls Algorithm {\tt Generic}$(P_2)$. Notice that $P_2 \geq  \phi$. In view of Lemma \ref{generic}, Algorithm {\tt Generic}$(P_2)$ solves leader election in $G$ in time at most $D+P_2+1 = D+2^{\lfloor \log \phi \rfloor+1}$.
This is at most $D+2\phi$ and hence at most $D+c\phi$, since $c$ is an integer larger than 1.

3. Algorithm {\tt Election}$_3$ uses advice $A_3$ which is the binary representation of $\lfloor \log \log \phi\rfloor$ of size $O(\log\log \log \phi)$. 
It first computes $\lfloor \log\log \phi\rfloor$ using $A_3$, then computes $P_3=2^{2^{\lfloor \log \log \phi \rfloor+1}}-1$ and calls Algorithm {\tt Generic}$(P_3)$. Notice that $P_3 \geq  \phi$. In view of Lemma \ref{generic}, Algorithm {\tt Generic}$(P_3)$ solves leader election in $G$ in time at most 
$D+P_3+1 = D+2^{2^{\lfloor \log \log \phi \rfloor+1}}$.
This is at most $D+\phi^2$ and hence at most $D+\phi^c$, since $c$ is an integer larger than 1.

4. Algorithm {\tt Election}$_4$ uses advice $A_4$ which is the binary representation of $\log^* \phi$ of size $O(\log( \log^* \phi))$. 
It first computes $\log^* \phi$ using $A_4$, then computes $P_4={^{(\log^*\phi)+1}}2-1$ and calls Algorithm {\tt Generic}$(P_4)$. Notice that $P_4 \geq  \phi$. In view of Lemma \ref{generic}, Algorithm {\tt Generic}$(P_4)$ solves leader election in $G$ in time at most $D+P_4+1 = D+{^{(\log^*\phi)+1}}2$.
This is at most $D+2^{\phi}$ and hence at most $D+c^{\phi}$, since $c$ is an integer larger than 1.
\end{proof}

{\bf Remark.} Notice that the first part of the theorem remains valid for $c=1$, and the proof remains the same. Hence, it is possible to perform leader election in time $D+\phi+1$ using $O(\log \phi)$
bits of advice. We do not know if the same is true for the time $D+\phi$, but in this time it is possible to elect a leader using $O(\log D +\log \phi)$ bits of advice. Indeed, it suffices to provide the nodes with values of the diameter $D$ and of the election index $\phi$. Equipped with this information, each node $u$ learns $\cB^{D+\phi}(u)$ in time $D+\phi$. Then, knowing $D$, it knows that nodes that it sees
in this augmented truncated view at distance at most $D$ from the root of this view represent all nodes of the graph. Knowing the value of $\phi$, node $u$ can reconstruct $\cB^{\phi}(v)$, 
for each such node $v$, and hence find in
$\cB^{D+\phi}(u)$ a representation of the node $w$ in the graph, whose augmented truncated view $\cB^{\phi}(v)$ is lexicographically smallest. Finally, the node $u$ can output a sequence of port numbers corresponding to {one of
the shortest paths} from $u$ to $w$ in $\cB^{D+\phi}(u)$.

The following theorem provides matching lower bounds (up to multiplicative constants) on the minimum size of advice sufficient to perform leader election in time corresponding to our four milestones.

\begin{theorem}
Let $\alpha$ be a positive integer,  and let $c>1$ be any integer constant. 
\begin{enumerate}
\item
Consider any leader election algorithm working in time at most $D+\phi +c$, for all graphs of diameter $D$ and election index $\phi$. 
There exist graphs with election index at most $\alpha$ such that this algorithm in these graphs requires advice of size $\Omega( \log \alpha)$.
\item
Consider any leader election algorithm working in time at most $D+c\phi$, for all graphs of diameter $D$ and election index $\phi$. 
There exist graphs with election index at most $\alpha$ such that this algorithm in these graphs requires advice of size $\Omega( \log\log \alpha)$.
\item
Consider any leader election algorithm working in time at most $D+\phi^c$, for all graphs of diameter $D$ and election index $\phi$. 
There exist graphs with election index at most $\alpha$ such that this algorithm in these graphs requires advice of size $\Omega( \log\log\log \alpha)$.
\item
Consider any leader election algorithm working in time at most $D+c^{\phi}$, for all graphs of diameter $D$ and election index $\phi$. 
There exist graphs with election index at most $\alpha$ such that this algorithm in these graphs requires advice of size $\Omega( \log(\log^* \alpha))$.

\end{enumerate}
\end{theorem}

\begin{proof}
The high-level idea of the proof relies on the construction of (ordered) families of graphs with controlled growth of election indices, and with the property that, for any election algorithm working in the prescribed time, graphs from different families must receive different advice. Since the growth of election indices in the constructed sequence of families is controlled (it is linear in part 1), this implies the desired lowered bound on the size of advice. In order to prove that advice in different families must be different, we have to show that otherwise the algorithm would fail in one of these families. The difficulty lies in constructing the families of graphs in such a way as to confuse the hypothetical algorithm in spite of the fact that, as opposed to the situation in Theorems \ref{lb1}
and \ref{lb>1}, the algorithm has now a lot of time: nodes can see all the differences in augmented truncated views of other nodes. In this situation, the way to confuse the algorithm is to make believe two nodes
that they are in a graph with a smaller diameter and thus have to stop early, outputting a path to the leader, while in reality they are in a graph of large diameter, and each of them has seen less than ``half'' of the graph, which results in outputting by each of them a path to a different leader. These nodes are fooled because their augmented truncated views in the large graph at the depth requiring them to stop in the smaller graph are the same as in this smaller graph. In order to assure this, parts of the smaller graph must be carefully reproduced in the large graph, which significantly complicates the construction and the analysis.  

We give the proof of part 1 and then show how to modify it to prove the three remaining parts.        
For any integer $z\geq 4$, we first define the following graph of size $z+2$, called a $z$-{\em lock}{, cf. Fig.~\ref{fig:f3}}. 

\begin{figure}[httb!]
	\begin{center}
	\includegraphics[width=0.4\textwidth]{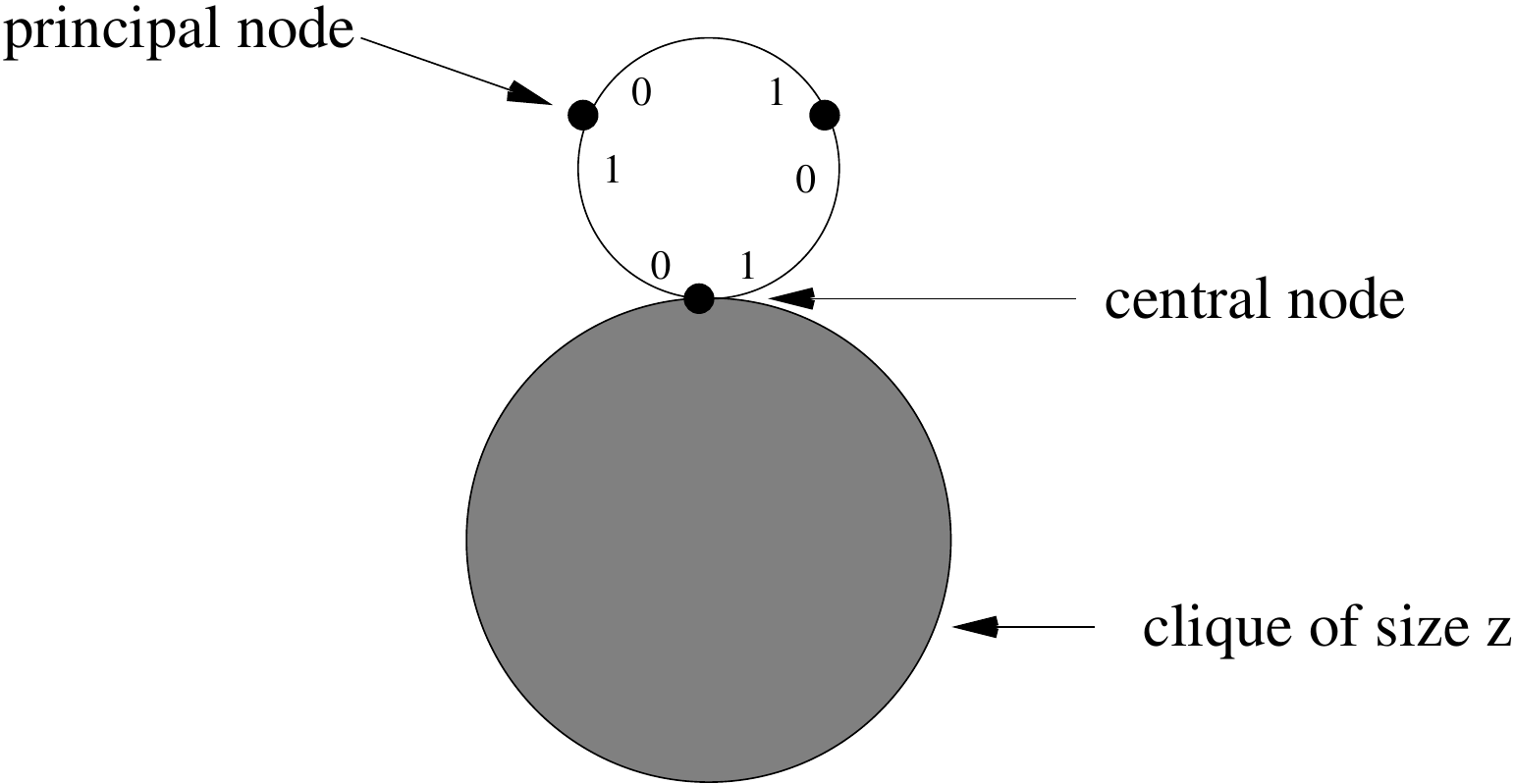}
	\caption{A representation of a $z$-lock.}
	\label{fig:f3}
	\end{center}
\end{figure}

Take a cycle of size 3, with port numbers $0,1$ in clockwise order at each node, and attach to one of the nodes of the cycle a clique of size $z$ by identifying this node with one of the nodes of the cycle. The port numbers in the clique are assigned arbitrarily. Let $w$ be the only node of degree $z+1$ in a $z$-lock. Call it the {\em central node} of the lock. The node $v$ of the cycle such that  the port at $w$ corresponding to the edge $\{w,v\}$  is 0, is called the {\em principal node} of the lock. 
 

Let $G_1$ and $G_2$ be disjoint graphs. We say that a graph $G$ is of the form $G_1*G_2$ {(cf. Fig.~\ref{fig:f4})}, if and only if, {there exist two nodes, a node $a$ from $G_1$ and a node $b$ from $G_2$, such that} the graph $G$ results from $G_1$ and $G_2$ by adding the edge $\{a,b\}$. We define similarly a graph of the form $G_1*G_2*\dots *G_r$. 

\begin{figure}[httb!]
	\begin{center}
	\includegraphics[width=0.6\textwidth]{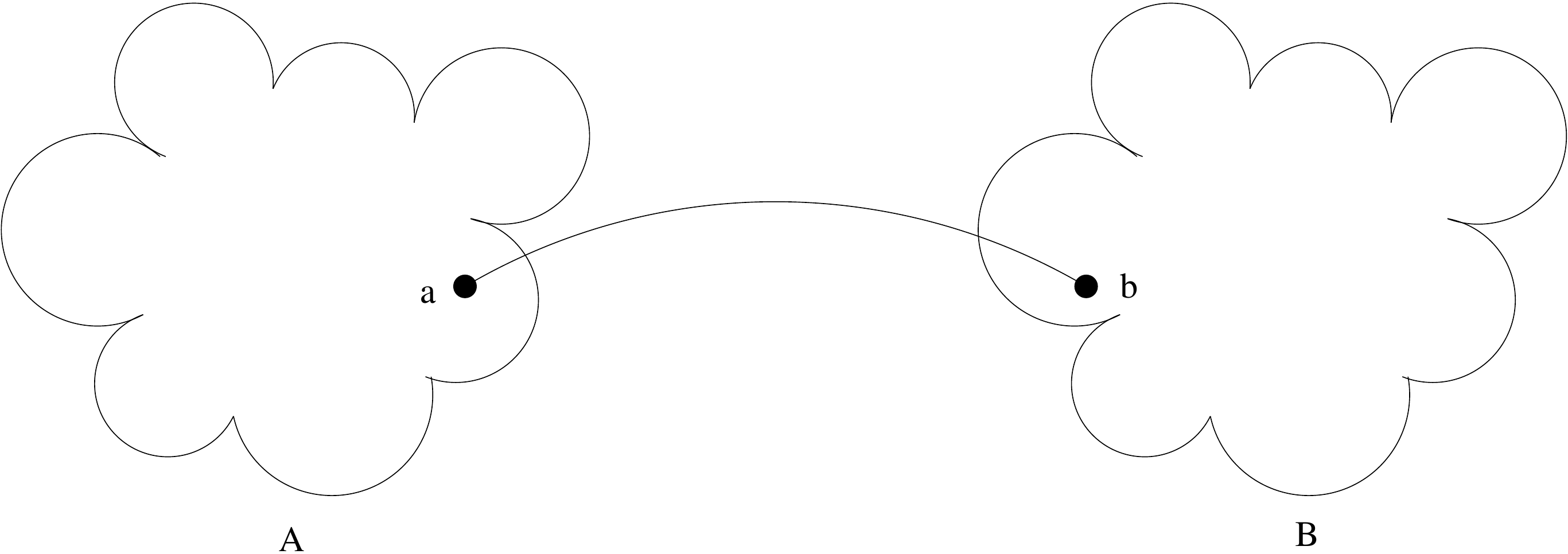}
	\caption{A representation of the graph $A*B$.}
	\label{fig:f4}
	\end{center}
\end{figure}
   
For any integer $x$, denote $A(x,c)=x +c$, $B(x,c)=cx +2x+1$, and $R(x)=x$. We use this notation to be able to derive the proofs of parts 2, 3, and 4 of the theorem from the proof of part 1, by suitably changing  only the definitions of functions $A$, $B$, and $R$. 

Fix a leader election algorithm $\cA$ working in time at most $D+A(\phi, c)$ for all graphs of diameter $D$ and election index $\phi$.
In part 1, $A(\phi, c)=\phi+c$, hence the time of the algorithm is $D+\phi+c$, as supposed. 
Let $k^*$ be such that $B(k^*,c) \leq \alpha<B(k^*+1,c)$. We will construct, by induction on $k$,
ordered families (i.e., sequences) $\cT_0,\dots, \cT_k$ of graphs, for $k \leq k^*$, satisfying the following properties.

\begin{enumerate}
\item
Any graph $G$ of any sequence $\cT_i$ can be unambiguously represented in the form $L_1\ast M \ast L_2$, where 
$L_i$ is a $z_i$-lock, for $i=1,2$, where $z_1 < z_2$. 

$L_1$ is called the {\em left} lock of the graph, 
$L_2$ is called the {\em right} lock of the graph and $M$ is called the {\em central part} of the graph. The principal node of $L_1$ is called the {\em left} principal node of the graph, and 
the principal node of $L_2$ is called the {\em right} principal node of the graph. 
\item
For any $i \leq k$,  any indices $j_1<j_2$, and any graphs $G_{j_1}$ and $G_{j_2}$ from $\cT_i$ ,  the size of the right lock of $G_{j_1}$ is smaller than the size of the left lock of $G_{j_2}$.
\item
All nodes of all graphs of any sequence $\cT_i$ have degree at least 2.
\item
For any $i\leq k$, the diameter of all graphs of the sequence $\cT_i$ is the same.
\item
For any $i < j$, the diameter of graphs from $\cT_i$ is smaller than the diameter of  graphs from $T_j$.  
\item
For any $i\leq k$, the advice used by algorithm $\cA$ for all graphs of the sequence $\cT_i$ is the same.
\item
For any $i < j$, the advice used by algorithm $\cA$ for  graphs of the sequence $\cT_i$ is different from the advice used by algorithm $\cA$ for graphs of the sequence $\cT_j$. 
\item
For any $i\leq k$, and any graph $G$ from $\cT_i$, the election index of $G$ is at most $B(i,c)$.
\item
For any $i < j$, for any graph $G$ from $\cT_j$, there exist graphs $G' \neq  G''$ from $\cT_i$, such that 
the augmented truncated view
$\cB^{D+A(B(i,c),c)}(v)$ in $G$ is equal to the augmented truncated view $\cB^{D+A(B(i,c),c)}(v')$ in $G'$, and
the augmented truncated view $\cB^{D+A(B(i,c),c)}(w)$ in $G$ is equal to the augmented truncated view $\cB^{D+A(B(i,c),c)}(w'')$ in $G''$, where $D$ is the diameter of graphs in $\cT_i$, $v$ is the left principal node of $G$,  
$v'$ is the left principal node of $G'$, $w$ is the right principal node of $G$, and  $w''$ is the right principal node of $G''$. 
\item
For any graph $G$ of any sequence $\cT_i$, the distance between the left principal node of $G$ and the right principal node of $G$ is equal to the diameter of $G$.
\item
For any graph $G$ of any sequence $\cT_i$, the diameter of $G$ is at least $A(\alpha,c)+4$.
\item
For any $i\leq k$, there are $(2\alpha)^{\alpha -i}$ graphs in the sequence $\cT_i$.
\item
For any $i\leq k$,  any graphs $G' \neq G''$ from $\cT_i$, any node $u'$ from $G'$, and any node $u''$ from $G''$, the augmented truncated views
$\cB^{B(i,c)}(u')$ and  $\cB^{B(i,c)}(u'')$ are different. 
\end{enumerate}

We first prove that, given sequences $\cT_0,\dots, \cT_{k^*}$, with the above properties, we can prove our result. By properties 6 and 7, there exist $k^*$ graphs that receive
different pieces of advice. By property 8, election indices of these graphs are all at most $B(k^*,c) \leq \alpha$. By definition, $k^* \in \Omega(R(\alpha))$. Hence
there exists a graph with election index at most $\alpha$ that requires advice of size $\Omega(\log(R(\alpha)))$. In part 1, $R(\alpha)=\alpha$, hence we get the required lower bound $\Omega (\log \alpha)$.
In order to complete the proof of part 1, it remains to construct sequences $\cT_0,\dots, \cT_{k^*}$, with the above properties. (Note that we used only properties 6, 7 and 8,
but the remaining properties are necessary to carry out the inductive construction.)

We proceed with the construction of sequences $\cT_0,\dots, \cT_k$, for $k \leq k^*$ of graphs, by induction on $k$. For $k=0$ we first construct a sequence $\cS_0$ of graphs. The sequence $\cS_0$ consists of $s_0=2^{\alpha} \cdot \alpha ^{\alpha +1}$ graphs $G_i$, for $0 \leq i \leq s_0-1$  defined as follows. For $0 \leq i \leq s_0-1$ define  $x_i=4+2i(\alpha +c+2)+i$. To construct the graph $G_i$ take  an $x_i$-lock with the node $u$ of degree $x_i+1$, which will be the left lock of this graph, and  an $(x_i+2(\alpha +c +2))$-lock with the node $v$ of degree $x_i+2(\alpha +c +2)+1$, which will be the right lock of this graph. Join nodes $u$ and $v$ by a chain of length $\alpha +c +2$ with internal nodes $w_1,w_2,\dots, w_{\alpha+c+1}$, where 
$w_1$ is adjacent to $u$ and  $w_{\alpha+c+1}$ is adjacent to $v$. Attach a clique of size $x_i+2j$ to node $w_j$ by identifying one of the nodes of the clique with $w_j$. All port numbers in locks remain unchanged and all port numbers outside of locks are assigned arbitrarily. Finally, remove all node labels.
This completes the construction of graph $G_i$. We have given its unambiguous representation in the form required in property 1. This completes the construction of the sequence  $\cS_0$, from which the subsequence $\cT_0$ will be extracted. {A representation of a graph from $\cS_0$ is given in Fig.~\ref{fig:f5}.}

\begin{figure}[httb!]
	\begin{center}
	\includegraphics[width=0.6\textwidth]{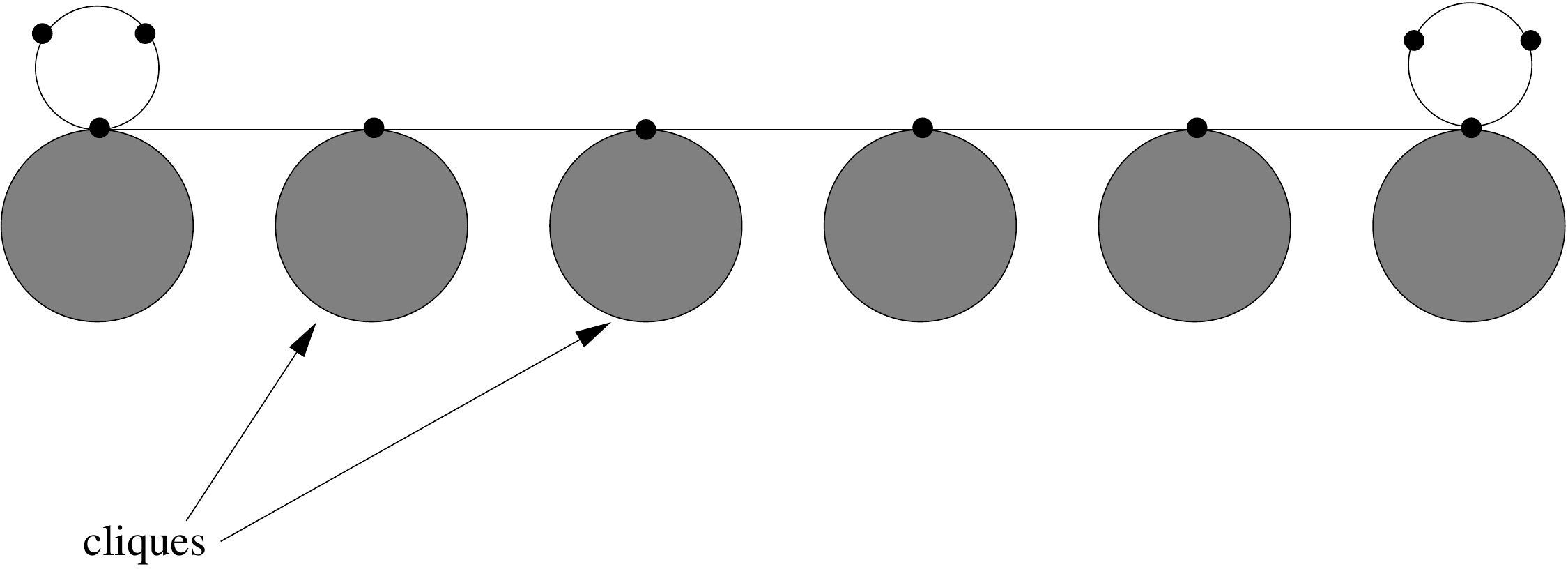}
	\caption{A representation of a graph from $\cS_0$.}
	\label{fig:f5}
	\end{center}
\end{figure}

Before continuing the construction we prove the following claim.

\begin{claim}\label{i1}
The election index of all graphs in $\cS_0$ is 1.
\end{claim}

To prove the claim it is enough to show that $\cB^1(w')\neq \cB^1(w'')$, for any nodes $w'\neq w ''$ in any graph of $\cS_0$.
By construction, nodes
$u,w_1,w_2,\dots, w_{\alpha+c+1},v$ in the chain have unique degrees, and every node of the graph is at distance at most 1 from one of them.
If $w'$ and $w''$ are at distance exactly 1 from two distinct nodes of this chain, then $\cB^1(w')\neq \cB^1(w'')$ because nodes in the chain have different degrees.
If  $w'$ and $w''$ are at distance exactly 1 from the same node $r$ of the chain then $\cB^1(w')\neq \cB^1(w'')$ because port numbers at $r$ corresponding to
edges $\{r,w'\}$ and $\{r,w''\}$ must be different. Finally, if one of the nodes $w'$ and $w''$ is in the chain then $\cB^1(w')\neq \cB^1(w'')$
because, {by construction}, they must have different degrees. This concludes the proof of the claim.

Next we define a subsequence $\cT_0$ of the sequence $\cS_0$ with the following two properties: the size of $\cT_0$ is $(2\alpha)^{\alpha}$, and all graphs in
$\cT_0$ receive the same advice. We can assume that such a subsequence exists because otherwise the number of distinct pieces of advice received by graphs from
$\cS_0$ would be at least $\alpha$ (because $\cS_0$ has size $2^{\alpha} \cdot \alpha ^{\alpha +1}$)  which would prove part 1 of our theorem, in view of Claim~\ref{i1}.
The verification that the sequence $\cT_0$ satisfies properties 1--12 is immediate, and property 13 is implied by the fact that nodes in chains of two different graphs from
$\cT_0$ have different degrees, and every node is at distance at most 1 from one of them.

Assume by induction that sequences  $\cT_0,\dots, \cT_k$ of graphs, for $k < k^*$, have been already constructed, and that they satisfy properties 1 -- 13. We now construct the sequence 
$ \cT_{k+1}$ of graphs. Let $\cT_k=\{H_1,\dots, H_{t_k}\}$, where $t_k=(2\alpha)^{\alpha -k}$. We first define the sequence $\cS_{k+1}=\{Q_1,\dots , Q_{t_{k/2}}\}$ of  graphs, where
$Q_i$ is the result of the {\em merge operation} of $H_{2i}$ and $H_{2i+1}$. 

In order to define the merge operation of graphs $H_{2i}$ and $H_{2i+1}$ from $\cT_k$,  we first define the {\em pruned view} of a node $u$ in any graph $G$. Let $p_1,\dots ,p_t$ be any port numbers at node  $u$. The pruned view of $u$ at depth $\ell$ with respect to ports $p_1,\dots ,p_t$ is a tree of height $\ell$ rooted at $u$, that is denoted by $\cP\cV_G(u,\{p_1,\dots ,p_t\},\ell)$ and is defined by induction on $\ell$. For $\ell=0$ we define $\cP\cV_G(u,\{p_1,\dots ,p_t\},0)=\{u\}$. Suppose that $\cP\cV_G(u,\{p_1,\dots ,p_t\},\ell)$ is already defined, for any node $v$ and with respect to any port numbers at $v$. Let $T$ be the tree of height 1 rooted at $u$ whose leaves are all neighbors $v$ of $u$ except the neighbors $w_i$ such that the 
port at $u$ corresponding to the edge {$\{u,w_i\}$} is $p_i$, for $1 \leq i \leq t$. Assign at all edges of this tree the same port numbers as in graph $G$. 
Let $v_1,\dots ,v_r$ be the leaves of $T$. Let $q_i$ be the port number at $v_i$ corresponding to the edge
$\{v_i,u\}$. Attach to $v_i$ the tree $\cP\cV_G(v_i,\{q_i\},\ell)$ by identifying $v_i$ with the root of this tree. The resulting tree is $\cP\cV_G(v,\{p_1,\dots ,p_t\},\ell+1)$.
Notice that, as opposed to the truncated view at depth $\ell$, the pruned view at depth $\ell$ does not contain repeated port numbers at any node, and hence can be used as a building block for graph constructions. We will use pruned views in this way in the sequel.

The following claim will enable us to replace a subgraph of a graph by the pruned view of an articulation node without changing its augmented truncated view.

\begin{claim}\label{art}
Let $u$ be an articulation node of a graph $G$, and let $p_1,\dots ,p_t$ be the port numbers at this node, such that the removal of edges corresponding to these ports
disconnects the graph into at least two connected components. Let $G'$ be the connected component containing $u$, after removal of these edges. Let $G^*$ be the graph resulting from $G$ by replacing the subgraph $G'$ by the pruned view $\cP\cV_G(u,\{p_1,\dots ,p_t\},\ell)$, for positive $\ell$.
Then the augmented truncated view $\cB^{\ell-1}(u)$ is the same in graphs $G$ and $G^*$, and the augmented truncated view $\cB^{d+\ell-1}(v)$, for all nodes $v$ outside of $G'$, is
the same in graphs $G$ and $G^*$, where $d$ is the distance between $u$ and $v$ in $G$.
\end{claim}

In order to prove the claim, it suffices to show that $\cB^{\ell-1}(u)$ is the same in graphs $G$ and $G^*$. The other part follows from this because $u$ is an articulation node.  In order to prove that $\cB^{\ell-1}(u)$ is the same in graphs $G$ and $G^*$, it is enough to prove that  $\cV^{\ell}(u)$ is the same in graphs $G$ and $G^*$. 
This is equivalent to the fact that sequences of port numbers of even length at most $2\ell$, corresponding to paths of $\cV^{\ell}(u)$, are the same in graphs $G$ and $G^*$.
Define a {\em normal sequence} of port numbers as a sequence $(q_1,q'_1,\dots, q_j,q'_j)$, such that $q_{i+1}\neq q'_i$, for any $i<j$. 
A normal sequence corresponds to  
a path in a graph, whose consecutive edges are never equal.
For any $\ell'\leq \ell$, let $\cV^{\ell'}$ be the set of sequences, of even length at most $2\ell'$, of port numbers, corresponding to the paths in 
the view  $\cV^{\ell'}(u)$ in $G$, let  $\cV^{\ell' *}$  be the set of sequences, of even length at most $2\ell'$, of port numbers, corresponding to the paths in 
the view  $\cV^{\ell'}(u)$ in $G^*$, let $\cW^{\ell'}$ be the set of normal sequences, of even length at most $2\ell'$, of port numbers, corresponding to the paths in $\cV^{\ell'}$, and let $\cW^{\ell' *}$ be the set of normal sequences, of even length at most $2\ell'$, of port numbers, corresponding to the paths in $\cV^{\ell' *}$. Note that $\cV^{\ell'}=\cV^{\ell' *}$ if and only if 
$\cW^{\ell'}=\cW^{\ell' *}$.

We prove by induction on $\ell'$, such that $1\leq \ell'\leq \ell$, that $\cW^{\ell'}=\cW^{\ell' *}$.
For $\ell'=1$ this follows {from the fact that, by construction, $\cV^1(u)$ is the same in both graphs}. Hence if $\ell=1$ the proof is done. So consider that $\ell\geq 2$. To prove the inductive step, suppose that,  for some integer $r$ such that $2\leq r\leq l$,  we have
$\cW^{\ell'}=\cW^{\ell' *}$, for {all} $\ell'< r$. We have to prove that $\cW^{r}=\cW^{r*}$.
Let $\pi$ be a sequence from $\cW^{r}$.
The definition of $\cP\cV_G(u,\{p_1,\dots ,p_t\},r)$ implies that
the sets of normal sequences of even length at most $2r$ starting at $u$ and having the first port number outside of $\{p_1,\dots ,p_t\}$  are the same in $\cV^{r}$ and
in $\cV^{r*}$.
If the first term of $\pi$ is
outside of $\{p_1,\dots ,p_t\}$ , then $\pi$ is in $\cW^{r*}$ by the preceding statement. Consider a sequence $\pi$ in $\cW^{r}$ whose first port number is in $\{p_1,\dots ,p_t\}$.
There are two cases. If the node $u$ does not appear again in the path in $G$ corresponding to $\pi$, then $\pi$ is also in $\cW^{r*}$ by construction. Otherwise, let $\pi'$ be the shortest  prefix of $\pi$, such that, in the corresponding path in $G$, the node
 $u$ appears exactly twice, i.e., this path is a loop ending at $u$. (The prefix $\pi'$ has necessarily even length). The path in $G^*$ starting at $u$ and corresponding to $\pi'$ ends at node $u$ as well, by construction. 
 Let $\pi''$ be the part of $\pi$ after removal of $\pi'$. The length of $\pi''$ is even and smaller than $r$. By the inductive hypothesis, $\pi''$ is in $\cW^{(r-1)*}$. Hence the sequence
 $\pi$, which is the concatenation of $\pi'$ and $\pi''$ belongs to  $\cW^{r*}$. This proves the inclusion $\cW^{r}\subseteq \cW^{r*}$. The proof of the other inclusion
 is similar.
 Hence, for all $\ell'$ such that $1\leq \ell'\leq \ell$, we have $\cW^{\ell'}=\cW^{\ell' *}$. As noted  before, this implies that $\cV^{\ell}(u)$ is the same in graphs $G$ and $G^*$ and hence $\cB^{\ell-1}(u)$ is also the same in these graphs. This proves the claim.


We are now ready to define the merge operation of graphs $H_{2i}$ and $H_{2i+1}$ from $\cT_k$. {The result of such a merge operation is illustrated in Fig.~\ref{fig:fig7}}.
By property~1, the graph $H_{2i}$ is of the form $L_1 \ast M' \ast  L_2$, where $L_1$ is its left lock and $L_2$ is its right lock,  and $H_{2i+1}$ is of the form $L_3 \ast M'' \ast  L_4$, where $L_3$ is its left lock and $L_4$ is its right lock. The result of the merge operation of graphs $H_{2i}$ and $H_{2i+1}$
is the graph $Q$ of the form $L_1 \ast N \ast  L_4$, where  $L_1$ is its left lock, $L_4$ is its right lock, and $N$ is defined below.

First define the transformations $T(L_2)$ and $T(L_3)$ of locks $L_2$ and $L_3${, cf. Fig.~\ref{fig:fig6}}. Suppose that $L_2$ is a $z$-lock and let $u$ be its central node. 
Replace the 3-cycle of the lock by the pruned view $\cP\cV_{H_{2i}}(u,\{2,\dots, z+1\},B(k+1,c))$. 
More precisely, remove the two nodes adjacent to $u$ in this cycle, together with the incident edges, and attach the above pruned view at $u$, by identifying $u$ with
the root of this pruned view.
Let $m_1,\dots, m_t$ be the leaves of $\cP\cV_{H_{2i}}(u,\{2,\dots, z+1\},B(k+1,c))$. Let $x$
be the largest degree of any of the previously constructed graphs. {For all $1\leq f \leq t$,} attach a clique of size $x+4f$ to the leaf $m_f$ by identifying one node of this clique with this leaf.
This concludes the description of $T(L_2)$. The central node of $L_2$ is also called the central node of {$T(L_2)$}. 
The transformation $T(L_3)$ is defined similarly, with $L_2$ replaced by $L_3$, {$t$ replaced by $t'$, $x+4f$ replaced by $x+4f+4t+4$} and $H_{2i}$ replaced by $H_{2i+1}$.

\begin{figure}[!htbp]
\begin{center}
  \begin{minipage}[t]{0.2\linewidth}
    \centering
	\includegraphics[width=0.8\textwidth]{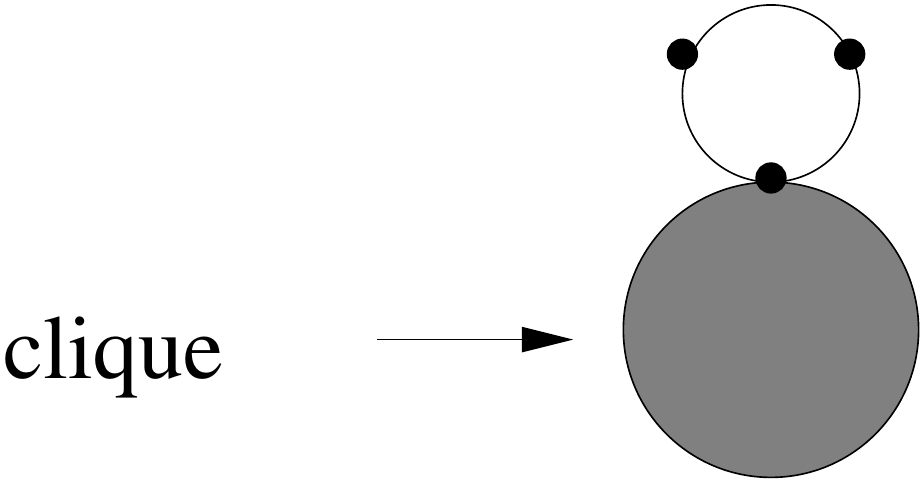}\\
    {\footnotesize ($a$) A $z$-lock $L$}
  \end{minipage}%
  \begin{minipage}[t]{0.8\linewidth}
    \centering
	\includegraphics[width=0.9\textwidth]{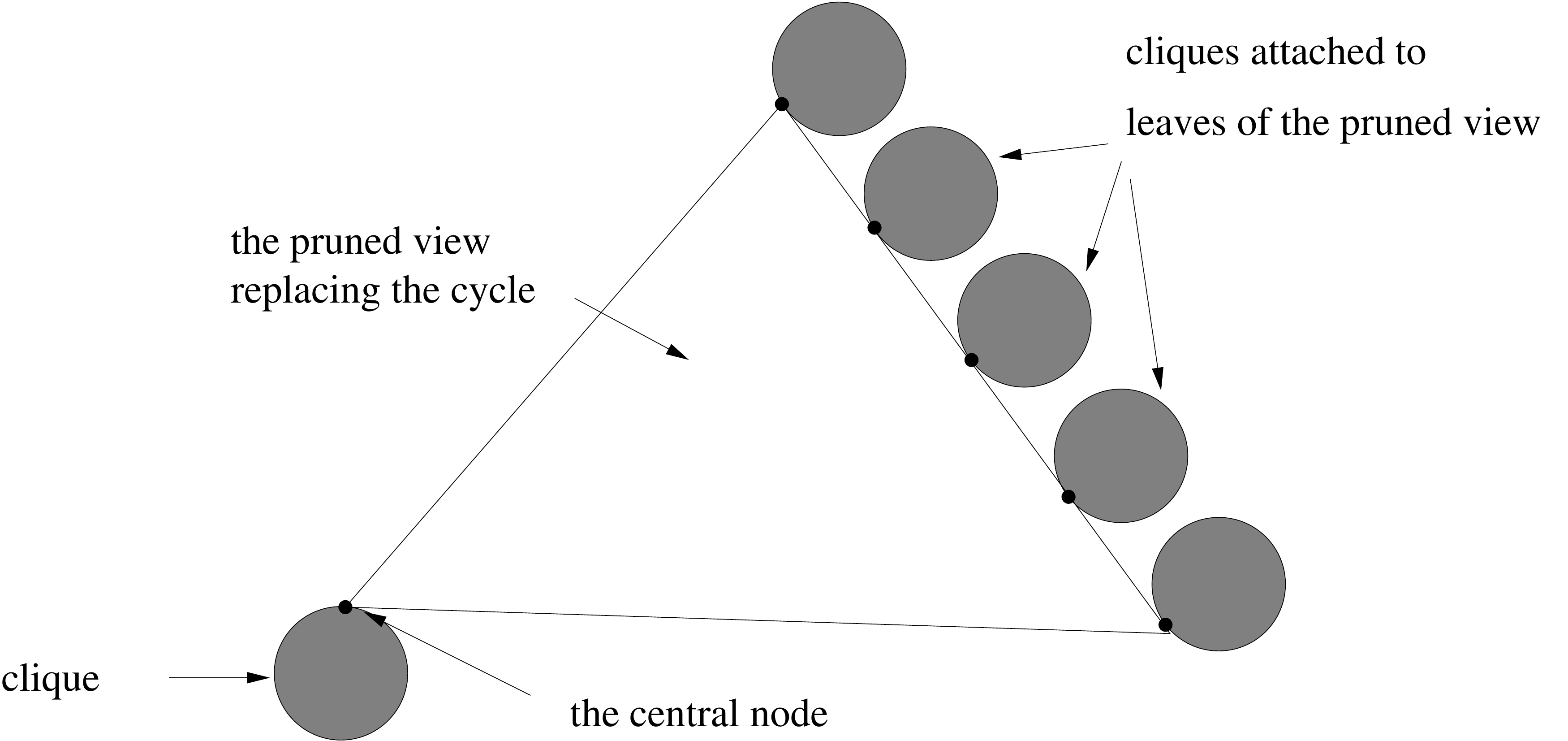}\\
    {\footnotesize ($b$) The graph $T(L)$.} 
  \end{minipage}
\end{center}
 \caption{A $z$-lock and its transformation. \label{fig:fig6}}
\end{figure}

Next define a subgraph $X$. Let $n$ be the maximum size of all previously defined graphs. Let $g_1,\dots, g_{2n}$ be nodes yet unused in the construction,  forming a chain. Let {$y$} be the largest degree of $T(L_3)$. 
{For all $1\leq f\leq 2n$, attach a clique of size $y+4f$} to node  $g_f$ by identifying one node of this clique with it. 
All attached cliques are pairwise disjoint and consist of nodes not used before. 
Assign all port numbers arbitrarily.
This completes the construction of the subgraph $X$.
Attaching larger and larger cliques to different nodes will guarantee properties 8 and 13 for graphs of $\cT_{k+1}$.

Finally, the subgraph $N$ of the graph $Q$ under construction is defined as follows. Let $a$ be {the node with the highest degree in $T(L_2)$}, and let $b$ be {the
 node with the highest degree in $T(L_3)$}. Let $c'$ be the node in $M'$ and let $b'$ be the node in $L_2$ such that the edge $\{c',b'\}$ attaches $M'$ to $L_2$ in $H_{2i}$. 
Let $c''$ be the node in $M''$ and let $b''$ be the node in $L_3$ such that the edge $\{c'',b''\}$ attaches $M''$ to $L_3$ in $H_{2i+1}$.
Note that the node $b'$ from $L_2$ remains in $T(L_2)$ and the node $b''$ remains in $T(L_3)$. {In fact, by construction, $b'$ is the central node of $T(L_2)$, and $b''$ is the central node of $T(L_3)$.}
Attach $M'$ to $T(L_2)$ by edge $\{c',b'\}$ (keeping the port numbers), 
attach $a$ to $g_1$ by a new edge with smallest port numbers not yet used at each endpoint, attach $g_{2n}$ to $b$ by a new edge with smallest port numbers not yet used at each endpoint, and attach $T(L_3)$ to
$M''$  by edge $\{b'',c''\}$ (keeping the port numbers). The resulting graph is $N$. In the graph $Q$ of the form $L_1 \ast N \ast  L_4$, the subgraph $N$ is attached to $L_1$
by the edge that attached $M'$ to $L_1$ in $H_{2i}$, and $N$ is attached to $L_4$ by the edge that attached $M''$ to $L_4$ in $H_{2i+1}$. This concludes the description of  the graph $Q$. {Fig.~\ref{fig:fig7} illustrates the merge operation leading to such a graph.}
We have given its unambiguous representation in the form required in property 1. This completes the construction of the sequence  $\cS_{k+1}$, from which the subsequence $\cT_{k+1}$ will be extracted.

\begin{figure}[!htbp]
\begin{center}
  \begin{minipage}[t]{0.5\linewidth}
    \centering
	\includegraphics[width=0.8\textwidth]{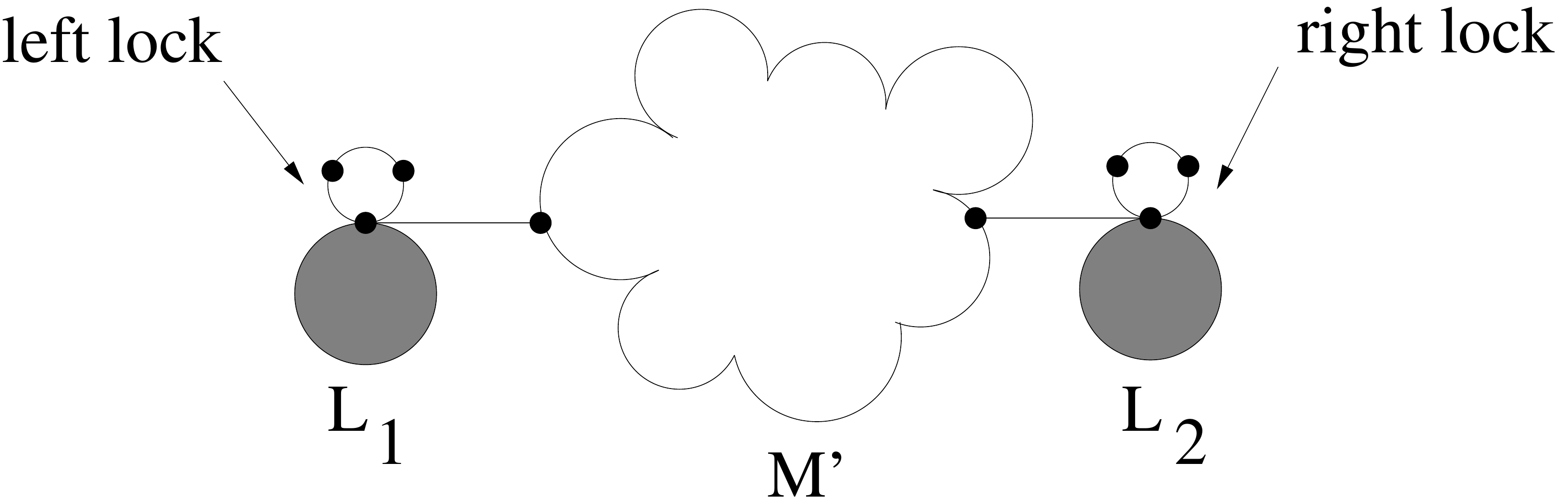}\\
    {\footnotesize ($a$) the graph $H_{2i}$.}
  \end{minipage}%
  \begin{minipage}[t]{0.5\linewidth}
    \centering
	\includegraphics[width=0.8\textwidth]{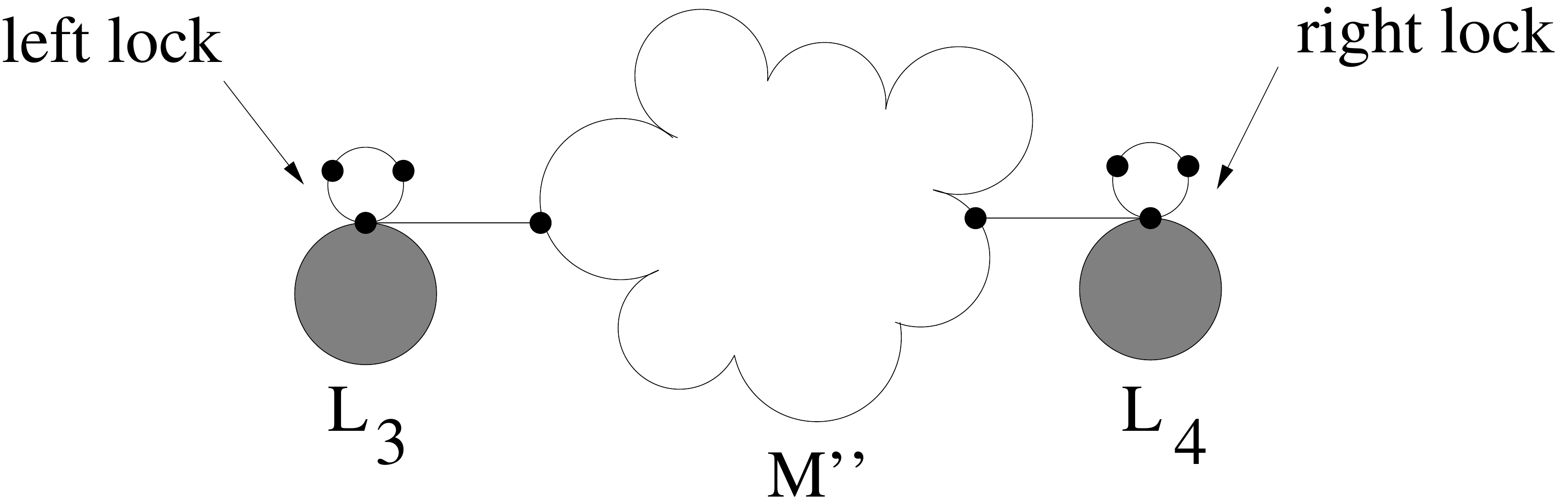}\\
    {\footnotesize ($b$) the graph $H_{2i+1}$.} 
  \end{minipage}
  \bigskip \\
  \begin{minipage}[t]{1\linewidth}
    \centering
	\includegraphics[width=1\textwidth]{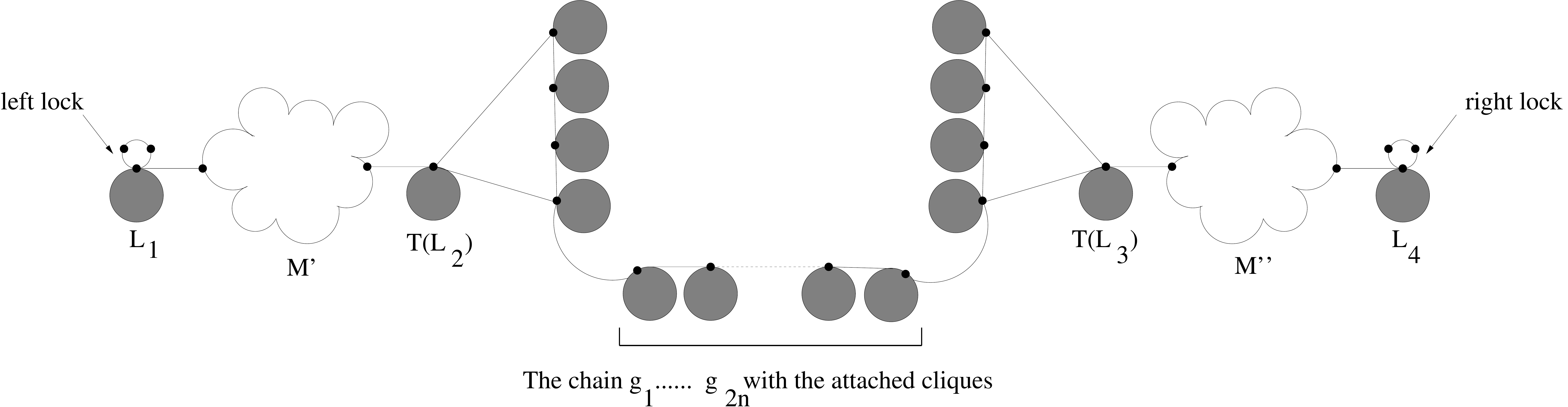}\\
    {\footnotesize ($c$) the result of the merge operation of $H_{2i}$ and $H_{2i+1}$.}
  \end{minipage}
\end{center}
 \caption{Illustration of the merge operation of $H_{2i}$ and $H_{2i+1}$.\label{fig:fig7}}
\end{figure}

Before defining the subsequence $\cT_{k+1}$ of $\cS_{k+1}$, we prove the following four claims. The first claim is implied by property 3 of graphs in $\cT_k$, that holds by the inductive assumption. Since graphs in $\cT_k$  do not have nodes of degree one, all branches of a pruned view of any node can be extended indefinitely.
More precisely, we have:

\begin{claim}\label{stretch}
For any graph $G$ in $\cT_k$,  any node $u$ of $G$, and any non-negative integer $\ell$, all leaves in $\cP\cV(u, P, \ell)$, are exactly at distance $\ell$
(in $\cP\cV(u, P, \ell)$) from the root
of $\cP\cV(u, P, \ell)$, provided that the size of  $P$ is strictly smaller than the degree of $u$ in $G$. 
\end{claim}

The next claim states property 1 for graphs in  $\cS_{k+1}$.  It follows from the construction of this class and from property 2 for the sequence $\cT_k$. 

\begin{claim}\label{prop1-ind}
Any graph $Q$ from $\cS_{k+1}$ can be unambiguously represented in the form $L_1\ast M \ast L_2$, where 
$L_i$ is a $z_i$-lock, for $i=1,2$, where $z_1 < z_2$. 
\end{claim}

The next claim states property 8 for graphs in $\cS_{k+1}$. 

\begin{claim}\label{prop8}\label{index-i} 
For any graph $Q$ from $\cS_{k+1}$, the election index of $Q$ is at most $B(k+1,c)$.
\end{claim}

In order to prove the claim, it is enough to show that for any distinct nodes $u$ and $v$ from $Q$, we have $\cB^{B(k+1,c)}(u)) \neq \cB^{B(k+1,c)}(v))$.
Let $Q$ be the result of the merge operation of graphs $H_{2i}$ and $H_{2i+1}$ from $\cT_k$. Hence $Q$ is of the form
$L_1*M'*T(L_2)*X*T(L_3)*M''*L_4$, where $L_1$ is the left lock of $H_{2i}$, $M'$ is the central part of $H_{2i}$, $M''$ is the central part of $H_{2i+1}$, and   $L_4$ is  the right lock of $H_{2i+1}$. Let $Y$ be the subgraph of $Q$ {of} the form $T(L_2)*X*T(L_3)$.
Let $Z'$ be the subgraph of $Q$ of the form $L_1*M'$, and let $Z''$ be the subgraph of $Q$ of the form $M''*L_4$.
Let $P(T(L_i))$, for $i=2,3$, be the set of nodes of the subclique of $T(L_i)$ attached to its central node, excluding the central node itself. 
Let $W$ be the set of all nodes of $Q$ at which cliques are attached during the merge operation of graphs $H_{2i}$ and $H_{2i+1}$.
By construction, each of the nodes of $W$ has a unique degree with respect to all graphs from $\cS_{k+1}$ and in particular,  a unique degree in $Q$.
By Claim \ref{stretch}, all nodes from $Y$ are at distance at most $B(k+1,c)$ from some node of $W$.
Consider three cases.

\noindent
Case 1. Both nodes $u$ and $v$ are from $Y$, outside of $P(T(L_2))\cup P(T(L_3))$.

In each of $\cB^{B(k+1,c)}(u)$ and $\cB^{B(k+1,c)}(v)$, there exists a node with a unique degree in $Q$.
If there is a node in one of these views, such that a node of the same degree does not appear in the other view, then $\cB^{B(k+1,c)}(u) \neq \cB^{B(k+1,c)}(v)$.
Otherwise, there is a node $w$ with a unique degree in $Q$ that appears in both views. Since {the two} sequences of port numbers corresponding to paths from 
$u$ to $w$ and from $v$ to $w$ must be {different}, this implies $\cB^{B(k+1,c)}(u) \neq \cB^{B(k+1,c)}(v)$.

\noindent
Case 2. One of the nodes $u$ or $v$ is in $Y$, outside of $P(T(L_2))\cup P(T(L_3))$, and the other is either in $Z'$, or $Z''$ or in $P(T(L_2))\cup P(T(L_3))$.

Without loss of generality, assume that $u$ is in $Y$, outside of $P(T(L_2))\cup P(T(L_3))$, and $v$ is either in $Z'$, or $Z''$ or in $P(T(L_2))\cup P(T(L_3))$.
In view of the Claim \ref{stretch}, the central nodes of $T(L_2)$ and $T(L_3)$ are at distance at least $B(k+1,c)$ from all nodes in $W$. By construction, the node
$v$ is at distance at least $B(k+1,c)+1$ from all nodes in $W$. Hence, there is a node in $\cB^{B(k+1,c)}(u)$, such that no node of the same degree appears in
$\cB^{B(k+1,c)}(v)$. Hence $\cB^{B(k+1,c)}(u) \neq \cB^{B(k+1,c)}(v)$.

\noindent
Case 3. Each of $u$ and $v$ is either in $Z'$, or $Z''$ or in $P(T(L_2))\cup P(T(L_3))$. 

Consider the subcase where $u$ is in $Z'$ or in $P(T(L_1))$ and $v$ is in $Z''$ or in $P(T(L_2))$. 
By Claim \ref{art}, $\cB^{B(k+1,c)}(u)$ is the same in $H_{2i}$ and in $Q$. Hence $\cB^{B(k,c)}(u)$ is the same in $H_{2i}$ and in $Q$.
Similarly, $\cB^{B(k,c)}(v)$ is the same in $H_{2i+1}$ and in $Q$. By property 13 for the sequence $\cT_k$, we have  $\cB^{B(k,c)}(u) \neq \cB^{B(k,c)}(v)$, 
hence $\cB^{B(k+1,c)}(u) \neq \cB^{B(k+1,c)}(v)$.
The other subcase is when both $u$ and $v$ are either in $Z'$ or in $P(T(L_1))$, or they are both in $Z''$ or in $P(T(L_2))$.
The argument in this subcase is similar as above.
This concludes the proof of the claim.

The last claim states property 13 for $\cS_{k+1}$. It follows from this property for $\cT_k$, using arguments similar to those used to prove Claim \ref{index-i}.

\begin{claim}\label{prop13}
For any graphs $G' \neq G''$ from $\cS_{k+1}$, any node $u'$ from $G'$ and any node $u''$ from $G''$, the augmented truncated views
$\cB^{B(k+1,c)}(u')$ and  $\cB^{B(k+1,c)}(u'')$ are different. 
\end{claim}

Finally we define a subsequence $\cT_{k+1}$ of the sequence $\cS_{k+1}$ with the following two properties: 
the size of $\cT_{k+1}$ is $(2\alpha)^{\alpha-k-1}$, and all graphs in
$\cT_{k+1}$ receive the same advice. We can assume that such a subsequence exists. Indeed, first observe that $\cT_k$ has size $2^{\alpha-k} \cdot \alpha ^{\alpha -k}$,
 by property 12 for $\cT_k$, and that the size of $\cS_{k+1}$ is half the size of $\cT_k$ by construction. Hence, the size of  $\cS_{k+1}$ is $2^{\alpha-k-1} \cdot \alpha ^{\alpha -k}$. If a subsequence $\cT_{k+1}$ of the sequence $\cS_{k+1}$ with the above two properties could not be extracted from $\cS_{k+1}$, this would mean
 that the number of distinct pieces of advice received by graphs from $\cS_{k+1}$ would be at least $\frac{|\cS_{k+1}|}{(2\alpha)^{\alpha-k-1}}=\alpha$. This in turn would imply, in view of Claim~\ref{index-i},
 that one of the graphs with election index at most $B(k+1,c)$ would receive advice of size $\Omega(\log \alpha)$
 which would prove part 1 of our theorem (because $k+1 \leq k^*$ and thus $B(k+1,c) \leq \alpha$).
 
 This concludes the construction of the sequence $\cT_{k+1}$. 
 It remains to prove that this sequence has all the properties 1 --13.
 
 Property 1 holds for $\cT_{k+1}$ because it holds for $\cS_{k+1}$  by Claim \ref{prop1-ind}.
 In order to prove property 2 for $\cT_{k+1}$, it is enough to prove it for all graphs in $\cS_{k+1}$. To do this, consider any graphs $Q_i$ and $Q_j$ from $\cS_{k+1}$,
 such that $i<j$. By construction, the graph $Q_i$ is the result of the merge of graphs $H_{2i}$ and $H_{2i+1}$ from $\cT_k$, and the
 graph $Q_j$ is the result of the merge of graphs $H_{2j}$ and $H_{2j+1}$ from $\cT_k$. 
 The right lock
 of $Q_i$ is the right lock $L_i$  of $H_{2i+1}$,  and the left lock of $Q_j$ is the left lock $L_j$ of $H_{2j}$.  Let $L_i$ be a $z_i$-lock, and let $L_j$ be a $z_j$-lock.
 By property 2 for $\cT_k$, we have $z_i <z_j$ because $2j>2i+1$. This proves property 2 for  $\cS_{k+1}$ and hence for $\cT_{k+1}$. In order to prove property 3
 for $\cT_{k+1}$, take a graph $Q$ from this sequence. It is a result of the merge operation of two graphs from $\cT_k$. By property 3 for $\cT_k$, these graphs have
 no nodes of degree 1. The merge operation does not create such nodes. Hence the graph $Q$ does not have nodes of degree 1. Properties 6 and 12 for $\cT_{k+1}$ hold by the definition of this sequence. 
  Properties 5 and 11 for $\cT_{k+1}$ follow from these properties for $\cT_k$ and from the fact that the result of the merge operation of two graphs is a graph of diameter larger than that of each of them. Property 8 for $\cT_{k+1}$ follows from Claim \ref{index-i}, and property 13 for $\cT_{k+1}$ follows from Claim \ref{prop13}.
  
  Properties 4 and 10 for $\cT_{k+1}$ will be proved together as follows. For any graph $Q$ from $\cT_{k+1}$, we first compute the distance between the left principal
  node of $Q$ and the right principal node of $Q$. This distance turns out to be the same  for all graphs in $\cT_{k+1}$. Then we prove that the distance between any 
  two nodes of any graph from $\cT_{k+1}$ is at most this value.
  
  Let $Q$ be the result of  the merge operation of graphs $H'$ and $H''$ from $\cT_k$.
The graph $H'$ is of the form $L_1 \ast M' \ast  L_2$, where $L_1$ is its left lock and $L_2$ is its right lock,  and the graph $H''$ is of the form $L_3 \ast M'' \ast  L_4$, where $L_3$ is its left lock and $L_4$ is its right lock. Let $u$ be the left principal node of $H'$, and let $u'$ be the right principal node of $H'$.
Let $v'$ be the left principal node of $H''$, and let $v$ be the right principal node of $H''$. Note that
$u$ is the left principal node of $Q$, and $v$ is the right principal node of $Q$.
By construction, the graph $Q$ can be represented in the form
$L_1*M'*T(L_2)*X*T(L_3)*M''*L_4$. 
Consider nodes $a$ in $L_1$, $b$ and $c$ in $M'$, $d$ and $e$ in $T(L_2)$, $f$ and $f'$ in $X$, $e'$ and $d'$ in $T(L_3)$,
$c'$ and $b'$ in $M''$, and $a'$ in $L_4$, such that the edge $\{a,b\}$ joins $L_1$ to $M'$, the edge $\{c,d\}$ joins $M'$ to $T(L_2)$, the edge $\{e,f\}$ 
joins $T(L_2)$ to $X$ ,
the edge $\{f',e'\}$ joins $X$ to $T(L_3)$, the edge $\{d',c'\}$ joins $T(L_3)$ to $M''$, and the edge $\{b',a'\}$ joins $M''$ to $L_4$. {A representation of graph $Q$ with the above notation is given in Fig.~\ref{fig:f8}}.

\begin{figure}[httb!]
	\begin{center}
	\includegraphics[width=1\textwidth]{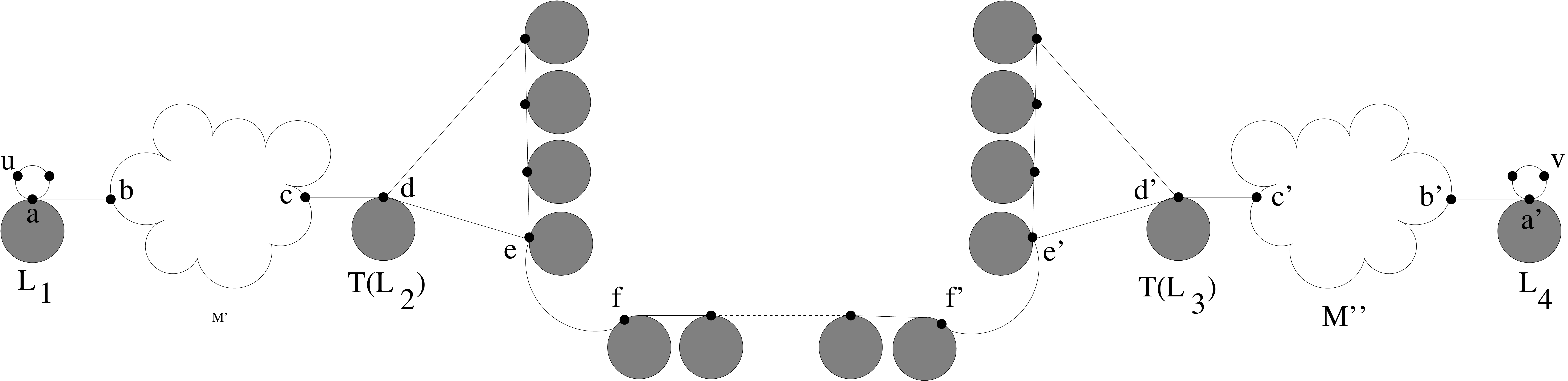}
	\caption{A representation of graph $Q$ with the notations of components used in the proof of Properties 4 and 10.}
	\label{fig:f8}
	\end{center}
\end{figure}

Let $\delta_G(x,y)$ denote the distance between nodes $x$ and $y$ in the graph $G$. We first compute $\delta_Q(u,v)$. 
By property 4 for $\cT_k$, all graphs in this sequence have the same diameter $D$.
By property 10 for $\cT_k$, we have
$\delta_{H'}(u,u')=D$, hence $\delta_Q(u,c)=D-2$, by construction. Similarly, $\delta_{H''}(v,v')=D$, hence $\delta_Q(v,c')=D-2$.
By construction and by Claim \ref{stretch}, we have $\delta_Q(d,e)=B(k+1,c)$. Similarly, $\delta_Q(d',e')=B(k+1,c)$. By construction, $\delta_Q(f,f')=2n-1$,
because $2n-1$ is the length of the chain used in the merge procedure to construct the graph $X$. Hence $\delta_Q(u,v)=2D+2B(k+1,c)+2n-1$.

We now show, again for any graph $Q$ from $\cT_k$, that the distance between any 
  two nodes $x$ and $y$ of $Q$ is at most this value. Consider six cases.
  
  \noindent
  Case 1. $x$ and $y$ are in $L_1$ or $M'$ (respectively in $L_4$ or $M''$).
  
  We give the argument for the first situation. The second one is symmetric. 
  Since $x$ and $y$ are in $L_1$ or $M'$, they are both in $H'$ whose diameter is $D$. Hence $\delta_{H'}(x,y)\leq D$ and hence $\delta_{Q}(x,y)\leq D<\delta_Q(u,v)$.
  
  \noindent
  Case 2. $x$ and $y$ are in $T(L_2)$ (respectively in $T(L_3)$).
  
    We give the argument for the first situation. The second one is symmetric. 
    By construction and by Claim \ref{stretch}, the diameter of $T(L_2)$ is $2B(k+1,c)+2<\delta_Q(u,v)$.
    
     \noindent
  Case 3. $x$ and $y$ are in $X$.
  
  By construction, the diameter of $X$ is $2n+1<\delta_Q(u,v)$. 
  
     \noindent
  Case 4. $x$ is in $T(L_2)$,  and $y$ is in $X$ or in $T(L_3)$ (resp. $x$  is in $T(L_2)$ or in $X$, and $y$ is  in $T(L_3)$).
  
   We give the argument for the first situation. The second one is symmetric. 
   By construction and by Claim \ref{stretch},  the diameter of $X$ is $2n+1$, and the diameter of $T(L_2)$ and $T(L_3)$ is $2B(k+1,c)+2$,
   hence  $\delta_{Q}(x,y)\leq 4B(k+1,c)+2n+6$. On the other hand, we have $\delta_Q(u,v)=2D+2B(k+1,c)+2n-1$. By property 11 for $\cT_k$ we have 
  $2D+2B(k+1,c)+2n-1 \geq 2(A(\alpha, c)+4)+2B(k+1,c)+2n-1$. Since $A(\alpha ,c )\geq B(k+1,c)$, we have $2(A(\alpha, c)+4)+2B(k+1,c)+2n-1 \geq 4B(k+1,c)+2n+7$.
  Hence  $\delta_{Q}(x,y) <  \delta_Q(u,v)$.
  
    \noindent
  Case 5. 
  $x$ is in $L_1$ or $M'$, and $y$ is either in $T(L_2)$ or in $X$ or in $T(L_3)$ (resp. $x$ is in $L_4$ or $M''$, and $y$ is either in $T(L_2)$ or in $X$ or in $T(L_3)$).
  
We give the argument for the first situation. The second one is symmetric. 
First observe that every node in $L_1$ or in $M'$ is at distance at most $D-2$ from $c$. Otherwise, it would be at distance at least $D+1$ in $H'$ from $u'$, which would contradict the fact that $D$ is the diameter of $H'$. Consider three possibilities. If $y$ is in $T(L_2)$ then, in view of the above observation and of the fact that $d$
is at distance at most $B(k+1,c)+1$ from every node in $T(L_2)$, we have $\delta_Q(x,y) \leq D+B(k+1,c) <  \delta_Q(u,v)$. If $y$ is in $X$ then
 $\delta_Q(x,y) \leq D+B(k+1,c)+2n+1<  \delta_Q(u,v)$. Finally, suppose that $y$ is in $T(L_3)$. Since the distance between $e'$ and any node in $T(L_3)$ is
 at most $2B(k+1,c)+1$, and $\delta_Q(f,f')=2n-1$, in view of the above observation we have $\delta_Q(x,y) \leq D+3B(k+1,c)+2n+1$. We have $A(\alpha, c) \geq B(k+1, c)$, and, by property 11 for $\cT_k$, we have $D\geq A(\alpha, c)+4$. Hence $ D+3B(k+1,c)+2n+1 \leq 2D+2B(k+1,c)+2n-3$. Hence {$\delta_{Q}(x,y)\leq 2D+2B(k+1,c)+2n-3<\delta_Q(u,v)$}.
 
   \noindent
  Case 6. 
  $x$ is in $L_1$ or $M'$, and $y$ is in $M''$ or in $L_4$.
  
  As noticed in the analysis of Case 5, $\delta_Q(x,c)\leq D-2$. Similarly, $\delta_Q(y,c')\leq D-2$. Moreover, $\delta_Q(d,d')=2B(k+1,c)+2n+1$.
  Hence $\delta_{Q}(x,y)\leq 2D+2B(k+1,c)+2n-1= \delta_Q(u,v)$.
  
  This concludes the proof of properties 4 and 10 for $\cT_{k+1}$.

  We now prove property 9 for $i\leq k$ and $j=k+1$. First suppose that $i= k$ and $j=k+1$. Consider a graph $Q$ from $\cT_{k+1}$ that
   results from  the merge operation of graphs $H'$ and $H''$ from $\cT_k$.
  We keep the notation used in the analysis of properties 4 and 10.
 By construction and in view of Claim \ref{art}, we have the augmented truncated view $\cB^{B(k+1,c)-1}(d)$ in graph $Q$ is equal 
 to the augmented truncated view $\cB^{B(k+1,c)-1}(d)$ in graph $H'$. Similarly, the augmented truncated view $\cB^{B(k+1,c)-1}(d')$ in graph $Q$ is equal 
 to the augmented truncated view $\cB^{B(k+1,c)-1}(d')$ in graph $H''$. By construction, $\delta_Q(u,d)=D-1$ and $\delta_Q(v,d')=D-1$.
 By Claim \ref{art}, the augmented truncated view $\cB^{D+B(k+1,c)-2}(u)$ in $Q$ is equal to the augmented truncated view $\cB^{D+B(k+1,c)-2}(u)$ in $H'$.
 Likewise,  the augmented truncated view $\cB^{D+B(k+1,c)-2}(v)$ in $Q$  is equal to the augmented truncated view $\cB^{D+B(k+1,c)-2}(v)$ in $H''$.
 By definition, $D+B(k+1,c)-2\geq D+A(B(k,c),c)$. This concludes the proof of property 9 when $i= k$ and $j=k+1$. 
 
 Next suppose that $i< k$ and $j=k+1$. By property 9 for indices $i$ and $k$ (holding by the inductive hypothesis), there exist two graphs
 $J'$ and $J''$ in $\cT_i$, such that the augmented truncated view
$\cB^{D'+A(B(i,c),c)}(u)$ in $H'$ is equal to the augmented truncated view $\cB^{D'+A(B(i,c),c)}(w')$ in $J'$, and
the augmented truncated view $\cB^{D'+A(B(i,c),c)}(v)$ in $H''$ is equal to the augmented truncated view $\cB^{D'+A(B(i,c),c)}(w'')$ in $J''$, where $D'$ is the diameter of graphs in $\cT_i$, $w'$ is the left principal node of $J'$,  and
$w''$ is the right principal node of $J''$.

{However, as proven above, the augmented truncated view $\cB^{D+A(B(k,c),c)}(u)$ in $Q$ is equal to the augmented truncated view $\cB^{D+A(B(k,c),c)}(u)$ in $H'$,  and  the augmented truncated view $\cB^{D+A(B(k,c),c)}(v)$ in $Q$  is equal to the augmented truncated view $\cB^{D+A(B(k,c),c)}(v)$ in $H''$.}
 
{Hence, since $D>D'$ (by property 5 for $j=k$) and $A(B(k,c),c)>  A(B(i,c),c)$ for all $i<k$ (by definition), the augmented truncated view
$\cB^{D'+A(B(i,c),c)}(u)$ in $Q$ is equal to the augmented truncated view $\cB^{D'+A(B(i,c),c)}(w')$ in $J'$, and
the augmented truncated view $\cB^{D'+A(B(i,c),c)}(v)$ in $Q$ is equal to the augmented truncated view $\cB^{D'+A(B(i,c),c)}(w'')$ in $J''$ .  As a result, the property also holds when $i< k$ and $j=k+1$. This concludes the proof of property 9.}

The last property to be proved for $\cT_{k+1}$ is property 7. By the inductive assumption, it is enough to prove it for $i\leq k$ and $j=k+1$.
Suppose, by contradiction, that graphs in the sequence $\cT_i$, for some $i\leq k$, receive the same advice as graphs in $\cT_{k+1}$. Let $Q$ be a graph in $\cT_{k+1}$. By property 8, all graphs from $\cT_i$ have election index at most $B(i,c)$. Hence, for any graph $H$ from $\cT_i$, and any node $z$ of $H$, $z$ elects a leader
after time at most $D+A(B(i,c),c)$, where $D$ is the diameter of graphs in $\cT_i$. (Recall that in part 1 of the theorem, election must be performed after time at most  $D+\phi +c$, where $\phi$ is the election index.) The node $z$ must output a sequence of port numbers of length at most $2(m-1)$ {(corresponding to a path of length at most $m-1$)}, where $m$ is the maximum size of a graph from $\cT_i$. By property 9, 
there exist graphs $H' \neq  H''$ from $\cT_i$, such that 
the augmented truncated view
$\cB^{D+A(B(i,c),c)}(u)$ in $Q$ is equal to the augmented truncated view $\cB^{D+A(B(i,c),c)}(u')$ in $H'$, and
the augmented truncated view $\cB^{D+A(B(i,c),c)}(v)$ in $Q$ is equal to the augmented truncated view $\cB^{D+A(B(i,c),c)}(v'')$ in $H''$, where $u$ is the left principal node of $Q$,  
$u'$ is the left principal node of $H'$, $v$ is the right principal node of $Q$, and  $v''$ is the right principal node of $H''$. Hence, in view of our assumption that 
graphs in the sequence $\cT_i$ receive the same advice as graphs in $\cT_{k+1}$, nodes $u$ and $v$ must also output sequences of port numbers 
of length at most $2(m-1)$ after time at most $D+A(B(i,c),c)$. By construction, the distance between $u$ and $v$ in $Q$ is at least $2m-1$. Hence the sequences
of port numbers outputted by $u$ and $v$ must correspond to paths in $Q$ whose other extremities are different. It follows that $u$ and $v$ elect different leaders in $Q$, which gives a contradiction. 

This concludes the inductive proof of all properties 1-13 for $\cT_{k+1}$ and hence concludes the proof that these properties hold for all $\cT_k$, where $k\leq k^*$.
This, in turn, finishes the proof of part 1 of the theorem.

It remains to show how our proof of part 1 has to be changed, in order to obtain proofs of parts 2, 3, and 4. Thanks to our parametrization using functions $A$, $B$ and $R$, the changes are very small. Indeed, it suffices to change the definitions of these functions and all (parametrized) constructions and arguments from the proof of part 1 remain unchanged.
We now give the definitions of functions $A$, $B$ and $R$ for each of parts 2, 3, and 4 separately.

To prove part 2, we define $A(x,c)=cx$, $B(x,c)=(c+2)^x$ and $R(x)=\log x$. Note that the election time is then $D+c\phi$, as assumed in part 2, and the lower bound on the size of advice becomes $\Omega(\log \log \alpha)$, as desired. Indeed, in part 2, the number $k^*$ of different pieces of advice is in $\Omega(\log \alpha)$,
since by definition we have $(c+2)^{k^*} \leq \alpha <(c+2)^{k^*+1}$. 

To prove part 3, we define $A(x,c)=x^c$, $B(x,c)=2^{(c^{3x})-c}$ and $R(x)=\log\log x$. Note that the election time is then $D+\phi ^c$, as assumed in part 3, and the lower bound on the size of advice becomes $\Omega(\log\log \log \alpha)$, as desired. Indeed, in part 3, the number $k^*$ of different pieces of advice is in $\Omega(\log\log \alpha)$,
since by definition we have $2^{c^{3k^*}-c} \leq \alpha <2^{c^{3(k^*+1)}-c}$.  

To prove part 4,  recall the notation $^ic$, defined by induction as follows: $^0c=1$ and $^{i+1}c=c^{^ic}$. 
We define $A(x,c)=c^x$, $B(x,c)= {^{2x}c}$,  and $R(x)=\log^* x$.
Note that the election time is then $D+c^{\phi}$, as assumed in part 4, and the lower bound on the size of advice becomes $\Omega(\log( \log^* \alpha))$, as desired.
Indeed, in part 4, the number $k^*$ of different pieces of advice is in $\Omega(\log^* \alpha)$,
since by definition we have $^{2k^*}c \leq \alpha <^{2(k^*+1)}c$.

\end{proof} 

We close this section by showing that constant advice is not enough for leader election in all feasible graphs, regardless of the allocated time.

\begin{proposition}\label{constant}
There is no algorithm using advice of constant size and performing leader election in all feasible graphs.
\end{proposition}

\begin{proof}
We define a family $\cH$ of graphs, called {\em hairy rings}, for which we will prove that no algorithm with advice of constant size performs correct leader election for all graphs in $\cH$. Let $R_n$ be the ring of size $n\geq 3$, with port numbers 0,1 at each node, in clockwise order. Let $S_k$, for any integer $k\geq 2$, be the $(k+1)$-node tree with $k$ leaves, called the $k$-star. The only node of degree larger than
1 of the $k$-star is called its {\em central node}. For $k=1$, $S_k$ is defined as the two-node graph with  the central node designated arbitrarily, and for $k=0$, $S_k$ is defined as the  one-node graph,
with the unique node being its central node. {The class $\cH$ is the set of all graphs that can be obtained in the following way.}  
{For all $n\geq3$}, attach to every node {$v$} of every ring $R_n$ some graph $S_k$, for $k \geq 0$, by identifying its central node
with the node $v$, in such a way that, for every ring, the star of maximum size attached to it is unique. Assign missing port numbers in any legal way, i.e., so that  port numbers at a node of degree $d$ are from 0 to $d-1$. 
Every graph obtained in this way is feasible because it has a unique node of maximum degree. {An example of a hairy ring is depicted in Fig.~\ref{fig:fig9}a.}

\begin{figure}[!htbp]
\begin{center}
  \begin{minipage}[t]{0.5\linewidth}
    \centering
	\includegraphics[width=0.4\textwidth]{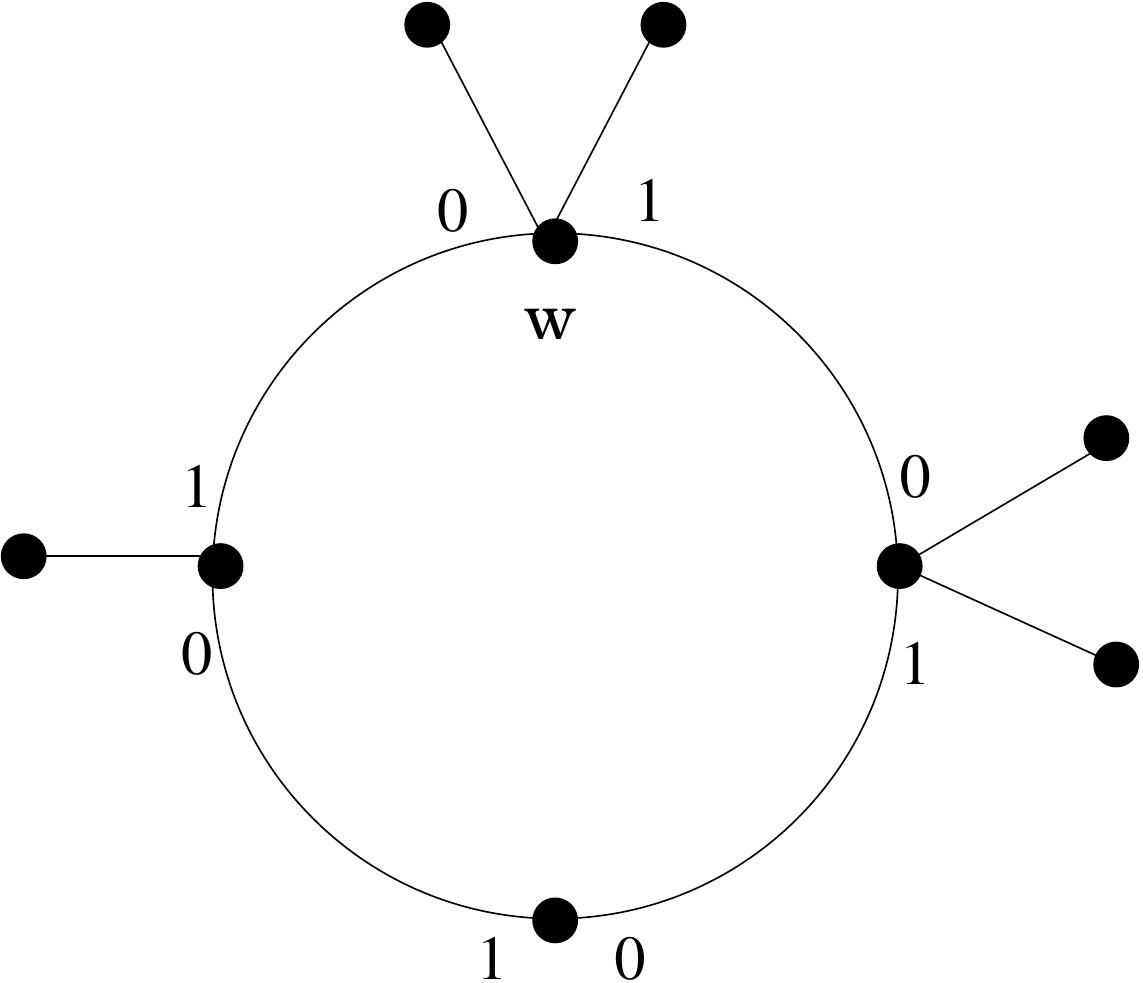}\\
    {\footnotesize ($a$) a hairy ring $H$.}
  \end{minipage}%
  \begin{minipage}[t]{0.5\linewidth}
    \centering
	\includegraphics[width=0.6\textwidth]{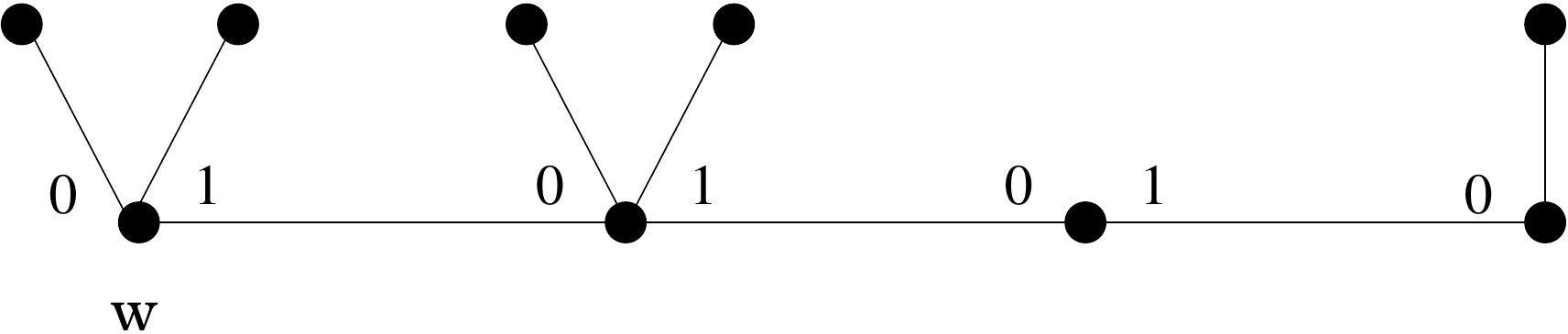}\\
    {\footnotesize ($b$) the cut of $H$ at node $w$.} 
  \end{minipage}
  \bigskip \\
  \begin{minipage}[t]{1\linewidth}
    \centering
	\includegraphics[width=0.6\textwidth]{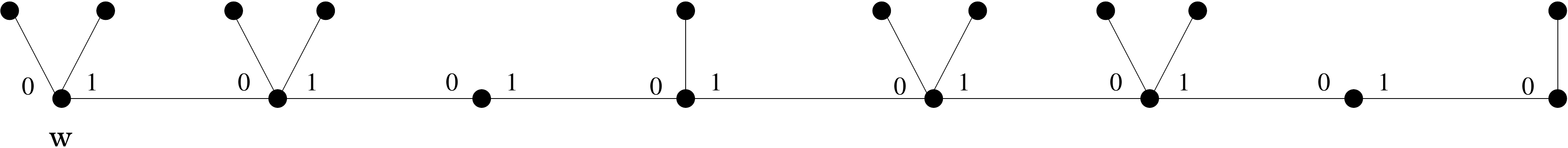}\\
    {\footnotesize ($c$) the $2$-stretch of $H$ at node $w$.}
  \end{minipage}
\end{center}
 \caption{Illustration of an hairy ring and its different transformations used in the proof of Proposition~\ref{constant}.\label{fig:fig9}}
\end{figure}

For any graph $H$ in $\cH$ we define a {\em cut} of $H$ as follows. Let $H$ be a graph resulting from a ring $R_n$ by attaching stars. Fix any node $w=w_1$ of this ring.
Let $w_1,\dots, w_n$ be nodes of this ring listed in clockwise order.
The cut of $H$ at node $w$ is the graph resulting from $H$ by removing the edge $\{w_1,w_n\}$. Node $w=w_1$ is called the first node of the cut and node $w_n$ is called the last node of the cut.
For any integer $\gamma \geq 2$, the $\gamma$-{\em stretch} of $H$ starting at node $w$  is the graph defined as follows. Take $\gamma$ pairwise disjoint isomorphic copies of the cut of $H$ at $w$.
For $1<i \leq \gamma$, attach the $i$th copy to the $(i-1)$th copy  joining the first node $a_i$ of the $i$th copy with the last node $b_{i-1}$ of the $(i-1)$th copy
by an edge with port 0 at $a_i$ and port 1 at $b_{i-1}$ . The first node of the first copy is called the first node of the $\gamma$-stretch, and the last node of the last copy is called the last node of the $\gamma$-stretch.

Suppose that there exists a leader election algorithm $\cA$ which uses advice of constant size to perform leader election in all hairy rings from the family $\cH$. Let $c$ be the smallest integer such  that a total of $c$ pieces of advice are sufficient to elect a leader in every graph from $\cH$ by algorithm $\cA$. Let $H_1,\dots,H_c$ be graphs from $\cH$ for which algorithm $\cA$ uses different pieces of advice. Let $N$ be the maximum of sizes of all graphs $H_1,\dots,H_c$, and let $T$ be the maximum execution time of $\cA$, for all graphs $H_1,\dots,H_c$. 

Let $\gamma=4(N+T)$ and let $G_j$ be the $\gamma$-stretch of $H_j$ starting at some node $u_j$ of $H_j$, for $j \leq c$. We define the graph $G$ as follows. Take pairwise disjoint isomorphic copies of graphs $G_j$, for $j \leq c$.
{For every $1<j\leq c$,} attach $G_j$ to $G_{j-1}$  joining the first node $c_i$ of the $i$th copy with the last node $d_{i-1}$ of the $(i-1)$th copy
by an edge with port 0 at $c_i$ and port 1 at $d_{i-1}$. 
Finally, take a $\gamma$-star
and join its central node by edges to the first node of $G_1$ and to the last node of $G_{\gamma}$, assigning missing port numbers in any legal way. The graph $G$ obtained in this way is in $\cH$ because
it has a unique node of maximum degree which is $\gamma+2$. {Let $n_{H_j}$ be the size of the ring that was used in the construction of $H_j$}. Let $a_j$ be the (unique) node of $G_j$ at distance {$n_{H_j}(N+T)$} from $u_j$, at the end of a simple path
all of whose ports are 0's and 1's. Let $b_j$ be the (unique) node of $G_j$ at distance {$3n_{H_j}(N+T)$} from $u_j$, at the end of a simple path
all of whose ports are 0's and 1's. Call these nodes the {\em foci} of $G_j$. Each of them corresponds to the first node of the cut serving to define $G_j$. Let $z_j$ be the 
node in $H_j$ at which this cut was done.

By definition of graphs $H_1,\dots,H_c$, the advice received by graph $G$ when algorithm $\cA$ is performed, is the same as that received by some graph $H_{j_0}$.
In $H_{j_0}$ the node $z_{j_0}$ executing algorithm $\cA$ must stop after time at most $T$. 
By construction, the augmented truncated view $\cB^T(z_{j_0})$ in $H_{j_0}$ is the same as
the augmented truncated views $\cB^T(a_{j_0})$ and $\cB^T(b_{j_0})$  in $G$. Hence nodes $a_{j_0}$ and $b_{j_0}$ executing  algorithm $\cA$
in $G$ must also stop after time at most $T$. Node $z_{j_0}$ in $H_{j_0}$ must output a sequence of port numbers of  length smaller than $2N$ because the size
of $H_{j_0}$ is at most $N$. Hence nodes $a_{j_0}$ and $b_{j_0}$ executing  algorithm $\cA$ in $G$ must also output a sequence of port numbers of  length smaller than $2N$, corresponding to simple paths in $G$ of length smaller than $N$, starting, respectively at nodes $a_{j_0}$ and $b_{j_0}$. However, the distance between  $a_{j_0}$ and $b_{j_0}$ in $G$ is at least $2N$, hence the other extremities of these simple paths must be different. It follows that the leaders elected by nodes
$a_{j_0}$ and $b_{j_0}$ executing algorithm $\cA$ in $G$ are different, and hence this algorithm is not correct for the class $\cH$.
\end{proof}

\section{Conclusion}

We established almost tight bounds on the minimum size of advice sufficient for election in minimum possible time (i.e., in time equal to the election index $\phi$) and tight bounds on this size for several large values of time.
The first big jump occurs between time $\phi$ and time $D+\phi$, where $D$ is the diameter of the graph. In the first case, the size of advice is (roughly)  linear in
the size $n$ of the graph, and in the second case it is at most logarithmic in $n$, in view of Proposition \ref{H}
and of the remark after Theorem \ref{theorem-quadruple}. The intriguing open question left by our results is how the minimum size of advice behaves in the range of election time strictly between $\phi$ and $D+\phi$, i.e., for time sufficiently large to elect if the map were known, but possibly too small for all nodes
to see the augmented truncated views at depth $\phi$ of all other nodes, and hence to realize all the differences in views. Note that, for time exactly $D+\phi$,
all nodes see all these differences, although, without any advice, they cannot realize that they see all of them: this is why some advice is needed
for time $D+\phi$. 

\pagebreak
 
\bibliographystyle{plain}


\end{document}